\documentclass[12pt,a4paper]{article}
\usepackage{upgreek}
\usepackage[mathscr]{euscript}
\usepackage{environ}
\usepackage{amsfonts} 
\usepackage{amsthm,amsmath,amssymb}

\usepackage{mathrsfs}
\usepackage{graphicx,psfrag,color}
\usepackage{enumerate}

\newtheorem{theorem}{\bf Theorem}
\newtheorem{claim}{\bf Claim}
\newtheorem{remark}[theorem]{\bf Remark}

\usepackage{xcolor}

\usepackage[firstinits=true, parentracker=true,
doi=true, isbn=false, url=true, natbib=true,
eprint=true,
sorting=none,
]{biblatex}

\usepackage[colorlinks=true,linkcolor=blue, citecolor=brown, urlcolor=purple]{hyperref}
\usepackage[left]{lineno}
\usepackage{lipsum} 
\setpagewiselinenumbers
\usepackage{enumerate}
\numberwithin{equation}{section}

\newcommand{\thicktilde}[1]{\mathbf{\tilde{\text{$#1$}}}}

\usepackage{MnSymbol}

\newcommand{\la}{\mathrm{b}}
\DeclareMathSymbol{\gimel}{\mathord}{MnSyC}{"B1}
\newcommand{\od}[2]{\frac{d#1}{d#2}}
\newcommand{\su}{\mathtt{u}}

\DeclareMathOperator{\saux}{\mathcal{W}}
\newcommand{\sgnM}{\gimel}
\newcommand{\poly}[1]{\mathcal{#1}}
\newcommand{\recip}[1]{{#1}^\dagger}
\newcommand{\old}{{\mathrm{inc}}}
\newcommand{\redc}{\Lambda}
\newcommand{\pertbn}[1]{\ensuremath{\breve{#1}}}

\hoffset=0pt\voffset=10pt\textwidth=500pt\textheight=730pt\topmargin=-80pt\oddsidemargin=-20pt\evensidemargin=-20pt\let\oldmarginpar\marginpar\renewcommand\marginpar[1]{\-\oldmarginpar[\raggedleft\footnotesize #1]{\raggedright\footnotesize #1}}

\begin{document}

\title{Discrete scattering by two staggered semi-infinite defects: reduction of matrix Wiener--Hopf problem\thanks{Revised submission to `Journal of Engineering Mathematics'}}
\author{Basant Lal Sharma\thanks{Department of Mechanical Engineering, Indian Institute of Technology Kanpur, Kanpur, U. P. 208016, India ({bls@iitk.ac.in}).}}
\date{}
\maketitle

\begin{abstract}
{As an extension of the discrete Sommerfeld problems on lattices, the scattering of a time harmonic 
wave 
is considered {on an infinite square lattice} when there exists a pair of semi-infinite cracks or rigid constraints. Due to {the} presence of stagger, also called offset, in the alignment of the defect edges the asymmetry in the problem leads to a matrix Wiener--Hopf kernel that cannot be reduced to scalar Wiener--Hopf in {any known} way. In the corresponding continuum model the same problem is {a} well known {formidable one which} possesses certain special structure with exponentially growing elements on {the} diagonal of kernel. {From this viewpoint} the present paper tackles a discrete analogue of the same by reformulating the Wiener--Hopf problem and reducing it to a {finite} set of linear algebraic equations; the coefficients of which can be found by an application of {the} scalar Wiener--Hopf factorization.
The {considered} discrete paradigm involving lattice waves is relevant for modern applications of mechanics and physics
at small length scales.}
\end{abstract}


\section*{Introduction}
The Wiener--Hopf 
technique \cite{Wiener,Paley} has many applications in understanding singular phenomena in mechanics and physics \cite{Mikhlin,Danielebook}. 
The scattering of waves in electro-magnetism, acoustics, and allied subjects \cite{williams_1954,jones_1952,jull1973aperture,johansen1965radiation,crease1958propagation,james1979double,kapoulitsas1984propagation,michaeli1985new,michaeli1996asymptotic} is one such wherein
the presence of sharp edges and {an} assortment of mixed boundary conditions 
allows an application of the method conceived by Wiener and Hopf.
A typical wave diffraction phenomenon 
that has been an interesting problem for researchers \cite{Heins1,Heins2,cheney1951diffraction,jones1973double,Jones3planes,Jones1,meister1996factorization,thompson2005mode,Abrahams0,Abrahams2,Abrahams4} involves an incident time harmonic wave on more than one semi-infinite parallel rows with either Neumann or Dirichlet condition.
Apart from some special cases \cite{Heins, Levine, Carlson, Levine1, Levine2,Heinslim,MeisterRottbrand,Meistersys1,Meistersys2,daniele1984solution,Khrapkov2} such multiple diffraction problems often give rise to the matrix {{Wiener--Hopf}} kernels for which there is, as yet, no general constructive method of factorization \cite{KreinGoh,rogosin2015constructive}. 
One such canonical problem is the determination of the sound field scattered by two semi-infinite parallel plates whose edges are not aligned. 
The asymmetry in this physical problem leads to the occurrence of certain exponential phase factors in the {{Wiener--Hopf}} kernel, for instance, the 
main term has form \cite{Abrahams1},
\begin{equation}\begin{split}
\begin{bmatrix}
1&{e}^{-h\gamma}e^{i\xi a}\\
{e}^{-h\gamma}e^{-i\xi a}&1
\end{bmatrix},
\label{Abrakernel}
\end{split}\end{equation}
for $\xi$ belonging to a {(infinite)} strip surrounding the real line in the complex plane, where $a$ is the `horizontal' offset between the edges along the plate direction, while $h$ is the vertical spacing between the plates, and $\gamma{\,:=}\gamma(\xi)=\sqrt{\xi^2-k^2}$ with $k$ as incident wave number.
{A} few decades ago, within the {{Wiener--Hopf}} formulation \cite{Noble} for finding the scattered velocity potential for this problem, a method was announced \cite{Abrahams1} that successfully reduced \eqref{Abrakernel} to solving a complex linear functional {{Wiener--Hopf}} equation. As a generalization, a method for factorizing such general class of matrix kernels, with exponential phase factors, has been given in \cite{AbrahamsExpo}.
{The subtlety behind such Wiener--Hopf kernel factorization has been investigated in several accounts as well, for example, see \cite{GohbergKaashoek}.}

The present paper is a discrete analogue of the work on diffraction by parallel staggered plates \cite{Abrahams1} {(see \cite{sharma2019wienerhopf} for a catalogue of many discrete scattering problems)}.
{Within the discrete scattering theory, it has been recently shown that} in different types of lattice models, certain mechanical analogues of soft or hard screens, namely, rigid constraints or cracks \cite{Bls0,Bls1,Bls2,Bls3,Bls4,Bls5,Bls6,Bls9s,Bls8arrayfinite}, respectively, {as well as steps on lattice surfaces \cite{Bls10mixed,Bls9s,Victor_Bls_surf2}} can be analyzed. 
Indeed, the problem of scattering by defects in arbitrary lattices has rich history \cite{Lifshitz, maradudin,Maradudinbook}.
{Modern applications of the mechanical models are also relevant to physics at small length scales as exemplified by the interest in transport across channels involving phononic \cite{Bls5k_tube} and electronic signal \cite{Bls5c_tube,Bls5ek_tube,Bls5c_tube_media}. From the viewpoint of specific geometry of staggered edges under consideration,}
it has been found \cite{GMthesis,Bls8staggerpair_asymp} that the discrete scattering problem involving a pair of staggered cracks or rigid constraints also involves a factor in the {{Wiener--Hopf}} matrix kernel, a counterpart of \eqref{Abrakernel} belonging to a formidable class of kernels \cite{rogosin2015constructive}, of the form
\begin{equation}\begin{split}
\begin{bmatrix}
1&{{\lambda}}^{{\mathtt{N}}}{{{z}}^{-{{\mathtt{M}}}}}\\
{{\lambda}}^{{\mathtt{N}}}{{{z}}^{{\mathtt{M}}}}&1
\end{bmatrix}
\label{myAbrakernel}
\end{split}\end{equation}
for ${z}$ belonging to an annulus surrounding the unit circle in the complex plane.
\begin{figure}[h!]
\centering
{(a)}\includegraphics[width=.5\textwidth]{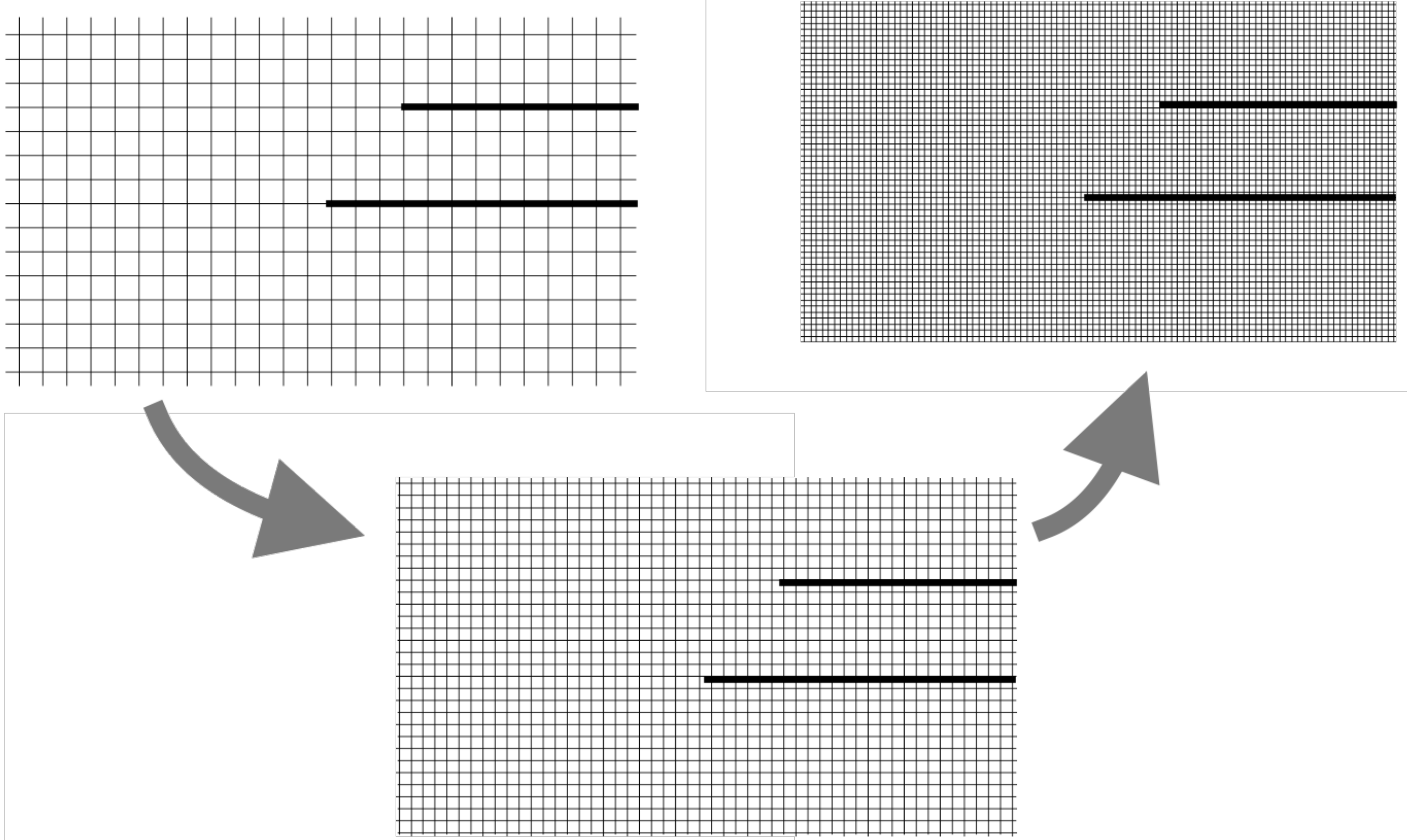}{(b)}\includegraphics[width=.4\textwidth]{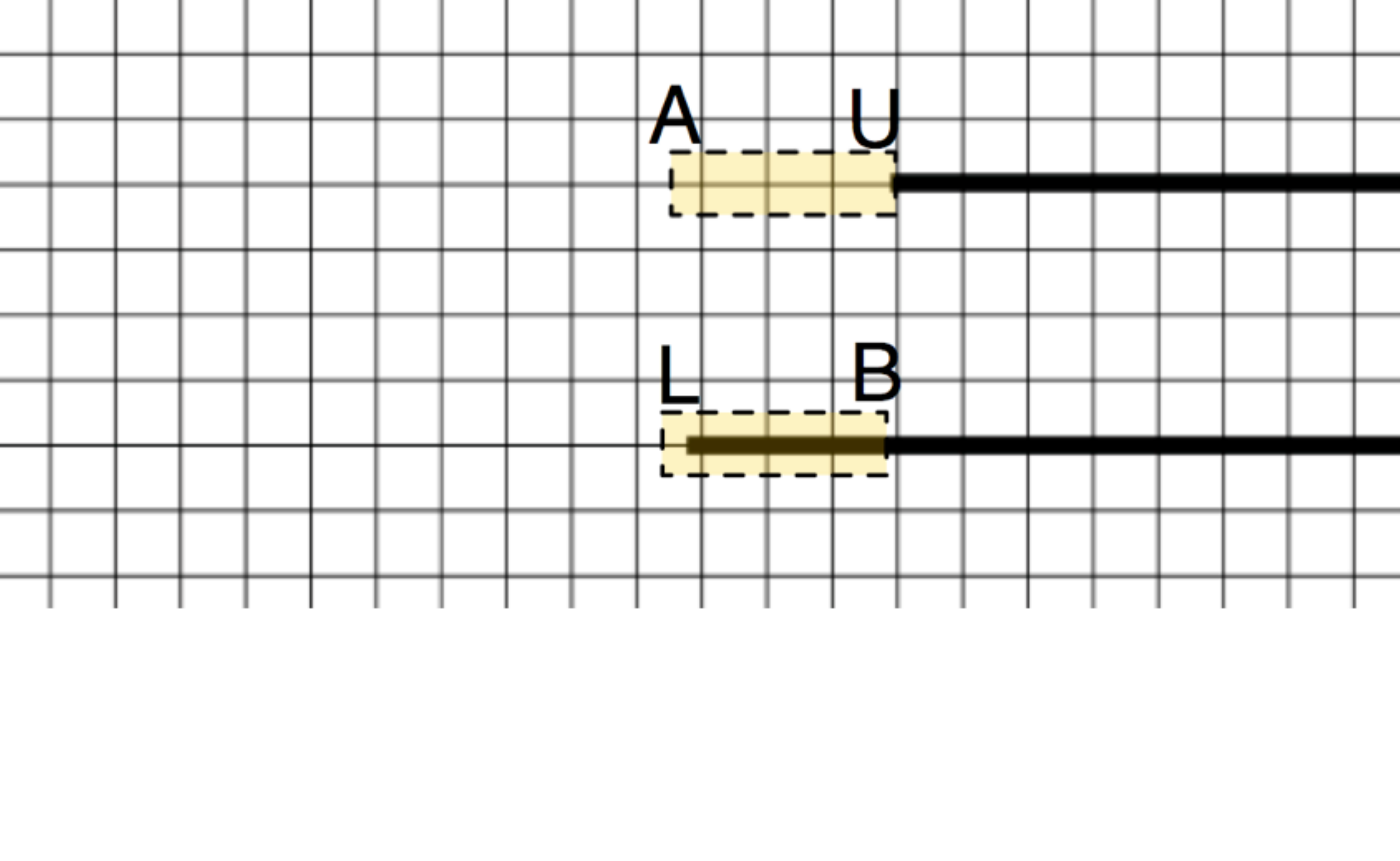}
\caption{{(a) Schematic of the lattice model with a shrinking nearest neighbour distance while keeping the vertical separation between edges and offset between tips as same. (b) Schematic of the two portions of near the staggered edges where the relevant wave field is evaluated.}}
\label{scalingschematic_2}
\end{figure}
In equation \eqref{myAbrakernel}, ${\lambda}{\,:=}{\lambda}({z})$ is the discrete analogue of $\gamma$ and ${z}$ that of $\xi$ \cite{Bls31}. 
In fact, as the ratio of square lattice grid spacing to the incident wavelength, i.e., $\la\times k=:{\upepsilon}\to0$, \eqref{Abrakernel} is recovered with $a={\mathtt{M}}{\upepsilon}, h={\mathtt{N}}{\upepsilon}, {\lambda}=e^{-{\upepsilon} \gamma}, {z}=e^{-i{\xi}{{\upepsilon}}}$ {(see Fig. \ref{scalingschematic_2}(a) for the lattice structure variation as the value of ${\upepsilon}$ is made smaller by halving it in each stage)}. 
In \cite{Bls8staggerpair_asymp}, {which is based on some results obtained in \cite{GMthesis}},  an asymptotic method \cite{Mishuris2014} was applied to factorize the kernel \eqref{myAbrakernel}. 
The primary precursor to the success of {such} a method 
appears to be the possibility of an exact solution for the zero offset case \cite{Bls8pair1}.
In this paper, following this line of reasoning, it is shown that the {{Wiener--Hopf}} problem can be reformulated so that eventually only a system of linear algebraic equation needs to be solved whose coefficients can be obtained using scalar {{Wiener--Hopf}} method \cite{Noble}. This is reminiscent of the distinguished work of \cite{Abrahams1,AbrahamsExpo}. 

{As the last point in this introduction, but not the least, it is noted that}
even though the negative offset {case of discrete scattering due to the staggered defects is physically equivalent to the case with} the positive offset, for the purpose of completeness both cases are studied in the paper. {The two cases of offset signs can be mapped into each other by flipping the structure. Moreover, after solving the reduced algebraic equations for the case of offsets with same magnitude, but opposite sign, paves an alternate way of simultaneous evaluation of certain relevant entity derived from the wave field on both segments AU and LB of the lower edge and upper edge, respectively, depicted in Fig. \ref{scalingschematic_2}(b).}

{The paper is organized as follows. In \S\ref{sqLattFT}, the lattice structure is described and the scattering problem is posed for both kinds of staggered edges. A general solution of the scattered field is also stated for the portion between the edges as well as above and below them in terms of a minimal set of unknown functions.
In \S\ref{WHeqformK}, the scattering problem associated with a pair of staggered cracks is analyzed within a Wiener--Hopf formulation and the final set of finite number of linear algebraic equations is derived based on scalar Wiener--Hopf factors of certain characteristic functions.
In \S\ref{WHeqformC}, the scattering problem arising due to rigid constraints is attended by the same method.
\S\ref{briefnumerics} provides some graphical results that enable a comparison of computations based on analytical method vis-a-vis direct numerical method, and also
presents a discussion of two special aspects of the analysis and resulting calculations; this is followed by the concluding remarks. The paper also includes auxiliary calculations, definitions and derivations in five appendices.}

{\bf Notation:}
Let $\mathbb{Z}$ stand for {the} set of integers. {Let} ${\mathbb{Z}^+}$ {denote} {the} non-negative integers and ${\mathbb{Z}^-}$ {denote} {the} negative integers, {i.e.,
\begin{equation}
{\mathbb{Z}^+}=\{0, 1, 2, \dotsc\},\quad
{\mathbb{Z}^-}=\{-1, -2, \dotsc\}.
\label{ZpZn}
\end{equation}}
Let {a discrete interval be denoted by}
{\begin{equation}
\mathbb{Z}_{a}^{b}=\{a, a+1, \dotsc, b\}\subset\mathbb{Z}.
\label{Zab}
\end{equation}}
{In this paper, it is supposed that the} letter ${{\mathcal{H}}}$ stands for the discrete Heaviside function: 
{\begin{equation}
{{\mathcal{H}}}({{\mathtt{x}}})=0, {{\mathtt{x}}}<0\text{ and }{{\mathcal{H}}}({{\mathtt{x}}})=1, {\mathtt{x}}\ge0.
\label{Hstep}
\end{equation}}
The discrete Fourier transform, simply addressed as Fourier transform, of a sequence $\{{\su}_m\}_{m\in\mathbb{Z}}$ is denoted by ${\su}^{{\mathrm{F}}}$ \cite{Bls0}; for instance, for the field $\{{\su}_{{\mathtt{x}}, {\mathtt{y}}}\}_{{\mathtt{x}}\in\mathbb{Z}}$ at given ${\mathtt{y}}\in\mathbb{Z}$, the Fourier transform is defined by \cite{jury, Silbermann}
\begin{equation}\begin{split}
{\su}_{{\mathtt{y}}}^{{\mathrm{F}}}{\,:=}{\su}_{{\mathtt{y}}}^{-}+{\su}_{{\mathtt{y}}}^{+}, \quad {\su}_{{\mathtt{y}}}^{\pm}=\sum\limits_{{{\mathtt{x}}}\in\mathbb{Z}}{{z}}^{-{{\mathtt{x}}}}{{\mathcal{H}}}(\pm{{\mathtt{x}}}-{\frac{1}{2}}\pm{\frac{1}{2}}){\su}_{{{\mathtt{x}}}, {\mathtt{y}}}.
\label{discreteFT}\end{split}\end{equation}
In general, the {decoration} $-$ (resp. $+$) is associated with a complex function which is analytic inside (resp. outside) and on an annulus {${{\mathscr{A}}}$ in the complex plane (shown in the schematic of Fig. \ref{schematicannulus_2})}.
\begin{figure}[h!]
\centering
\includegraphics[width=.6\textwidth]{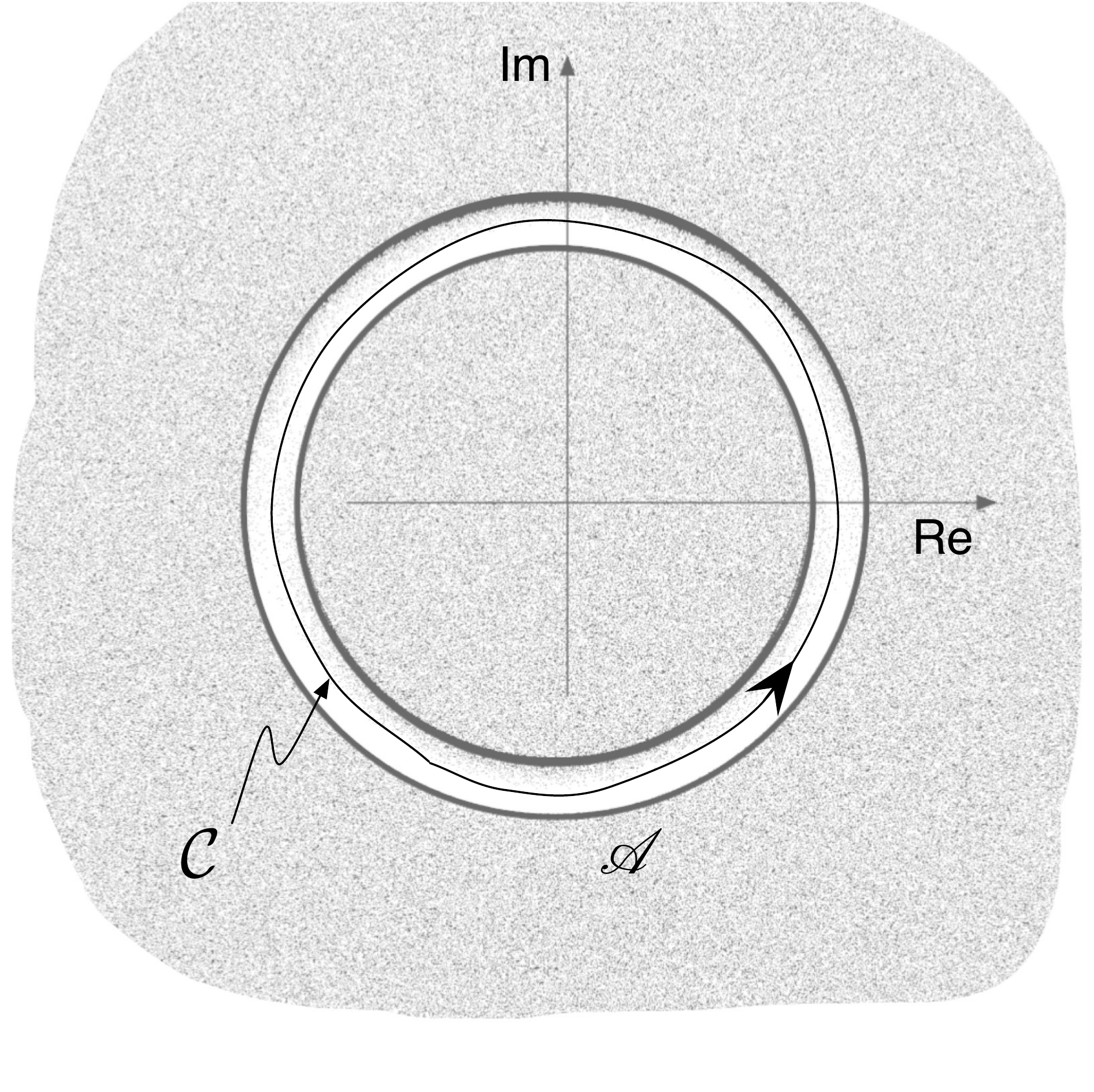}
\caption{{Schematic of the annulus ${{\mathscr{A}}}$ in the complex plane and a contour ${\mathcal{C}}$ for the inverse Fourier transform.}}
\label{schematicannulus_2}
\end{figure}
{Eventually, the inverse Fourier transform leads to the following expression of ${\su}_{{\mathtt{x}},{\mathtt{y}}}$ in terms of ${\su}_{{\mathtt{y}}}^{{\mathrm{F}}}$,
\begin{equation}\begin{split}
{\su}_{{\mathtt{x}},{\mathtt{y}}}=\frac{1}{2\pi i}\oint_{{\mathcal{C}}} {\su}_{{\mathtt{y}}}^{{\mathrm{F}}}({z}){{z}}^{{\mathtt{x}}-1}d{z}, 
 \quad\quad {\mathtt{x}}\in\mathbb{Z}, 
\label{uFTinv}
\end{split}\end{equation}
where ${\mathcal{C}}$ is a counter-clockwise contour in the annulus ${\mathscr{A}}$ where both ${\su}_{{\mathtt{y}}}^{+}$ and ${\su}_{{\mathtt{y}}}^{-}$ are analytic.
{\em In the paper, the {Wiener--Hopf} formulation is posed on the annulus ${\mathscr{A}}$, while this fact is not emphasized repeatedly for brevity, and all of the entities and terms appearing in {Wiener--Hopf} equation are analytic on ${\mathscr{A}}$ and the {Wiener--Hopf} kernel is regular (non-vanishing) too on ${\mathscr{A}}$.}}
{The additive {Wiener--Hopf} factors are denoted by {\em super}script $\pm$ while multiplicative ones by {\em sub}script $\pm$, i.e., for a suitable function $f$, 
$$
f=f^-+f^+,\quad f=f_+f_-.
$$}
The symbol ${{z}}$ is exclusively used throughout as a complex variable for the Fourier transform. 
{However, in order to avoid a cluttering of symbols, at several places {\em the argument $z$ of some relevant complex functions has been suppressed while in the same equation it appears for some other function; the dependence on $z$ is clear from the context, however}. In this paper, following the traditional choice,}
the square root function, $\sqrt{\cdot}$, has the usual branch cut in the complex plane running from $-\infty$ to $0$ {on the real axis}. 
To avoid cumbersome notation {and a host of supplementary statements}, whenever a series is provided in the paper, it is assumed that it describes an analytic function, i.e., it converges in the specified region of the complex plane {(which usually often happens to include the fixed annulus ${\mathscr{A}}$ in the complex plane)}.

\section{Square lattice model}
\label{sqLattFT}
Let
$\ensuremath{\hat{\mathbf{e}}}_1$, $\ensuremath{\hat{\mathbf{e}}}_2$ be the {standard} unit basis vectors in $\mathbb{R}^2$ {($\mathbb{R}^2$ can be considered as the plane corresponding to $x_3=0$ in three dimensional physical space described by $\mathbb{R}^3$)}.
Consider an infinite square lattice, denoted by ${\mathfrak{S}}$, of identical particles,
$$
{\mathfrak{S}}\simeq\mathbb{Z}^2=\{({\mathtt{x}},{\mathtt{y}})| {\mathtt{x}}\in\mathbb{Z}, {\mathtt{y}}\in\mathbb{Z}\}.
$$
{From a visualization viewpoint, above is understood according to the notation, $x{\hat{\mathbf{e}}}_1+y{\hat{\mathbf{e}}}_2={\mathtt{x}}\la {\hat{\mathbf{e}}}_1+{\mathtt{y}}\la {\hat{\mathbf{e}}}_2$ for the position vector of a typical point that belongs to ${\mathfrak{S}}$, here ${\mathtt{x}}\in\mathbb{Z}, {\mathtt{y}}\in\mathbb{Z}$. In the rest of the paper, the notation $({\mathtt{x}}, {\mathtt{y}})$ will be used to identify the corresponding site in ${\mathfrak{S}}$ in accordance with this description.}
The out-of-plane displacement of a particle in ${\mathfrak{S}}$, indexed by its {\em coordinates} $({\mathtt{x}}, {\mathtt{y}})\in{{\mathbb{Z}^2}}$, is described by ${\su}_{{\mathtt{x}}, {\mathtt{y}}}\in\mathbb{C}$. 
{In vectorial notation, the displacement is ${\su}_{{\mathtt{x}}, {\mathtt{y}}}{\hat{\mathbf{e}}}_3$, (where ${\hat{\mathbf{e}}}_3$ is the unit vector orthogonal to ${\hat{\mathbf{e}}}_1$ and ${\hat{\mathbf{e}}}_2$) in physical space $\mathbb{R}^3$, for the particle located at the site of ${\mathfrak{S}}$ with the position vector $x{\hat{\mathbf{e}}}_1+y{\hat{\mathbf{e}}}_2+0{\hat{\mathbf{e}}}_3$. However, in the rest of the paper, the three dimensional nature of the physical problem, or its visualization, does not play any role in the mathematical analysis and therefore such remarks will not appear henceforth (more so because there are other strictly two dimensional geometries which entertain problems of similar nature \cite{Bls5c_tube,Bls5ek_tube,Bls5c_tube_media})}.
{On other other hand, the more important part of the `physical' assumption is that} each particle in ${\mathfrak{S}}$ interacts with {(atmost)} its four nearest neighbors, separated by an in-plane spacing $\la$ (see Fig. \ref{Fig1}), through `linearly elastic' identical (massless) bonds. 

\begin{figure}[h!]
\centering
\includegraphics[width=\textwidth]{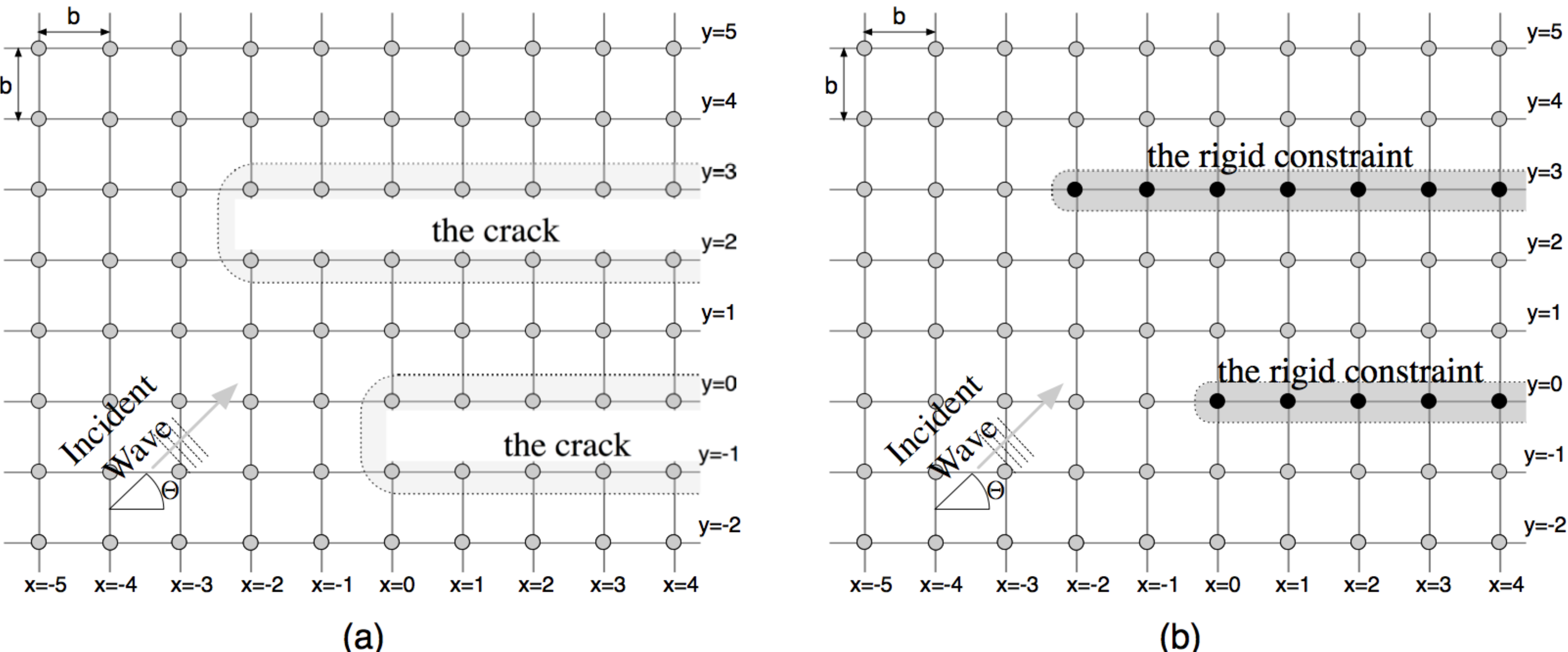}
\caption{(a) Square lattice ${\mathfrak{S}}$ with one semi-infinite crack between ${\mathtt{y}}=0$ and ${\mathtt{y}}=-1$, i.e., broken bonds at all ${\mathtt{x}}\ge0$, and another such crack above it between ${\mathtt{y}}={{\mathtt{N}}}=3$ and ${\mathtt{y}}={{\mathtt{N}}}-1=2$ with broken bonds at all ${\mathtt{x}}\ge {{\mathtt{M}}}=-2$.
(b) Square lattice ${\mathfrak{S}}$ with one semi-infinite rigid constraint at ${{\mathtt{y}}}=0$ with zero displacement at all ${{\mathtt{x}}}\ge0$ and another such constraint above it at ${{\mathtt{y}}}={{\mathtt{N}}}=3$ and all ${{\mathtt{x}}}\ge {\mathtt{M}}=-2$.}
\label{Fig1}
\end{figure}

{From the viewpoint of discrete scattering on the assumed square lattice model \cite{Martin}, the two types of scatterer geometries considered in this paper are depicted in Fig. \ref{Fig1}.} 
The cracks exist {between} two rows of particles as schematically shown in Fig. \ref{Fig1}(a), where the specific crack faces occur at ${\mathtt{y}}=0, -1$ and ${\mathtt{y}}={\mathtt{N}}, {\mathtt{N}}-1$. 
Also in another problem, semi-infinite rigid constraints are assumed to exist at two rows of particles as schematically shown in Fig. \ref{Fig1}(b), where the specific rows are indexed by ${\mathtt{y}}=0$ and ${\mathtt{y}}={\mathtt{N}}$. 
The semi-infinite crack is modeled by assuming zero spring constant between the crack faces \cite{Bls0}
while the rigid constraint is modeled by assuming zero total displacement at constrained sites \cite{Bls1}. 
Let 
\begin{equation}\begin{split}
{{\Sigma}}_{{k}}=\{({\mathtt{x}}, {\mathtt{y}})\in{{\mathbb{Z}^2}}: {\mathtt{x}}\ge0, {\mathtt{y}}=0, -1\}\cup\{({\mathtt{x}}, {\mathtt{y}})\in{{\mathbb{Z}^2}}: {\mathtt{x}}\ge{\mathtt{M}}, {\mathtt{y}}={\mathtt{N}}, {\mathtt{N}}-1\},
\label{defectK}
\end{split}\end{equation}
shown as gray dots in Fig. \ref{Fig1}(a).
Above can be interpreted as the union of both, upper and lower, crack faces. 
Let ${{\Sigma}}_{{c}}$ denote the set of all lattice sites in ${\mathfrak{S}}$ that are rigidly constrained, i.e.,
\begin{equation}\begin{split}
{{\Sigma}}_{{c}}=\{({\mathtt{x}}, 0)\in{{\mathbb{Z}^2}}: {\mathtt{x}}\ge0\}\cup\{({\mathtt{x}}, {\mathtt{N}})\in{{\mathbb{Z}^2}}: {\mathtt{x}}\ge{\mathtt{M}}\},
\label{defectC}
\end{split}\end{equation}
shown as black dots in Fig. \ref{Fig1}(b).
{Following \cite{Bls0,Bls1}}, henceforth, the two dimensional lattice ${\mathfrak{S}}$ is considered, with each particle of unit mass, and, an interaction with only its four nearest neighbours through bonds with a spring constant ${1}/{\la^2}$. 
On the square lattice model ${\mathfrak{S}}$ described thus far, a time harmonic lattice wave is considered incident (see Fig. \ref{Fig1}) and its diffraction by two cracks or rigid constraints is studied.

The equation of motion of particles in the lattice ${\mathfrak{S}}$, while excluding the perturbed sites ${{\Sigma}}$, i.e., constrained set ${{\Sigma}}_{{c}}$ or crack ${{\Sigma}}_{{k}}$, and suppressing an explicit dependence of {${\su}_{{\mathtt{x}}, {\mathtt{y}}}$} on time $t$, is
\begin{equation}\begin{split}
\od{^2}{t^2}{\su}_{{\mathtt{x}}, {\mathtt{y}}}&=\frac{1}{\la^2}{\triangle}{\su}_{{\mathtt{x}}, {\mathtt{y}}}, 
\quad\quad({\mathtt{x}}, {\mathtt{y}})\in{{\mathbb{Z}^2}}\setminus{\Sigma},\\
\text{where }
{\triangle}{\su}_{{\mathtt{x}}, {\mathtt{y}}}&{\,:=}{\su}_{{\mathtt{x}}+1, {\mathtt{y}}}+{\su}_{{\mathtt{x}}-1, {\mathtt{y}}}+{\su}_{{\mathtt{x}}, {\mathtt{y}}+1}+{\su}_{{\mathtt{x}}, {\mathtt{y}}-1}-4{\su}_{{\mathtt{x}}, {\mathtt{y}}}, \quad({\mathtt{x}}, {\mathtt{y}})\in{{\mathbb{Z}^2}}.
\label{dimnewtoneq}
\end{split}\end{equation}

Suppose ${\su}^{{\mathrm{inc}}}$ describes the {\em incident lattice wave} with frequency ${\omega}$ (in the pass band of lattice ${\mathfrak{S}}$ \cite{Brillouin}) and a {\em lattice wave vector} $({\upkappa}_x, {\upkappa}_y)$.
Specifically, it is assumed that ${\su}^{{\mathrm{inc}}}$ is given by the expression
\begin{equation}\begin{split}
{\su}_{{\mathtt{x}}, {\mathtt{y}}}^{{\mathrm{inc}}}{\,:=}{{\mathrm{A}}}e^{i{\upkappa}_x {\mathtt{x}}+i{\upkappa}_y {\mathtt{y}}-i{\omega} t}, 
\quad\quad({\mathtt{x}}, {\mathtt{y}})\in{{\mathbb{Z}^2}}, 
\label{uinc}
\end{split}\end{equation}
where ${{\mathrm{A}}}\in\mathbb{C}$ is constant. 
In the remaining text, the explicit time dependence factor, $e^{-i{\omega} t}$, is suppressed.

By virtue of \eqref{dimnewtoneq} in intact lattice ${\mathfrak{S}}$, taking ${\su}={\su}^{{\mathrm{inc}}}$, with ${\upomega}{\,:=} \la{\omega},$ it is easy to see that ${\upomega}, {\upkappa}_x,$ and ${\upkappa}_y$ satisfy the dispersion relation \cite{Brillouin,Bls0} 
\begin{equation}\begin{split}
{\upomega}^2
=4(\sin^2{\frac{1}{2}}{\upkappa}_x+\sin^2{\frac{1}{2}}{\upkappa}_y), \quad\quad({\upkappa}_x, {\upkappa}_y)\in [-\pi, \pi]^2.
\label{dispersion}
\end{split}\end{equation}
The lattice wave \eqref{uinc} is diffracted by the crack tips and the rigid constraint tips. In order to avoid non-decaying wavefronts and associated technical issues, following a traditional choice in diffraction theory \cite{Bouwkamp, Noble,Bls0,Bls1}, it is assumed that
\begin{equation}\begin{split}
{\upomega}={\upomega}_1+i{\upomega}_2, \quad
{\upomega}_2>0.
\label{complexfreq}
\end{split}\end{equation}
Due to the dispersion relation \eqref{dispersion}, and \eqref{complexfreq}, ${\upkappa}_x$ and ${\upkappa}_y$ are also complex numbers. {It is noteworthy that the existence of annulus of Fig. \ref{schematicannulus_2}, which allows a convenient Wiener--Hopf formulation analyzed in the paper, owes to this assumption \eqref{complexfreq}}. 
The total displacement field ${\su}^{{t}}$ satisfies the discrete Helmholtz equation
\begin{subequations}\begin{eqnarray}
{\triangle}{\su}^{{t}}_{{\mathtt{x}}, {\mathtt{y}}}+{\upomega}^2{\su}^{{t}}_{{\mathtt{x}}, {\mathtt{y}}}=0, 
\quad({\mathtt{x}}, {\mathtt{y}})\in{{\mathbb{Z}^2}}\setminus{{\Sigma}}, \label{dHelmholtz}\\
\text{where }{\su}_{{\mathtt{x}}, {\mathtt{y}}}^{{t}}={\su}_{{\mathtt{x}}, {\mathtt{y}}}^{{\mathrm{inc}}}+{\su}_{{\mathtt{x}}, {\mathtt{y}}}, ({\mathtt{x}}, {\mathtt{y}})\in{{\mathbb{Z}^2}},
\label{utsplit}
\end{eqnarray}\label{dHelmholtzfull}\end{subequations}
and certain other equations on the defects ${\Sigma}$ (\eqref{defectK}, \eqref{defectC}) {as it is clear in the sequel where individual problem formulation appears}.
\begin{figure}[h!]
\centering
\includegraphics[width=\textwidth]{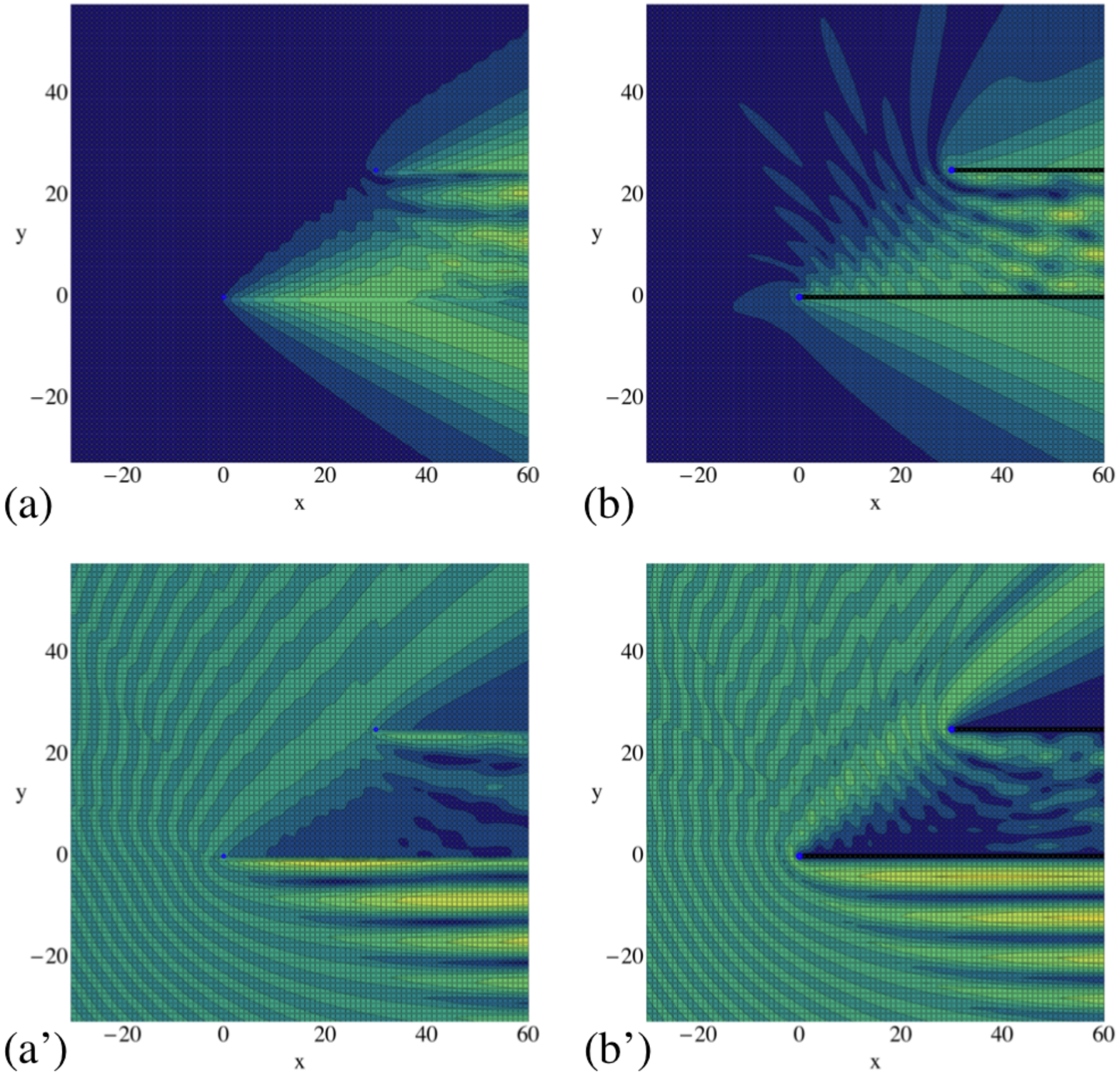}
\caption{{
Contour plot of the scattered field $|{\su}|$ on square lattice (a) with a pair of staggered cracks and (b) with a pair of staggered rigid constraints for ${\upomega}=0.9, \Theta=25$ deg, ${\mathtt{N}}=25, {\mathtt{M}}=30$.
Contour plot of the total field $\Re{\su}^{t}$ shown in (a') corresponding to (a) and in (b') that corresponding to (b).
The blue dots mark the edge sites with broken bond in (a), (a') and rigid constraint in (b), (b').
Numerical grid is $2N_g\times1\times2N_g\times1$ with $N_g=91+|{\mathtt{M}}|.$}}
\label{Fig1num}
\end{figure}
{As an illustration of the associated scattering phenomenon, Fig. \ref{Fig1num} presents the graphical results (contour plot of modulus of the scattered field and real part of the total field) based on a numerical solution (i.e, using the scheme described in Appendix D of \cite{Bls0}) of the discrete Helmholtz equation \eqref{dHelmholtz} in conjunction with the conditions at the scatterers. It is clear that the wave field generated due to the presence of a stagger involves a chain of complicated inter-tip interactions, so called multiple scattering effect. In a way, this observation is also related to the difficulty \cite{SpeckPros,GohbergKaashoek} of the associated Wiener--Hopf formulation \cite{GMthesis,Bls8staggerpair_asymp}, when it is defined using definition of Fourier transforms in a way natural to the defect tip geometry, and resides in the presence of peculiar off-diagonal factors \cite{Abrahams1} (recall \eqref{Abrakernel}).} 

Since the incident wave ${\su}^{{\mathrm{inc}}}_{{\mathtt{x}}, {\mathtt{y}}}$ \eqref{uinc} satisfies \eqref{dHelmholtz}, the scattered wave ${\su}_{{\mathtt{x}}, {\mathtt{y}}}$ \eqref{utsplit} also satisfies \eqref{dHelmholtz} {in the perfect lattice, i.e., away from the scatterers}. With {some auxiliary} details provided in Appendix \ref{applyFT}, a general solution of the latter can be prepared so that the scattered field, in between the defects as well as above and below them, can be written in terms of {certain} unknown functions.

In particular, for the problem of two cracks, since \eqref{complexfreq} {(along with a well known hypothesis of causality)} implies that ${\su}_{{\mathtt{y}}}^{{\mathrm{F}}}\to 0$ when ${\mathtt{y}}\to \pm\infty$, the Fourier transform {\eqref{discreteFT}} of the scattered wave field (as complex function, analytic on an annulus ${{\mathscr{A}}}$, described in Appendix \ref{applyFT} {and schematically illustrated in Fig. \ref{schematicannulus_2}}, see also \cite{Bls0,Bls1})
is found to be
\begin{equation}\begin{split}
{\su}_{{\mathtt{y}}}^{{\mathrm{F}}}
&={\su}^{{\mathrm{F}}}_0(\frac{{{\lambda}}^{-2{{\mathtt{N}}}+2}{{\lambda}}^{{\mathtt{y}}}-{{\lambda}}^{-{\mathtt{y}}}}{{{\lambda}}^{-2{{\mathtt{N}}}+2}-1})+{\su}^{{\mathrm{F}}}_{{{\mathtt{N}}}-1}(\frac{{{\lambda}}^{-{{\mathtt{N}}}+1}{{\lambda}}^{-{\mathtt{y}}}-{{\lambda}}^{-{{\mathtt{N}}}+1}{{\lambda}}^{{\mathtt{y}}}}{{{\lambda}}^{-2{{\mathtt{N}}}+2}-1})\\
&={{\mathscr{F}}}_{{\mathtt{y}}}{\su}^{{\mathrm{F}}}_0+{{\mathscr{G}}}_{{\mathtt{y}}}{\su}^{{\mathrm{F}}}_{{{\mathtt{N}}}-1}~({{\mathtt{y}}}\in\mathbb{Z}_0^{{{\mathtt{N}}}-1}),\\
{\su}_{{\mathtt{y}}}^{{\mathrm{F}}}&=\su^{{\mathrm{F}}}_{{\mathtt{N}}}{{\lambda}}^{{\mathtt{y}}-{{\mathtt{N}}}}~({{\mathtt{y}}} \ge {{\mathtt{N}}}), \quad{\su}_{{\mathtt{y}}}^{{\mathrm{F}}}=\su^{{\mathrm{F}}}_{-1}{{\lambda}}^{-({{\mathtt{y}}} +1)}~({{\mathtt{y}}} \le-1),
\label{ubulkK}
\end{split}\end{equation}
where the function ${\su}_{-1}^{{\mathrm{F}}}$, ${\su}_{0}^{{\mathrm{F}}}$, ${\su}_{{\mathtt{N}}-1}^{{\mathrm{F}}}$ and ${\su}_{{\mathtt{N}}}^{{\mathrm{F}}}$ are unknown functions,
{while, the definition of ${\lambda}$ is provided in \eqref{lamL} as part of Appendix \ref{applyFT}.}
As expected due to symmetry between the two cracked rows, various ${\mathscr{F}}$s are related to ${\mathscr{G}}$s {(in \eqref{ubulkK}${}_2$)}, for instance,
${\su}_{1}^{{\mathrm{F}}}={{\mathscr{F}}}_1{\su}^{{\mathrm{F}}}_0+{{\mathscr{F}}}_{{{\mathtt{N}}}-2}{\su}^{{\mathrm{F}}}_{{{\mathtt{N}}}-1}, {\su}_{{{\mathtt{N}}}-2}^{{\mathrm{F}}}={{\mathscr{F}}}_{{{\mathtt{N}}}-2}{\su}^{{\mathrm{F}}}_0+{{\mathscr{F}}}_1{\su}^{{\mathrm{F}}}_{{{\mathtt{N}}}-1}.$
In the following section, it is found that instead of four there are mainly two unknown functions to be determined, as expected \cite{Bls0}.

Similarly, for the problem of two rigid constraints,
the Fourier transform {\eqref{discreteFT}} of the scattered wave field
(as complex function, analytic on the annulus ${{\mathscr{A}}}$) is
\begin{equation}\begin{split}
{\su}_{{\mathtt{y}}}^{{\mathrm{F}}}
&={\su}_0^{{\mathrm{F}}}(\frac{{{\lambda}}^{-2{\mathtt{N}}}{{\lambda}}^{{\mathtt{y}}}}{{{\lambda}}^{-2{\mathtt{N}}}-1}-\frac{{{\lambda}}^{-{{\mathtt{y}}}}}{{{\lambda}}^{-2{\mathtt{N}}}-1})+{\su}^{{\mathrm{F}}}_{{\mathtt{N}}}(\frac{{{\lambda}}^{-{\mathtt{N}}}{{\lambda}}^{-{{\mathtt{y}}}}}{{{\lambda}}^{-2{\mathtt{N}}}-1}-\frac{{{\lambda}}^{-{\mathtt{N}}}{{\lambda}}^{{\mathtt{y}}}}{{{\lambda}}^{-2{\mathtt{N}}}-1})\\
&={{\mathscr{F}}}_{{\mathtt{y}}}{\su}^{{\mathrm{F}}}_0+{{\mathscr{G}}}_{{\mathtt{y}}}{\su}^{{\mathrm{F}}}_{{\mathtt{N}}}~({{\mathtt{y}}}\in\mathbb{Z}_0^{{{\mathtt{N}}}}), \\
{\su}_{{\mathtt{y}}}^{{\mathrm{F}}}&={\su}^{{\mathrm{F}}}_{{\mathtt{N}}}{{\lambda}}^{{{\mathtt{y}}}-{\mathtt{N}}}~({{\mathtt{y}}}\ge {{\mathtt{N}}}), {\su}_{{\mathtt{y}}}^{{\mathrm{F}}}={\su}^{{\mathrm{F}}}_{0}{{\lambda}}^{-{{\mathtt{y}}}}~({{\mathtt{y}}}\le0), 
\label{ubulkC}
\end{split}\end{equation}
where the two functions ${\su}_0^{{\mathrm{F}}}$ and ${\su}_{{\mathtt{N}}}^{{\mathrm{F}}}$ are unknown. 
{Again, note that the function ${\lambda}$ is provided in \eqref{lamL}.}
Indeed, the expressions \eqref{ubulkC} can be also written using the one provided {above} for the crack \eqref{ubulkK} by mapping ${\mathtt{y}}$ appropriately.
{The determination of ${\su}_0^{{\mathrm{F}}}$ and ${\su}_{{\mathtt{N}}}^{{\mathrm{F}}}$ is the purpose of a Wiener--Hopf technique based approach carried out in the section following the next one.}

\section{{Pair of cracks}}
\label{WHeqformK}
{In this section, the scattering problem associated with a pair of staggered cracks is analyzed within a Wiener--Hopf formulation, while building on the defining form of the equation of motion \eqref{dHelmholtz} and the incident wave \eqref{uinc} in the previous section \S\ref{sqLattFT}, as well as the expression in terms of the Fourier transform \eqref{ubulkK}. The scatterer is located at the sites stated in \eqref{defectK}. The latter clearly brings out the exact specification of the unknown functions as described in the statements immediately following \eqref{ubulkK}}.

{On the lattice rows ${\mathtt{y}}=0$ and ${\mathtt{y}}=-1$ (referred as the `lower' crack sometimes), as schematically shown in Fig. \ref{Fig1}(a), the sites with ${\mathtt{x}}\in{\mathbb{Z}^+}$ (recall the definition ${\mathbb{Z}^+}$ of in \eqref{ZpZn}) do not interact with all four nearest neighbours, and in fact miss one interaction. The bonds between the sites at ${\mathtt{y}}=0$ and ${\mathtt{y}}=-1$ are broken for ${\mathtt{x}}\in{\mathbb{Z}^+}$. On the other hand the incident wave satisfies \eqref{dHelmholtz} automatically owing to the dispersion relation \eqref{dispersion}. It is useful to recall also} the splitting of the total wave field according to \eqref{utsplit}. 
{After substitution of \eqref{utsplit} in the equation for the total displacement field, i.e., 
\begin{equation}\begin{split}
{\su}^{t}_{{\mathtt{x}}+1, {\mathtt{y}}}+{\su}^{t}_{{\mathtt{x}}-1, {\mathtt{y}}}+{\su}^{t}_{{\mathtt{x}}, {\mathtt{y}}+1}-3{\su}^{t}_{{\mathtt{x}}, {\mathtt{y}}}+{\upomega}^2{\su}^{{t}}_{{\mathtt{x}}, {\mathtt{y}}}=0, \\
{\su}^{t}_{{\mathtt{x}}+1, {\mathtt{y}}}+{\su}^{t}_{{\mathtt{x}}-1, {\mathtt{y}}}+{\su}^{t}_{{\mathtt{x}}, {\mathtt{y}}-1}-3{\su}^{t}_{{\mathtt{x}}, {\mathtt{y}}}+{\upomega}^2{\su}^{{t}}_{{\mathtt{x}}, {\mathtt{y}}}=0,
\label{twocrackeqn}
\end{split}\end{equation}
for ${\mathtt{y}}=0, -1$, respectively, and ${\mathtt{x}}\in{\mathbb{Z}^+}$, it is clear that the contribution of the incident wave via the broken bonds appears as a source term. Specifically,} the equation of motion {for the scattered component} at the rows corresponding to the lower crack is
\begin{subequations}\begin{eqnarray}-{\upomega}^2{\su}_{{\mathtt{x}}, -1}+({\su}_{{\mathtt{x}}, -1}-{\su}_{{\mathtt{x}}, 0}){\mathcal{H}}(-{\mathtt{x}}-1)&=&({\su}^{{\mathrm{inc}}}_{{\mathtt{x}}, -1}-{\su}^{{\mathrm{inc}}}_{{\mathtt{x}}, 0}){\mathcal{H}}({\mathtt{x}})+{\su}_{{\mathtt{x}}+1, -1}\notag\\
&&+{\su}_{{\mathtt{x}}-1, -1}+{\su}_{{\mathtt{x}}, -2}-3{\su}_{{\mathtt{x}}, -1},\\-{\upomega}^2{\su}_{{\mathtt{x}}, 0}+({\su}_{{\mathtt{x}}, 0}-{\su}_{{\mathtt{x}}, -1}){\mathcal{H}}(-{\mathtt{x}}-1)&=&({\su}^{{\mathrm{inc}}}_{{\mathtt{x}}, 0}-{\su}^{{\mathrm{inc}}}_{{\mathtt{x}}, -1}){\mathcal{H}}({\mathtt{x}})+{\su}_{{\mathtt{x}}+1, 0}\notag\\
&&+{\su}_{{\mathtt{x}}-1, 0}+{\su}_{{\mathtt{x}}, 1}-3{\su}_{{\mathtt{x}}, 0}.
\end{eqnarray}\label{crackt1}\end{subequations}
{Indeed, for ${\mathtt{x}}\in{\mathbb{Z}^-},$ \eqref{crackt1} reduces to the discrete Helmholtz equation (with ${\su}$ replacing ${\su}^{{t}}$ in \eqref{dHelmholtz}), where the definition of ${\mathbb{Z}^-}$ has been provided in \eqref{ZpZn}.}

For the rows ${{\mathtt{y}}}={\mathtt{N}}, {\mathtt{N}}-1$, {the presence of offset ${{\mathtt{M}}}\in\mathbb{Z}$ leads to it being a negative or positive integer}, while ${\mathtt{M}}=0$ as a special {zero offset case, associated with the geometric configuration of aligned parallel crack edges}, has been dealt with earlier \cite{Bls8pair1}. 
{The necessary details of the Wiener--Hopf formulation for the two semi-infinite arrays of broken bonds, i.e., `lower' and `upper' cracks, are provided below for both cases when ${{\mathtt{M}}}>0$ as well as when ${{\mathtt{M}}}<0$.}

\subsection{{${\mathtt{M}}>0$}}
\label{crackMP}
{For the rows ${\mathtt{y}}={{\mathtt{N}}}, {{\mathtt{N}}}-1$ (referred as the `upper' crack sometimes), respectively, 
\eqref{twocrackeqn} holds,
whenever ${\mathtt{x}}-{{\mathtt{M}}}\in{\mathbb{Z}^+}.$}
Then, {after substitution of \eqref{utsplit} and an addition and subtraction of the same terms}, the equation of motion at the rows corresponding to the upper crack, as schematically shown in Fig. \ref{Fig1}(a), {can be written as, at ${\mathtt{y}}={{\mathtt{N}}}-1, {{\mathtt{N}}}$, respectively,}
\begin{subequations}\begin{eqnarray}
&&-{\upomega}^2{\su}_{{\mathtt{x}}, {{\mathtt{N}}}-1}+({\su}_{{\mathtt{x}}, {{\mathtt{N}}}-1}-{\su}_{{\mathtt{x}}, {{\mathtt{N}}}}){\mathcal{H}}(-{\mathtt{x}}+{{\mathtt{M}}}-1)+{{{\mathtt{f}}}}_{{\mathtt{x}}}-{{{\mathtt{f}}}}_{{\mathtt{x}}}\notag\\
&=&({\su}^{{\mathrm{inc}}}_{{\mathtt{x}}, {{\mathtt{N}}}-1}-{\su}^{{\mathrm{inc}}}_{{\mathtt{x}}, {{\mathtt{N}}}}){\mathcal{H}}({\mathtt{x}}-{{\mathtt{M}}})-{{{\mathtt{f}}}}^{{\mathrm{inc}}}_{{\mathtt{x}}}+{{{\mathtt{f}}}}^{{\mathrm{inc}}}_{{\mathtt{x}}}+{\su}_{{\mathtt{x}}+1, {{\mathtt{N}}}-1}\notag\\
&&+{\su}_{{\mathtt{x}}-1, {{\mathtt{N}}}-1}+{\su}_{{\mathtt{x}}, {{\mathtt{N}}}-2}-3{\su}_{{\mathtt{x}}, {{\mathtt{N}}}-1},\label{ckt2a}\\
&&-{\upomega}^2{\su}_{{\mathtt{x}}, {{\mathtt{N}}}}+({\su}_{{\mathtt{x}}, {{\mathtt{N}}}}-{\su}_{{\mathtt{x}}, {{\mathtt{N}}}-1}){\mathcal{H}}(-{\mathtt{x}}+{{\mathtt{M}}}-1)-{{{\mathtt{f}}}}_{{\mathtt{x}}}+{{{\mathtt{f}}}}_{{\mathtt{x}}}\notag\\
&=&({\su}^{{\mathrm{inc}}}_{{\mathtt{x}}, {{\mathtt{N}}}}-{\su}^{{\mathrm{inc}}}_{{\mathtt{x}}, {{\mathtt{N}}}-1}){\mathcal{H}}({\mathtt{x}}-{{\mathtt{M}}})+{{{\mathtt{f}}}}^{{\mathrm{inc}}}_{{\mathtt{x}}}-{{{\mathtt{f}}}}^{{\mathrm{inc}}}_{{\mathtt{x}}}+{\su}_{{\mathtt{x}}+1, {{\mathtt{N}}}}\notag\\
&&+{\su}_{{\mathtt{x}}{{\mathtt{N}}}-1, {{\mathtt{N}}}}+{\su}_{{\mathtt{x}}, {{\mathtt{N}}}+1}-3{\su}_{{\mathtt{x}}, {{\mathtt{N}}}},\label{ckt2b}
\end{eqnarray}\label{crackt2}\end{subequations}
{where the adopted definition of ${{{\mathtt{f}}}}_{{\mathtt{x}}}$ and ${{{\mathtt{f}}}}^{{\mathrm{inc}}}_{{\mathtt{x}}}$ is}
\begin{equation}\begin{split}
{{{\mathtt{f}}}}_{{\mathtt{x}}}{\,:=}-({\su}_{{\mathtt{x}}, {{\mathtt{N}}}-1}-{\su}_{{\mathtt{x}}, {{\mathtt{N}}}}){\mathcal{H}}({\mathtt{x}}){\mathcal{H}}(-{\mathtt{x}}+{{\mathtt{M}}}-1),\\
{{{\mathtt{f}}}}^{{\mathrm{inc}}}_{{\mathtt{x}}}{\,:=}-({\su}^{{\mathrm{inc}}}_{{\mathtt{x}}, {{\mathtt{N}}}-1}-{\su}^{{\mathrm{inc}}}_{{\mathtt{x}}, {{\mathtt{N}}}}){\mathcal{H}}({\mathtt{x}}){\mathcal{H}}(-{\mathtt{x}}+{{\mathtt{M}}}-1).
\label{defppinc}
\end{split}\end{equation}
{Indeed, for ${\mathtt{x}}-{{\mathtt{M}}}\in{\mathbb{Z}^-},$ \eqref{crackt2} reduces to the discrete Helmholtz equation (with ${\su}$ replacing ${\su}^{{t}}$ in \eqref{dHelmholtz}).}
{In light of the peculiar structure of the equations \eqref{crackt1} and \eqref{crackt2} along with the definitions \eqref{defppinc}, it is natural to introduce a separate notation for the bond length in the lattice rows where a crack is present}.
Let the scattered and incident component of the displacement field relative to the vertical bonds in cracked rows be defined by
\begin{subequations}\begin{eqnarray}
{\mathrm{v}}_{{\mathtt{x}}, 0}{\,:=}{\su}_{{\mathtt{x}}, 0}-{\su}_{{\mathtt{x}}, -1}, \quad
{\mathrm{v}}^{{\mathrm{inc}}}_{{\mathtt{x}}, 0}{\,:=}{\su}^{{\mathrm{inc}}}_{{\mathtt{x}}, 0}-{\su}^{{\mathrm{inc}}}_{{\mathtt{x}}, -1},
\label{v0defAbra}\\
{\mathrm{v}}_{{\mathtt{x}}, {{\mathtt{N}}}}{\,:=}{\su}_{{\mathtt{x}}, {{\mathtt{N}}}}-{\su}_{{\mathtt{x}}, {{\mathtt{N}}}-1},
 \quad
{\mathrm{v}}^{{\mathrm{inc}}}_{{\mathtt{x}}, {{\mathtt{N}}}}{\,:=}{\su}^{{\mathrm{inc}}}_{{\mathtt{x}}, {{\mathtt{N}}}}-{\su}^{{\mathrm{inc}}}_{{\mathtt{x}}, {{\mathtt{N}}}-1},\label{vNdefAbra}
\end{eqnarray}\label{vdefAbrafull}\end{subequations}
respectively. 
{Thus, according to \eqref{defppinc},
$$
\text{for }{{\mathtt{x}}\in\mathbb{Z}_0^{{\mathtt{M}}-1}}, {{{\mathtt{f}}}}_{{\mathtt{x}}}={\mathrm{v}}_{{\mathtt{x}}, {{\mathtt{N}}}}, {{{\mathtt{f}}}}^{{\mathrm{inc}}}_{{\mathtt{x}}}={\mathrm{v}}^{{\mathrm{inc}}}_{{\mathtt{x}}, {{\mathtt{N}}}}.
$$
The reason for introduction of certain added and subtracted terms in \eqref{crackt2} is clear from a re-grouping of the terms appearing therein, that is,
\begin{subequations}\begin{eqnarray}
-{\upomega}^2{\su}_{{\mathtt{x}}, {{\mathtt{N}}}-1}-{\mathrm{v}}_{{\mathtt{x}}, {{\mathtt{N}}}}{\mathcal{H}}(-{\mathtt{x}}-1)&=&-{\mathrm{v}}^{{\mathrm{inc}}}_{{\mathtt{x}}, {{\mathtt{N}}}}{\mathcal{H}}({\mathtt{x}})+{\su}_{{\mathtt{x}}+1, {{\mathtt{N}}}-1}+{\su}_{{\mathtt{x}}-1, {{\mathtt{N}}}-1}\notag\\
&&+{\su}_{{\mathtt{x}}, {{\mathtt{N}}}-2}-3{\su}_{{\mathtt{x}}, {{\mathtt{N}}}-1}+[{{{\mathtt{f}}}}_{{\mathtt{x}}}+{{{\mathtt{f}}}}^{{\mathrm{inc}}}_{{\mathtt{x}}}],\label{t2a}\\
-{\upomega}^2{\su}_{{\mathtt{x}}, {{\mathtt{N}}}}+{\mathrm{v}}_{{\mathtt{x}}, {{\mathtt{N}}}}{\mathcal{H}}(-{\mathtt{x}}-1)&=&{\mathrm{v}}^{{\mathrm{inc}}}_{{\mathtt{x}}, {{\mathtt{N}}}}{\mathcal{H}}({\mathtt{x}})+{\su}_{{\mathtt{x}}+1, {{\mathtt{N}}}}+{\su}_{{\mathtt{x}}{{\mathtt{N}}}-1, {{\mathtt{N}}}}\notag\\
&&+{\su}_{{\mathtt{x}}, {{\mathtt{N}}}+1}-3{\su}_{{\mathtt{x}}, {{\mathtt{N}}}}
-[{{{\mathtt{f}}}}_{{\mathtt{x}}}+{{{\mathtt{f}}}}^{{\mathrm{inc}}}_{{\mathtt{x}}}],\label{t2b}
\end{eqnarray}\end{subequations}
which allows the equations at ${\mathtt{y}}={\mathtt{N}}-1, {\mathtt{N}}$ to have the same form as the equations at ${\mathtt{y}}=-1, 0$ in \eqref{crackt1} (using \eqref{v0defAbra}) modulo the presence of `source terms' shown in square brackets. In this manner, the equation of motion for the faces of both staggered cracks have been written in terms of zero offset, i.e., `aligned' cracks (for which it is to be noted that the corresponding scattering problem has been endowed with an exact solution \cite{Bls8pair1}).}

Let
\begin{equation}\begin{split}
{{\mathrm{v}}}_0^{{\mathrm{inc}}}{}^+=
\sum\limits_{{\mathtt{x}}\in{\mathbb{Z}^+}}{{z}}^{-{\mathtt{x}}}{\mathrm{v}}^{{\mathrm{inc}}}_{{\mathtt{x}}, 0}, \quad
{{\mathrm{v}}}_{{\mathtt{N}}}^{{\mathrm{inc}}}{}^+=
\sum\limits_{{\mathtt{x}}\in{\mathbb{Z}^+}}{{z}}^{-{\mathtt{x}}}{\mathrm{v}}^{{\mathrm{inc}}}_{{\mathtt{x}}, {\mathtt{N}}},
\label{t3}
\end{split}\end{equation}
Indeed, {in an another symbolic notation}, ${{\mathrm{v}}}_0^{{\mathrm{inc}}}{}^+={{\mathrm{A}}}(1-e^{-i\upkappa_y})\delta_{D}^{+}({{z}} {z}_{{P}}^{-1})$ and ${{\mathrm{v}}}_{{\mathtt{N}}}^{{\mathrm{inc}}}{}^+={{\mathrm{A}}}(1-e^{-i\upkappa_y})e^{i\upkappa_y{\mathtt{N}}}\delta_{D}^{+}({{z}} {z}_{{P}}^{-1})$ {according to the expression of the incident wave \eqref{uinc}}.
{Here,}
\begin{equation}\begin{split}
{\su}_{0}^{+}({{z}})=\sum\limits_{{\mathtt{x}}\in{\mathbb{Z}^+}} {\su}_{{\mathtt{x}}, 0}{{z}}^{-{\mathtt{x}}}, \quad
{\su}_{0}^{-}({{z}})=\sum\limits_{{\mathtt{x}}\in{\mathbb{Z}^-}} {\su}_{{\mathtt{x}}, 0}{{z}}^{-{\mathtt{x}}},
\label{u0pmAbra}
\end{split}\end{equation}
\begin{equation}\begin{split}
\delta_{D}^+({z})=\sum\limits_{{\mathtt{x}}\in{\mathbb{Z}^+}}{z}^{-{\mathtt{x}}}=\frac{1}{1-{z}^{-1}}~(|{z}|>1),
\label{delDp}
\end{split}\end{equation}
\begin{equation}\begin{split}
\text{and }
{{z}}_{{P}}{\,:=} e^{{+}i{\upkappa}_x}.
\label{zP}
\end{split}\end{equation}
Consider definitions similar to \eqref{u0pmAbra} for ${\su}_{{\mathtt{N}}+1}^{\pm}$ and ${\su}_{{{\mathtt{N}}}-1}^{\pm},$ and so on.
Note that 
according to \eqref{defppinc}, {using \eqref{vdefAbrafull}},
\begin{equation}\begin{split}
{{{\mathtt{f}}}}^{+}({z})=\sum\limits_{{\mathtt{x}}\in\mathbb{Z}}{{{\mathtt{f}}}}_{{\mathtt{x}}}{z}^{-{\mathtt{x}}}=\sum\limits_{{\mathtt{x}}\in\mathbb{Z}_0^{{\mathtt{M}}-1}}{\mathrm{v}}_{{\mathtt{x}}, {{\mathtt{N}}}}{z}^{-{\mathtt{x}}},\quad
{{{\mathtt{f}}}}^{{\mathrm{inc}} {+}}({z})=\sum\limits_{{\mathtt{x}}\in\mathbb{Z}}{{{\mathtt{f}}}}^{{\mathrm{inc}}}_{{\mathtt{x}}}{z}^{-{\mathtt{x}}}=\sum\limits_{{\mathtt{x}}\in\mathbb{Z}_0^{{\mathtt{M}}-1}}{\mathrm{v}}^{{\mathrm{inc}}}_{{\mathtt{x}}, {{\mathtt{N}}}}{z}^{-{\mathtt{x}}}.
\label{ppiAbra}
\end{split}\end{equation}
{The decorative notation $\mathtt{f}^{+}$ and $\mathtt{f}^{{\mathrm{inc}} {+}}$ has been adopted as $\mathtt{M}>0$ implies that these complex functions are polynomials in the variable ${z}^{-1}$; recall the statement following \eqref{discreteFT}}.

Taking the Fourier transform \eqref{discreteFT} of \eqref{crackt1} and \eqref{crackt2}, using the definition of {a relevant complex function} ${\mathtt{H}}$ \eqref{h2} {(stated as part of Appendix \ref{applyFT})}, it is found that
\begin{subequations}\begin{eqnarray}
({{\mathtt{H}}}+1-{{\lambda}})({\su}_{-1}^{-}+{\su}_{-1}^{+})-({\su}_{0}^{-}-{\su}_{-1}^{-})&=&-{{\mathrm{v}}}_0^{{\mathrm{inc}}}{}^+, 
\label{twocracksn1Abra}
\\
({{\mathtt{H}}}+1-{{{\mathscr{F}}}_1})({\su}_{0}^{-}+{\su}_{0}^{+})+({\su}_{0}^{-}-{\su}_{-1}^{-})
&=&{{\mathrm{v}}}_0^{{\mathrm{inc}}}{}^++{\su}_{{{\mathtt{N}}}-1}^{{\mathrm{F}}}{{{\mathscr{F}}}_{{{\mathtt{N}}}-2}}, \label{twocracks0Abra}
\\
({{\mathtt{H}}}+1-{{{\mathscr{F}}}_1})({\su}_{{{\mathtt{N}}}-1}^{-}+{\su}_{{{\mathtt{N}}}-1}^{+})+({\su}_{{{\mathtt{N}}}-1}^{-}-{\su}_{{{\mathtt{N}}}}^{-})
&=&-{{\mathrm{v}}}_{{\mathtt{N}}}^{{\mathrm{inc}}}{}^++{{{\mathtt{f}}}}^{+}+{{{\mathtt{f}}}}^{{\mathrm{inc}} {+}}+{\su}_{0}^{{\mathrm{F}}}{{{\mathscr{F}}}_{{{\mathtt{N}}}-2}}, \label{twocracksNn1Abra}
\\
({{\mathtt{H}}}+1-{{\lambda}})({\su}_{{{\mathtt{N}}}}^{-}+{\su}_{{{\mathtt{N}}}}^{+})-({\su}_{{{\mathtt{N}}}-1}^{-}-{\su}_{{{\mathtt{N}}}}^{-})&=&{{\mathrm{v}}}_{{\mathtt{N}}}^{{\mathrm{inc}}}{}^+-{{{\mathtt{f}}}}^{+}-{{{\mathtt{f}}}}^{{\mathrm{inc}} {+}}.
\label{twocracksNAbra}
\end{eqnarray}\label{t5}\end{subequations}
Note that according to \eqref{ubulkK},
${\su}_{-2}^{{\mathrm{F}}}=\su^{{\mathrm{F}}}_{-1}{{\lambda}}, {\su}_{1}^{{\mathrm{F}}}={{\mathscr{F}}}_1{\su}_0+{{\mathscr{G}}}_1{\su}_{{{\mathtt{N}}}-1}, {\su}_{{{\mathtt{N}}}-2}^{{\mathrm{F}}}={{\mathscr{F}}}_{{{\mathtt{N}}}-2}{\su}_0+{{\mathscr{G}}}_{{{\mathtt{N}}}-2}{\su}_{{{\mathtt{N}}}-1}, {\su}_{{{\mathtt{N}}}+1}^{{\mathrm{F}}}=\su^{{\mathrm{F}}}_{{\mathtt{N}}}{{\lambda}},$ which has been used in above equations.
The particular ${\mathscr{F}}$s and ${\mathscr{G}}$s appearing are according to the definitions that can be read from \eqref{ubulkK}.
The four equations \eqref{t5} are coupled through the terms ${\su}_{0}^{{\mathrm{F}}}$ and ${\su}_{{{\mathtt{N}}}-1}^{{\mathrm{F}}}$ which becomes weaker as ${{\mathtt{N}}}$ increases. 
But this system of four equations can be {effectively} reduced to two equations whose solution is sufficient to solve the problem {of discrete scattering due to a pair of cracks}.

Using \eqref{twocracksn1Abra} and \eqref{twocracksNAbra}, respectively,
\begin{subequations}\begin{eqnarray}
{\su}_{0}^{-}&=&({{\lambda}}^{-1}-1)({\su}_{-1}^{-}+{\su}_{-1}^{+})+{\su}_{-1}^{-}+{{\mathrm{v}}}_0^{{\mathrm{inc}}}{}^+, 
\label{u0nAbra}
\\
{\su}_{{{\mathtt{N}}}-1}^{-}&=&({{\lambda}}^{-1}-1)({\su}_{{{\mathtt{N}}}}^{-}+{\su}_{{{\mathtt{N}}}}^{+})+{\su}_{{{\mathtt{N}}}}^{-}-{{\mathrm{v}}}_{{\mathtt{N}}}^{{\mathrm{inc}}}{}^++{{{\mathtt{f}}}}^{+}+{{{\mathtt{f}}}}^{{\mathrm{inc}} {+}}.
\label{uNn1Abra}
\end{eqnarray}\label{t6}\end{subequations}

Adding first and second equation (\eqref{twocracksn1Abra} and \eqref{twocracks0Abra}), and third and fourth (\eqref{twocracksNn1Abra} and \eqref{twocracksNAbra}), respectively, 
\begin{subequations}\begin{eqnarray}
({{\lambda}}^{-1}-1)({\su}_{-1}^{-}+{\su}_{-1}^{+})&=&-({{\mathtt{H}}}+1-{{{\mathscr{F}}}_1}){\su}_{0}^{{\mathrm{F}}}+{{{\mathscr{F}}}_{{{\mathtt{N}}}-2}}({\su}_{{{\mathtt{N}}}-1}^{-}+{\su}_{{{\mathtt{N}}}-1}^{+}),
\label{un1Abra} 
\\
({{\lambda}}^{-1}-1)({\su}_{{{\mathtt{N}}}}^{-}+{\su}_{{{\mathtt{N}}}}^{+})&=&{{{\mathscr{F}}}_{{{\mathtt{N}}}-2}}{\su}_{0}^{{\mathrm{F}}}-({{\mathtt{H}}}+1-{{{\mathscr{F}}}_1})({\su}_{{{\mathtt{N}}}-1}^{-}+{\su}_{{{\mathtt{N}}}-1}^{+}).
\label{uNAbra}
\end{eqnarray}\label{t7}\end{subequations}
{In view of \eqref{discreteFT}, it is observed that} above equations, {i.e., \eqref{un1Abra} and \eqref{uNAbra},} are merely algebraic equations {involving the `full' Fourier transforms} and provide an expression of ${\su}_{-1}^{{\mathrm{F}}}$ and ${\su}_{{\mathtt{N}}}^{{\mathrm{F}}}$ in terms of ${\su}_{0}^{{\mathrm{F}}}$ and ${\su}_{{{\mathtt{N}}}-1}^{{\mathrm{F}}}$, or vice versa. 
Eventually, using {\eqref{v0defAbra}${}_1$ and \eqref{vNdefAbra}${}_1$} in \eqref{un1Abra} and \eqref{uNAbra}, respectively,
\begin{subequations}\begin{eqnarray}
({{\lambda}}^{-1}-1){\mathrm{v}}_{0}^{{\mathrm{F}}}
&=&({{\mathtt{H}}}-{{{\mathscr{F}}}_1}+{{\lambda}}^{-1}){\su}_{0}^{{\mathrm{F}}}-{{{\mathscr{F}}}_{{{\mathtt{N}}}-2}}{\su}_{{{\mathtt{N}}}-1}^{{\mathrm{F}}}, 
\label{v0nrelAbra}
\\
({{\lambda}}^{-1}-1){\mathrm{v}}^{{\mathrm{F}}}_{{\mathtt{N}}}&=&
{{{\mathscr{F}}}_{{{\mathtt{N}}}-2}}{\su}_{0}^{{\mathrm{F}}}-({{\mathtt{H}}}-{{{\mathscr{F}}}_1}+{{\lambda}}^{-1}){\su}_{{{\mathtt{N}}}-1}^{{\mathrm{F}}},
\label{vNnrelAbra}
\end{eqnarray}\label{v0NAbra}\end{subequations}
indeed, {yields a $2\times2$ matrix based relation}
\begin{equation}\begin{split}\begin{bmatrix}{\su}_{0}^{{\mathrm{F}}}\\{\su}_{{\mathtt{N}}-1}^{{\mathrm{F}}}\end{bmatrix}&=\frac{1-{{\lambda}}^{-1}}{({{\mathscr{F}}}_1-{{\mathtt{H}}}-{{\lambda}}^{-1})^2- {{\mathscr{F}}}_{{{\mathtt{N}}}-2}^2}\begin{bmatrix}({{\mathscr{F}}}_1-{{\mathtt{H}}} -{{\lambda}}^{-1})&{{{\mathscr{F}}}_{{{\mathtt{N}}}-2}}\\-{{{\mathscr{F}}}_{{{\mathtt{N}}}-2}}&-({{\mathscr{F}}}_1-{{\mathtt{H}}} -{{\lambda}}^{-1})\end{bmatrix}\begin{bmatrix}{\mathrm{v}}_{0}^{{\mathrm{F}}}\\{{\mathrm{v}}}_{{\mathtt{N}}}^{{\mathrm{F}}}\end{bmatrix},\label{v0NAbraeqn}\end{split}\end{equation}
which can be substituted in \eqref{u0nAbra} and \eqref{uNn1Abra}, respectively, i.e., {after a bit of re-writing, leading to}
\begin{subequations}\begin{eqnarray}
{\mathrm{v}}_{0-}={\su}_{0}^{-}-{\su}_{-1}^{-}&=&({{\lambda}}^{-1}-1){\su}_{-1}^{{\mathrm{F}}}+{{\mathrm{v}}}_0^{{\mathrm{inc}}}{}^+, 
\label{v0nAbra}
\\
{\mathrm{v}}_{{{\mathtt{N}}}-}={\su}_{{{\mathtt{N}}}}^{-}-{\su}_{{{\mathtt{N}}}-1}^{-}&=&
-({{\lambda}}^{-1}-1){\su}^{{\mathrm{F}}}_{{\mathtt{N}}}+{{\mathrm{v}}}_{{\mathtt{N}}}^{{\mathrm{inc}}}{}^+-{{{\mathtt{f}}}}^{+}-{{{\mathtt{f}}}}^{{\mathrm{inc}} {+}},
\label{vNnAbra}
\\
\text{and {using}}
\begin{bmatrix}
{\su}_{-1}^{{\mathrm{F}}}\\
{\su}_{{\mathtt{N}}}^{{\mathrm{F}}}
\end{bmatrix}&=&\begin{bmatrix}
{\su}_{0}^{{\mathrm{F}}}\\
{\su}_{{{\mathtt{N}}}-1}^{{\mathrm{F}}}
\end{bmatrix}-\begin{bmatrix}
{\mathrm{v}}_{0}^{{\mathrm{F}}}\\ 
-{\mathrm{v}}^{{\mathrm{F}}}_{{\mathtt{N}}}
\end{bmatrix}.
\label{un1uNF}
\end{eqnarray}\label{v0NAbrafull}\end{subequations}

\subsection{{${\mathtt{M}}<0$}}
\label{crackMN}
{The manipulations for the negative values of ${\mathtt{M}}$ follow those in \S\ref{crackMP}.}
Note that {the equation of motion at ${\mathtt{y}}={{\mathtt{N}}}-1, {{\mathtt{N}}}$, respectively, continues to be
\eqref{t2a} and \eqref{t2b} with the definitions}
\begin{equation}\begin{split}
{{{\mathtt{f}}}}_{{\mathtt{x}}}=+({\su}_{{\mathtt{x}}, {{\mathtt{N}}}-1}-{\su}_{{\mathtt{x}}, {{\mathtt{N}}}}){\mathcal{H}}(-{\mathtt{x}}-1){\mathcal{H}}(-{\mathtt{x}}+{{\mathtt{M}}}),\\
{{{\mathtt{f}}}}^{{\mathrm{inc}}}_{{\mathtt{x}}}=+({\su}^{{\mathrm{inc}}}_{{\mathtt{x}}, {{\mathtt{N}}}-1}-{\su}^{{\mathrm{inc}}}_{{\mathtt{x}}, {{\mathtt{N}}}}){\mathcal{H}}(-{\mathtt{x}}-1){\mathcal{H}}(-{\mathtt{x}}+{{\mathtt{M}}}),
\label{defppincN}
\end{split}\end{equation}
in place of \eqref{defppinc}, so that {(as analogues of \eqref{ppiAbra})}
\begin{equation}\begin{split}
{{{\mathtt{f}}}}^{-}({z})=\sum\limits_{{\mathtt{x}}\in\mathbb{Z}}{{{\mathtt{f}}}}_{{\mathtt{x}}}{z}^{-{\mathtt{x}}}
=-\sum\limits_{{\mathtt{x}}\in\mathbb{Z}_{{\mathtt{M}}}^{-1}}{{\mathrm{v}}}_{{\mathtt{x}}, {{\mathtt{N}}}}{z}^{-{\mathtt{x}}}, 
{{{\mathtt{f}}}}^{{\mathrm{inc}} {-}}({z})=\sum\limits_{{\mathtt{x}}\in\mathbb{Z}}{{{\mathtt{f}}}}^{{\mathrm{inc}}}_{{\mathtt{x}}}{z}^{-{\mathtt{x}}}
=-\sum\limits_{{\mathtt{x}}\in\mathbb{Z}_{{\mathtt{M}}}^{-1}}{{\mathrm{v}}}^{{\mathrm{inc}}}_{{\mathtt{x}}, {{\mathtt{N}}}}{z}^{-{\mathtt{x}}},
\label{ppiNAbra}
\end{split}\end{equation}
while other {steps} follow above case of ${\mathtt{M}}>0$. 
{The notation $\mathtt{f}^{-}$ and $\mathtt{f}^{{\mathrm{inc}} {-}}$ has been adopted as $\mathtt{M}>0$ implies that these complex functions are polynomials in the variable ${z}$; recall the statement following \eqref{discreteFT}}.

{After the introduction of these definitions, it is found that \eqref{t5} appears in the same form except that} the functions ${{{\mathtt{f}}}}^{-}$ and ${{{\mathtt{f}}}}^{{\mathrm{inc}} {-}}$ {(stated above \eqref{ppiNAbra})} take the place of ${{{\mathtt{f}}}}^{+}$ and ${{{\mathtt{f}}}}^{{\mathrm{inc}} {+}}$.
Finally, \eqref{v0nAbra} and \eqref{vNnAbra} need to be replaced with the equations
\begin{subequations}\begin{eqnarray}
{\mathrm{v}}_{0-}={\su}_{0}^{-}-{\su}_{-1}^{-}&=&({{\lambda}}^{-1}-1){\su}_{-1}^{{\mathrm{F}}}+{{\mathrm{v}}}_0^{{\mathrm{inc}}}{}^+, 
\label{v0nAbraN}
\\
{\mathrm{v}}_{{{\mathtt{N}}}-}={\su}_{{{\mathtt{N}}}}^{-}-{\su}_{{{\mathtt{N}}}-1}^{-}&=&
-({{\lambda}}^{-1}-1){\su}^{{\mathrm{F}}}_{{\mathtt{N}}}+{{\mathrm{v}}}_{{\mathtt{N}}}^{{\mathrm{inc}}}{}^+-{{{\mathtt{f}}}}^{-}-{{{\mathtt{f}}}}^{{\mathrm{inc}} {-}}.
\label{vNnAbraN}
\end{eqnarray}\label{v0NAbrafullN}\end{subequations}

\subsection{{Discrete {{Wiener--Hopf}} Equation}}
\label{WHcrackeqns}
In view of the {similarity of} expressions for the two cases in \S\ref{crackMP} and \S\ref{crackMN}, {it is desirable that a unified presentation is placed}. {For this purpose, suppose that} $\sgnM$ denotes the sign of ${\mathtt{M}}$.
{Introducing the definition, for convenience,}
\begin{equation}\begin{split}
\boldsymbol{p}^{\sgnM}&=\begin{bmatrix}
0\\
-{{{\mathtt{f}}}}^{\sgnM}-{{{\mathtt{f}}}}^{{\mathrm{inc}} {\sgnM}}
\end{bmatrix},
\label{qPAbra}
\end{split}\end{equation}
it is found from \eqref{v0NAbrafull} and \eqref{v0NAbrafullN} that\footnote{\label{invnote}For convenience of writing some long expressions, the notation $\recip{A}$ is adopted to denote the reciprocal, multiplicative inverse of $A$ (for matrix functions, it denotes the inverse matrix function while for non-zero real functions, it is simply the algebraic reciprocal).}

\begin{equation}\begin{split}
\begin{bmatrix}
{\mathrm{v}}_{0}^{-}\\ {\mathrm{v}}_{{\mathtt{N}}}^{-}\end{bmatrix}=(\recip{{\mathbf{K}}}-\mathbf{I})\begin{bmatrix}
{\mathrm{v}}_{0}^{+}\\ {\mathrm{v}}_{{\mathtt{N}}}^{+}\end{bmatrix}+\recip{{\mathbf{K}}}(\boldsymbol{q}^{{\mathrm{inc}}}_++\boldsymbol{p}^{\sgnM}), 
\label{WHKpre}
\end{split}\end{equation}
where {(with ${z}_{{P}}$ defined in \eqref{zP})}
\begin{equation}\begin{split}
\boldsymbol{q}^{{\mathrm{inc}}}_+=\begin{bmatrix}{{\mathrm{v}}}_0^{{\mathrm{inc}}}{}^+\\{{\mathrm{v}}}_{{\mathtt{N}}}^{{\mathrm{inc}}}{}^+\end{bmatrix}
={{\mathrm{A}}}(1-e^{-i\upkappa_y})\delta_{D}^{+}({{z}} {z}_{{P}}^{-1})\begin{bmatrix}
1\\
e^{i\upkappa_y{\mathtt{N}}}
\end{bmatrix},
\end{split}\end{equation}
\begin{equation}\begin{split}
\text{and }
\recip{{\mathbf{K}}}
=\frac{1}{{{\lambda}}+1}\begin{bmatrix}
{2 {{\lambda}}} & {({{\lambda}}-1) {{\lambda}}^{{\mathtt{N}}}} \\ {({{\lambda}}-1) {{\lambda}}^{{\mathtt{N}}}} & {2 {{\lambda}}}
\end{bmatrix}.\label{defAtilde}
\end{split}\end{equation}
{Using the expression \eqref{defAtilde},
it is found that
$$\recip{{\mathbf{K}}}-\mathbf{I}=-\frac{1-{\lambda}}{1+{\lambda}}\begin{bmatrix}
1&{\lambda}^{{\mathtt{N}}}\\
{\lambda}^{{\mathtt{N}}}&1
\end{bmatrix}.$$
Using the definitions of ${\lambda}, {\mathtt{h}}$, and ${\mathtt{r}}$ as provided in \eqref{lamL} (also see \eqref{h2}), the function
$(1-{\lambda})/(1+{\lambda})$ can be also written as ${{\mathtt{h}}}/{{\mathtt{r}}}$.
In fact, ${{\mathtt{h}}}/{{\mathtt{r}}}$ is recognized as the (scalar) {{Wiener--Hopf}} kernel for the single crack scattering problem \cite{Bls0,Bls2}, so that it is pertinent to denote it by the symbol
\begin{equation}\begin{split}
{{\mathcal{L}}}_{{k}}{\,:=}\frac{{\mathtt{h}}}{{\mathtt{r}}}.
\label{defLk}\end{split}\end{equation}}
Re-arranging the equation \eqref{WHKpre} to {the} standard form {\cite{Noble}}, the {{Wiener--Hopf}} equation is recognized as
\begin{equation}\begin{split}
\boldsymbol{{\mathrm{v}}}^-+\mathbf{L}\boldsymbol{{\mathrm{v}}}^+=\widetilde{\boldsymbol{c}}, 
\label{WHKAbra}
\end{split}\end{equation}
with the ${2\times2}$ kernel matrix
\begin{equation}\begin{split}
\mathbf{L}={{\mathcal{L}}}_{{k}}\begin{bmatrix}
1&{\lambda}^{{\mathtt{N}}}\\
{\lambda}^{{\mathtt{N}}}&1
\end{bmatrix}, 
\label{WHKkernel}
\end{split}\end{equation}
\begin{equation}\begin{split}
\text{and }
\boldsymbol{{\mathrm{v}}}^\pm=\begin{bmatrix}
{\mathrm{v}}_{0}^{ \pm}\\ {\mathrm{v}}_{{\mathtt{N}}}^{ \pm}\end{bmatrix}, \quad
\widetilde{\boldsymbol{c}}=\recip{{\mathbf{K}}}(\boldsymbol{q}^{{\mathrm{inc}}}{}^++\boldsymbol{p}^{\sgnM})
=(\mathbf{I}-\mathbf{L})(\boldsymbol{q}^{{\mathrm{inc}}}{}^++\boldsymbol{p}^{\sgnM}).
\end{split}\end{equation}
Here $\boldsymbol{p}^{\sgnM}$ is an unknown polynomial in ${z}^{-1}$ (resp. ${z}$) for $\sgnM=1$, i.e., ${\mathtt{M}}>0$ (resp. $\sgnM=-1$, i.e., ${\mathtt{M}}<0$) given by \eqref{ppiAbra}, \eqref{ppiNAbra}, and \eqref{qPAbra}. 
For example, the set of unknowns in case of ${\mathtt{M}}=-2$ {(resp. ${\mathtt{M}}=+2$)} has two elements, see Fig. \ref{Fig1_newunknowns1}(a) {(resp. Fig. \ref{Fig1_newunknowns1}(b))}.

\begin{remark}
As derived in \cite{Bls8staggerpair_asymp} (and \cite{GMthesis}), when the natural definition {(accounting for the ``shift'')} of Fourier transforms is considered related to the staggered edges, then the {{Wiener--Hopf}} kernel is found to contain a factor \eqref{myAbrakernel} in place of 
{$$\begin{bmatrix}
1&{\lambda}^{{\mathtt{N}}}\\
{\lambda}^{{\mathtt{N}}}&1
\end{bmatrix}$$}
in \eqref{WHKkernel}. According to the well known result \cite{SpeckPros}, its factorization falls in the category of unsolved problems. However, due to a alternate definition of Fourier transform {\eqref{discreteFT}} considered above, the problem can be reduced to an algebraic equation whose coefficients depend on certain scalar {{Wiener--Hopf}} factorization. {This is a succinct way to capture the essence of the technique of this paper that enables progress on a formidable Wiener--Hopf factorization problem \cite{SpeckPros,GohbergKaashoek}}.
\end{remark}

\begin{remark}
It is noteworthy that the zero offset case \cite{Bls8pair1} is a special case of the presented formulation in an elegant sense that $\boldsymbol{p}^{\sgnM}\equiv\boldsymbol{0}$.
\label{perturbcrack}
\end{remark}

\begin{figure}[h!]
\centering
\includegraphics[width=\textwidth]{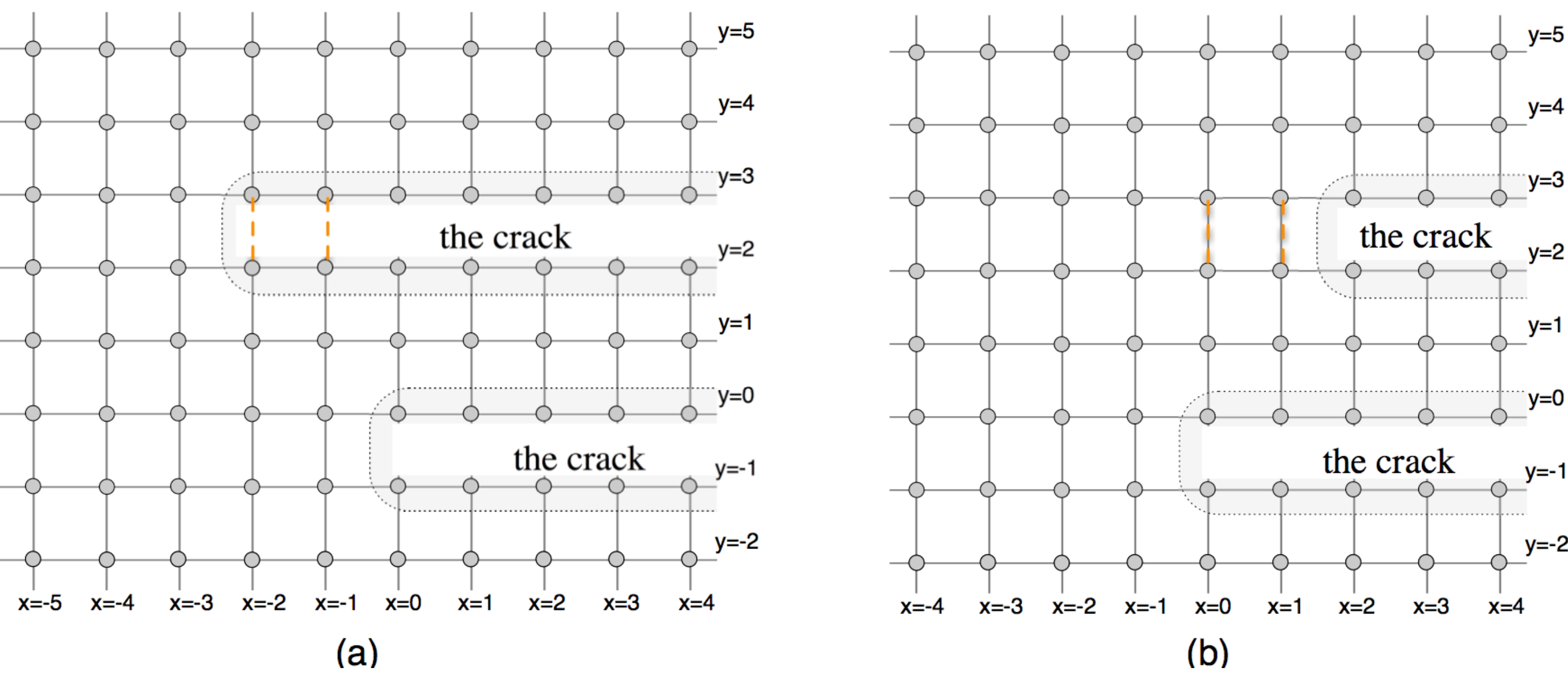}
\caption{
{
For the reduced (algebraic) problem, the unknowns are marked in orange, corresponding to Fig. \ref{Fig1}(a).}
(a) ${{\mathtt{M}}}=-2$ and (b) ${{\mathtt{M}}}=2$.
}
\label{Fig1_newunknowns1}
\end{figure}

\subsection{{Reduction of matrix Wiener--Hopf problem}}
\label{reductionK}
{It can be observed} that in \eqref{WHKkernel}, ${{\mathcal{L}}}_{{k}}={{{\mathcal{L}}}_{{k}{+}}}{{{\mathcal{L}}}_{{k}{-}}}$ with\footnote{\label{multfacsnote}{Here, for the multiplicative factorization of an admissible $f$, $f({z})=f_+({z})f_-({z})$ with $f_+({z})=\sum_{m\in\mathbb{Z}^+}f_m{z}^{-m}$ and $f_-({z})=\sum_{m\in\mathbb{Z}^-\cup\{0\}}f_m{z}^{-m}$; note that $m=0$ is in general present in the Fourier expansion of multiplicative factors. Typically, depending on $f$ a property $f_+(1/z)=f_-(z)$ can be also be incorporated so that $f_+(\infty)=f_-(\infty)$.}} \cite{Bls0} ${{\mathcal{L}}}_{{k}{\pm}}=\frac{{\mathtt{h}}_{\pm}}{{\mathtt{r}}_{\pm}}$.
Further, $\mathbf{L}$ stated in \eqref{WHKkernel} can be expressed as 
\begin{equation}\begin{split}
\mathbf{L}&=\recip{{\mathbf{J}}}{\mathbf{D}}{{\mathbf{J}}},\\
\text{where }
{{\mathbf{D}}}&={{{\mathcal{L}}}_{{k}}}\begin{bmatrix}
1-{{\lambda}}^{{\mathtt{N}}}&0\\
0&1+{{\lambda}}^{{\mathtt{N}}}
\end{bmatrix}=\begin{bmatrix}
{\upalpha}&0\\
0&{\upbeta}
\end{bmatrix}, \quad
{\mathbf{J}}=\frac{1}{\sqrt{2}}\begin{bmatrix}
1&{-}1\\
{+}1&1
\end{bmatrix}.
\end{split}\end{equation}
{It is also useful to consider ${{\mathbf{D}}}={{{\mathcal{L}}}_{{k}}}\tilde{{\mathbf{D}}}$}; the factorization of $\tilde{{\mathbf{D}}}$ is discussed in Appendix \ref{appDfacs}.
Thus,
{$\mathbf{L}=\mathbf{L}_-\mathbf{L}_+$ where}
(also recall Footnote \ref{invnote}) $\mathbf{L}_{-}={\mathbf{J}}_-{{\mathbf{D}}}_-, \mathbf{L}_{+}={{\mathbf{D}}}_+{\mathbf{J}}_+, $ i.e.,
\begin{subequations}\begin{eqnarray}
(\mathbf{L}_\pm)^{\pm1}&=&{({{\mathbf{D}}}_\pm)^{\pm1}({\mathbf{J}}_\pm)^{\pm1}=}({{\mathcal{L}}}_{k\pm})^{\pm1}(\tilde{{\mathbf{D}}}_\pm)^{\pm1}({\mathbf{J}}_\pm)^{\pm1},\label{Lfactors}\\
\text{with }
{{\mathbf{D}}}_\pm&=&({{\mathcal{L}}}_{{k}{\pm}})^{\pm1}\begin{bmatrix}
(1-{{\lambda}}^{{\mathtt{N}}})_\pm&0\\
0&(1+{{\lambda}}^{{\mathtt{N}}})_\pm
\end{bmatrix}=\begin{bmatrix}
{\upalpha}_\pm&0\\
0&{\upbeta}_\pm
\end{bmatrix},\label{Dfactors}\\
{\mathbf{J}}_+&=&{\mathbf{J}}=\frac{1}{\sqrt{2}}\begin{bmatrix}
1&-1\\
+1&1
\end{bmatrix}, \quad{\mathbf{J}}_-=\recip{{\mathbf{J}}}=\frac{1}{\sqrt{2}}\begin{bmatrix}
1&+1\\
-1&1
\end{bmatrix}\label{Efactors}.
\end{eqnarray}\label{Lfactorsfull}\end{subequations}
As $\mathbf{L}$ admits a multiplicative factorization $\mathbf{L}=\mathbf{L}{}_-\mathbf{L}{}_+,$ the {{Wiener--Hopf}} equation \eqref{WHKAbra} is re-written as
\begin{equation}\begin{split}
\recip{\mathbf{L}}_-\boldsymbol{{\mathrm{v}}}^-+\mathbf{L}_+\boldsymbol{{\mathrm{v}}}^+=\boldsymbol{c}{\,:=}(\recip{\mathbf{L}}_--\mathbf{L}_+)(\boldsymbol{q}^{{\mathrm{inc}}}{}^++\boldsymbol{p}^{\sgnM}), 
\end{split}\end{equation}
Therefore, by an application of the expressions \eqref{Lfactors},
\begin{equation}\begin{split}
\boldsymbol{c}
&=(\recip{{{\mathbf{D}}}}_-{\mathbf{J}}-{{\mathbf{D}}}_+{\mathbf{J}})(\boldsymbol{q}^{{\mathrm{inc}}}{}^++\boldsymbol{p}^{\sgnM}).
\label{cfunction}
\end{split}\end{equation}
Due to {(formal)} availability \cite{Bls8pair1} of the factorization of the kernel $\mathbf{L}$, the question regarding the {{Wiener--Hopf}} factorization is thus reduced to the {\em additive} factorization $\boldsymbol{c}=\boldsymbol{c}^++\boldsymbol{c}^-$ {according to \eqref{cfunction} and that $\boldsymbol{c}^+=\sum_{m\in\mathbb{Z}^+}c_m{z}^{-m}$ and $\boldsymbol{c}^-=\sum_{m\in\mathbb{Z}^-}c_m{z}^{-m}$}.
As a consequence of the standard application of Lioville's Theorem \cite{Ablowitz,Bls0},
finally, the solution can be written in terms of Fourier transforms {\eqref{discreteFT}} as
\begin{equation}\begin{split}
\recip{\mathbf{L}}_-\boldsymbol{{\mathrm{v}}}^-=\boldsymbol{c}^-, \quad  \mathbf{L}_+\boldsymbol{{\mathrm{v}}}^+=\boldsymbol{c}^+.
\label{WHsolK}
\end{split}\end{equation}
The additive factors {$\boldsymbol{c}^\pm$} of $\boldsymbol{c}$ are discussed in the following.

Substitution of {\eqref{Dfactors} and \eqref{Efactors}} in \eqref{cfunction} leads to, after certain natural re-arrangement,
\begin{equation}\begin{split}
\boldsymbol{c}&=\boldsymbol{c}^{\old}+\boldsymbol{c}^{\redc},\quad
\boldsymbol{c}^{\old}=[\cdot]\boldsymbol{q}^{{\mathrm{inc}}}{}^+, \quad
\boldsymbol{c}^{\redc}=[\cdot]\boldsymbol{p}^{\sgnM},\\
\text{where }
[\cdot]&=\frac{1}{\sqrt{2}}(\begin{bmatrix}
\recip{{\upalpha}}_-&-\recip{{\upalpha}}_-\\
\recip{{\upbeta}}_-&\recip{{\upbeta}}_-
\end{bmatrix}-\begin{bmatrix}
{\upalpha}_+&-{\upalpha}_+\\
{\upbeta}_+&{\upbeta}_+
\end{bmatrix}).
\label{ccincp}
\end{split}\end{equation}

{\bf Additive factorization of $\boldsymbol{c}^{\old}$:}
Note that according to \eqref{cfunction} {(${z}_{{P}}$ is defined in \eqref{zP})},
\begin{equation}\begin{split}
\boldsymbol{c}^{\old}({z})&=(\recip{{{\mathbf{D}}}}_-({z}){\mathbf{J}}-{(\recip{{{\mathbf{D}}}}_-({z}))}|_{{z}={z}_{{P}}}{\mathbf{J}}\\
&+{(\recip{{{\mathbf{D}}}}_-({z}))}|_{{z}={z}_{{P}}}{\mathbf{J}}-{{\mathbf{D}}}_+({z}){\mathbf{J}})\boldsymbol{q}^{{\mathrm{inc}}}{}^+({z})\\
&=\boldsymbol{c}^{\old}{}^-({z})+\boldsymbol{c}^{\old}{}^+({z}).
\label{cincfacs}
\end{split}\end{equation}
The last expression in \eqref{cincfacs} satisfies the requirements of the Cauchy's theorem \cite{Noble} and possess the desired behavior as ${z}\to0$ and $\infty$ in $\boldsymbol{c}^{\old}{}^-({z})$ and $\boldsymbol{c}^{\old}{}^+({z})$, respectively.
The additive factors of $\boldsymbol{c}^{\redc}$ are discussed in the following.

\begin{remark}
In view of the Remark \ref{perturbcrack}, it is noted that $\boldsymbol{c}^{\redc}$ describes the effect of a perturbation introduced by offset ${\mathtt{M}}.$ 
{A re-look at \eqref{WHKpre} confirms that the exact solution for $\mathtt{M}\ne0$ can be expressed as a superposition of that for the case $\mathtt{M}=0$ \cite{Bls8pair1} as it corresponds to the vanishing of $\boldsymbol{c}^{\redc}$, i.e., the result when the right hand side in the Wiener--Hopf equation involves only
\begin{equation}\begin{split}
\mathring{\boldsymbol{c}}:=\boldsymbol{c}^{\old}.
\label{caligned}
\end{split}\end{equation}
The exact nature of the effect of stagger is thus brought out by $\pertbn{\boldsymbol{c}}:=\boldsymbol{c}^{\redc}$ (i.e., ${\boldsymbol{c}}=\mathring{\boldsymbol{c}}+\pertbn{\boldsymbol{c}}$) which is studied after this remark.}
{By addition of $\pm$ functions in \eqref{WHsolK}, 
$$\boldsymbol{{\mathrm{v}}}^{\mathrm{F}}=\boldsymbol{{\mathrm{v}}}^-+\boldsymbol{{\mathrm{v}}}^+={\mathbf{L}_-}\boldsymbol{c}^-+\recip{\mathbf{L}}_+\boldsymbol{c}^+=\mathring{\boldsymbol{{\mathrm{v}}}}^{\mathrm{F}}+\pertbn{\boldsymbol{{\mathrm{v}}}}^{\mathrm{F}},
$$
$$
\text{where }
\mathring{\boldsymbol{{\mathrm{v}}}}^{\mathrm{F}}
={\mathbf{L}_-}\mathring{\boldsymbol{c}}^{-}+\recip{\mathbf{L}}_+\mathring{\boldsymbol{c}}^{+},
$$
\begin{equation}\begin{split}
\text{and }
\pertbn{\boldsymbol{{\mathrm{v}}}}^{\mathrm{F}}={\mathbf{L}_-}\boldsymbol{c}^{\redc-}+\recip{\mathbf{L}}_+\boldsymbol{c}^{\redc+}.
\label{vhatK}
\end{split}\end{equation}
}
\label{perturbcrack2}
\end{remark}

The formidable matrix Wiener-Hopf equation \cite{GMthesis,Bls8staggerpair_asymp} has been thus reduced to handling 
the occurrence of these two terms depend on the sign of ${\mathtt{M}}$:
\begin{equation}\begin{split}
{\recip{{\mathbf{D}}}_-({z})}{\mathbf{J}}\boldsymbol{p}^{+}({z})\text{ and }-{{\mathbf{D}}}_+({z}){\mathbf{J}}\boldsymbol{p}^{-}({z}).
\label{difficultK}
\end{split}\end{equation}
As shown below, the respective cases can be reduced to additive factorization of $f_-\poly{P}^+$ and $g_+\poly{Q}^-$ where $\poly{P}^+$ and ${z}^{-1}\poly{Q}^-$ are polynomials in ${z}^{-1}$ and ${z}$ of degree $|{\mathtt{M}}|-1$.
Naturally, it is sufficient to explore the terms of type $f_-{z}^{-m}$ and $g_+{z}^m$ for $m>0$.
Using the expansion of the functions $f_-$ and $g_+$ in their region of analyticity, for $m\ge0,$
\begin{subequations}\begin{eqnarray}
\phi_m({z})&=&f_-({z}){z}^{-m}=\sum\limits_{x\in{\mathbb{Z}^-}\cup\{0\}}f_x{z}^{-x}{z}^{-m}\notag\\
&=&\sum\limits_{x\in{\mathbb{Z}^-}}f_{x-m}{z}^{-x}+\sum\limits_{x=-m}^{0}f_x{z}^{-x-m}=\phi_{m}^{-}({z})+\phi_{m}^{+}({z}),\label{deffminus}\\
\Phi_{-m}({z})&=&F_+({z}){z}^m=\sum\limits_{x\in{\mathbb{Z}^+}}F_x{z}^{-x}{z}^{m}\notag\\
&=&\sum\limits_{x\in{\mathbb{Z}^+}}F_{x+m}{z}^{-x}+\sum\limits_{x=0}^{m-1}F_x{z}^{-x+m}=\Phi_{-m}^{+}({z})+\Phi_{-m}^{-}({z}).
\label{defFplus}
\end{eqnarray}\label{deffF}\end{subequations}
\begin{remark}
It is emphasized that $\phi_{m}^{+}$ (resp. $\Phi_{-m}^{-}$) involves a {{\em finite}} number of Fourier coefficients of $f_-$ (resp. $F_+$).
$\phi_{0}^{+}=f_-(0)=f_0, \phi_{1}^{+}=f_-(0)=f_0{z}^{-1}+f_{-1},$ and so on,
and
$\Phi_{0}^{-}=0, \Phi_{-1}^{-}=F_0{z}, \Phi_{-2}^{-}=F_0{z}^2+F_{1}{z},$ and so on.
\label{finiteremark1}
\end{remark}

Notice that $\{{{\mathrm{v}}}_{{{\mathtt{x}}},{\mathtt{N}}}\}_{{{\mathtt{x}}}=0}^{{\mathtt{M}}-1}$ or $\{{{\mathrm{v}}}_{{{\mathtt{x}}},{\mathtt{N}}}\}_{{{\mathtt{x}}}={\mathtt{M}}}^{-1}$ are unknowns. 
In other words, $\boldsymbol{p}^{\sgnM}$ is an unknown polynomial in ${z}^{-1}$ (resp. ${z}$) for $\sgnM=1$, i.e., ${\mathtt{M}}>0$ (resp. $\sgnM=-1$, i.e., ${\mathtt{M}}<0$) given by \eqref{ppiAbra} (resp. \eqref{ppiNAbra}). {The goal of the reduction technique of this paper is to obtain a finite set of linear algebraic equations that determine precisely these unknowns.}

For {an} illustration {of the final expressions and steps in an expanded form, for convenience}, consider the case ${\mathtt{M}}>0$, the details for the other case {(${\mathtt{M}}<0$)} are provided in Appendix \ref{AppcrackMN}. {As evident from Fig. \ref{Fig1_newunknowns1}(a) and Fig. \ref{Fig1_newunknowns1}(b), the difference between the two cases lies in the appearance of (shown orange colored links) the unknown bond lengths. ${\mathtt{M}}<0$ is associated with the occurrence of extra set of broken bonds while ${\mathtt{M}}>0$ leads to formation of intact bonds out of broken bonds.}

{\bf Additive factorization of $\boldsymbol{c}^{\redc}$:}
In accordance with the expressions provided in \eqref{ppiAbra}, let
\begin{equation}\begin{split}
\poly{P}^+({z})={{{\mathtt{f}}}}^{+}({z})+{{{\mathtt{f}}}}^{{\mathrm{inc}} {+}}({z})
=\sum\limits_{{\mathtt{x}}\in{{}\mathbb{D}}}{\mathrm{v}}^{{t}}_{{{\mathtt{x}}},{\mathtt{N}}}{z}^{-{\mathtt{x}}},
\end{split}\end{equation}
{where it is supposed that}
\begin{equation}\begin{split}
{\text{${{}\mathbb{D}}$ denotes the set $\mathbb{Z}_0^{{\mathtt{M}}-1}$.}}
\label{defnDcrack}
\end{split}\end{equation}
{It is natural to} recall the definition of $\boldsymbol{p}^{\sgnM}$ \eqref{qPAbra}, so that
{$$\boldsymbol{p}^{\sgnM}=\boldsymbol{p}^{+}=\begin{bmatrix}
0\\
-\poly{P}^+
\end{bmatrix}.$$}
After substitution of the detailed expression of the Wiener--Hopf kernel factorization \eqref{Lfactorsfull}, the first term in \eqref{difficultK} becomes
\begin{equation}\begin{split}
{\recip{{\mathbf{D}}}_-({z})}{\mathbf{J}}\boldsymbol{p}^{+}({z})
=\frac{1}{\sqrt{2}}\begin{bmatrix}
\recip{{\upalpha}}_-\poly{P}^+({z})\\
-\recip{{\upbeta}}_-\poly{P}^+({z})
\end{bmatrix}
\label{difficulttermMp1}
\end{split}\end{equation}
Using the splitting suggested in \eqref{deffminus},
\begin{equation}\begin{split}
\recip{{\upalpha}}_-\poly{P}^+&=f^-\poly{P}^+
=\sum\limits_{{\mathtt{x}}\in{{}\mathbb{D}}}{\mathrm{v}}^{{t}}_{{{\mathtt{x}}},{\mathtt{N}}}\phi_{{\mathtt{x}}}=\sum\limits_{{\mathtt{x}}\in{{}\mathbb{D}}}{\mathrm{v}}^{{t}}_{{{\mathtt{x}}},{\mathtt{N}}}(\phi_{{\mathtt{x}}}^{+}+\phi_{{\mathtt{x}}}^{-}),\\
\recip{{\upbeta}}_-\poly{P}^+&=g^-\poly{P}^+
=\sum\limits_{{\mathtt{x}}\in{{}\mathbb{D}}}{\mathrm{v}}^{{t}}_{{{\mathtt{x}}},{\mathtt{N}}}\psi_{{\mathtt{x}}}=\sum\limits_{{\mathtt{x}}\in{{}\mathbb{D}}}{\mathrm{v}}^{{t}}_{{{\mathtt{x}}},{\mathtt{N}}}(\psi_{{\mathtt{x}}}^{+}+\psi_{{\mathtt{x}}}^{-}).
\label{phipsicrack}
\end{split}\end{equation}
Hence, the additive factorization of $\boldsymbol{c}^{\redc}$ defined in \eqref{ccincp} follows as
\begin{equation}\begin{split}
{\sqrt{2}}\boldsymbol{c}^{\redc}&={\sqrt{2}}\boldsymbol{c}^{\redc}{}^-+{\sqrt{2}}\boldsymbol{c}^{\redc}{}^+\\
&=\begin{bmatrix}
\sum\limits_{{\mathtt{x}}\in{{}\mathbb{D}}}{\mathrm{v}}^{{t}}_{{{\mathtt{x}}},{\mathtt{N}}}\phi_{{\mathtt{x}}}^{-}\\
-\sum\limits_{{\mathtt{x}}\in{{}\mathbb{D}}}{\mathrm{v}}^{{t}}_{{{\mathtt{x}}},{\mathtt{N}}}\psi_{{\mathtt{x}}}^{-}\\
\end{bmatrix}\\
&+\begin{bmatrix}
\sum\limits_{{\mathtt{x}}\in{{}\mathbb{D}}}{\mathrm{v}}^{{t}}_{{{\mathtt{x}}},{\mathtt{N}}}\phi_{{\mathtt{x}}}^{+}\\
-\sum\limits_{{\mathtt{x}}\in{{}\mathbb{D}}}{\mathrm{v}}^{{t}}_{{{\mathtt{x}}},{\mathtt{N}}}\psi_{{\mathtt{x}}}^{+}
\end{bmatrix}
-\begin{bmatrix}
{\upalpha}_+&-{\upalpha}_+\\
{\upbeta}_+&{\upbeta}_+
\end{bmatrix}\boldsymbol{p}^+({z}).
\label{cpfacs}
\end{split}\end{equation}
The expressions of $\boldsymbol{c}^{\redc}{}^-$ and $\boldsymbol{c}^{\redc}{}^+$ can be read from \eqref{cpfacs} and can be easily seen to satisfy the requirements of the Cauchy's theorem \cite{Noble} and possess the desired behavior as ${z}\to0$ and $\infty$, respectively.

As the last statement of closure of the problem for the case ${\mathtt{M}}>0,$ the polynomial $\boldsymbol{p}^+({z})$ can be determined as follows.

{\bf Equation for $\{{\mathrm{v}}_{{{\mathtt{x}}},{\mathtt{N}}}\}_{{\mathtt{x}}\in{{}\mathbb{D}}}$:}
Let $\mathfrak{P}_{{}\mathbb{D}}$ denote the projection of Fourier coefficients of a typical $f^+(z)$ for $|z|>{\mathrm{R}}_+$ to the set ${{}\mathbb{D}}$ {(defined in \eqref{defnDcrack})}. Thus,
\begin{equation}\begin{split}
\mathfrak{P}_{{}\mathbb{D}}(f^+)=\sum\limits_{{\mathtt{x}}\in{{}\mathbb{D}}}f_{{\mathtt{x}}}{z}^{-{\mathtt{x}}}, \forall f^+=\sum\limits_{{\mathtt{x}}\in{\mathbb{Z}^+}}f_{{\mathtt{x}}}{z}^{-{\mathtt{x}}}, |{z}|>{\mathrm{R}}_+.
\label{defdomainproj}
\end{split}\end{equation}
{It is clear that $\mathfrak{P}_{{}\mathbb{D}}$ is a linear operator.}

Using \eqref{WHsolK}${}_2$, the set of first ${\mathtt{M}}$ Fourier coefficients of the second component of $\boldsymbol{{\mathrm{v}}}^+$, i.e., {(recall \eqref{Lfactors} for the expression of ${\mathbf{L}^+}$)}
\begin{equation}\begin{split}
{\mathfrak{P}_{{}\mathbb{D}}}(\ensuremath{\hat{\mathbf{e}}}_2\cdot\boldsymbol{{\mathrm{v}}}^+)={\mathfrak{P}_{{}\mathbb{D}}}(\ensuremath{\hat{\mathbf{e}}}_2\cdot\recip{\mathbf{L}}_+\boldsymbol{c}^+)={\mathfrak{P}_{{}\mathbb{D}}}(\ensuremath{\hat{\mathbf{e}}}_2\cdot\recip{\mathbf{L}}_+(\boldsymbol{c}^{\old}{}^++\boldsymbol{c}^{\redc}{}^+))
\label{reducedcrackP}
\end{split}\end{equation}
yields a ${\mathtt{M}}\times{\mathtt{M}}$ coefficient matrix for the total number of ${\mathtt{M}}$ unknowns, i.e., $\{{\mathrm{v}}_{{{\mathtt{x}}},{\mathtt{N}}}\}_{{\mathtt{x}}\in{{}\mathbb{D}}}$ {(recall \eqref{defnDcrack} giving the definition of ${{}\mathbb{D}}$)}. The equation \eqref{reducedcrackP} is {\em the reduced algebraic problem} and it can be observed that its coefficients can be found by using only the scalar {{Wiener--Hopf}} factorization \cite{Noble}.

Using \eqref{cincfacs} and \eqref{cpfacs}, with 
{
\begin{equation}\begin{split}
\boldsymbol{\tau}^+=\frac{1}{\sqrt{2}}\begin{bmatrix}
\sum\limits_{{\mathtt{x}}\in{{}\mathbb{D}}}{\mathrm{v}}^{{t}}_{{{\mathtt{x}}},{\mathtt{N}}}\phi_{{\mathtt{x}}}^{+}\\
-\sum\limits_{{\mathtt{x}}\in{{}\mathbb{D}}}{\mathrm{v}}^{{t}}_{{{\mathtt{x}}},{\mathtt{N}}}\psi_{{\mathtt{x}}}^{+}
\end{bmatrix},
\label{taucrackP1}
\end{split}\end{equation}
}
it is found that {(i.e., $\recip{\mathbf{L}}_+\boldsymbol{c}^+$ in \eqref{reducedcrackP} equals)}
\begin{equation}\begin{split}
{\recip{\mathbf{L}}_+(\boldsymbol{c}^{\old}{}^++\boldsymbol{c}^{\redc}{}^+)}
&=(\recip{{\mathbf{J}}}\recip{{{\mathbf{D}}}}_+\recip{{{\mathbf{D}}}}_-({z}_{{P}}){\mathbf{J}}-\mathbf{I})\boldsymbol{q}^{{\mathrm{inc}}}{}^+
+\recip{{\mathbf{J}}}\recip{{{\mathbf{D}}}}_+{\boldsymbol{\tau}}^+-\boldsymbol{p}^+.
\end{split}\end{equation}
{In above, ${z}_{{P}}$ is defined by \eqref{zP}.}
Indeed, using above and expanding and re-arranging the terms in \eqref{reducedcrackP} further,
\begin{subequations}\begin{equation}\begin{split}
&{\mathfrak{P}_{{}\mathbb{D}}}({{\mathrm{v}}}_{{\mathtt{N}}}^+({z})+{{\mathrm{v}}}^{{\mathrm{inc}}}_{{\mathtt{N}}}{}^+({z})-{{{\mathtt{f}}}}^+({z})-{{{\mathtt{f}}}}^{{\mathrm{inc}}}{}^+({z}))\\
&={\mathfrak{P}_{{}\mathbb{D}}}({\frac{1}{2}}\mathcal{F}^{{\mathrm{inc}}}({z})-{\frac{1}{2}}\sum\limits_{{\mathtt{x}}\in{{}\mathbb{D}}}{\mathrm{v}}^{{t}}_{{{\mathtt{x}}},{\mathtt{N}}}\mathcal{A}_{{\mathtt{x}}}({z})),
\end{split}\end{equation}
\begin{eqnarray}
\text{where }
\mathcal{A}_{{\mathtt{x}}}({z})&{\,:=}&\phi_{{\mathtt{x}}}^{+}({z})\recip{{\upalpha}}_+({z})+\psi_{{\mathtt{x}}}^{+}({z})\recip{{\upbeta}}_+({z}),
\label{Axcrack}\\
\text{and }
\mathcal{F}^{{\mathrm{inc}}}({z})&{\,:=}&\big(- (1-e^{i\upkappa_y{\mathtt{N}}})\recip{{\upalpha}}_-({z}_{{P}})\recip{{\upalpha}}_+({z})\notag\\
&&+ (1+e^{i\upkappa_y{\mathtt{N}}})\recip{{\upbeta}}_-({z}_{{P}})\recip{{\upbeta}}_+({z})\big)e^{-i\upkappa_y{\mathtt{N}}}{{\mathrm{v}}}_{{\mathtt{N}}}^{{\mathrm{inc}}}{}^+({z}).
\label{Finccrack}
\end{eqnarray}
\end{subequations}

In view of the definitions of ${{{\mathtt{f}}}}^+$ and ${{{\mathtt{f}}}}^{{\mathrm{inc}}}{}^+$ given by \eqref{ppiAbra}, {above} equation leads to
\begin{equation}\begin{split}
\sum\limits_{{\mathtt{x}}\in{{}\mathbb{D}}}{\mathrm{v}}^{{t}}_{{\mathtt{x}},{\mathtt{N}}}\mathfrak{P}_{{}\mathbb{D}}(\mathcal{A}_{{\mathtt{x}}})
=\mathfrak{P}_{{}\mathbb{D}}\big(\mathcal{F}^{{\mathrm{inc}}}\big),
\label{eqnMbyMcrack}
\end{split}\end{equation}
which {is} a ${\mathtt{M}}\times{\mathtt{M}}$ system of linear algebraic equations for $\{{\mathrm{v}}^{{t}}_{{\mathtt{x}},{\mathtt{N}}}\}_{{\mathtt{x}}\in{{}\mathbb{D}}}$, i.e., 
$\{{\mathrm{v}}_{{\mathtt{x}},{\mathtt{N}}}\}_{{{}\mathbb{D}}}$ since $\{{\mathrm{v}}^{{\mathrm{inc}}}_{{\mathtt{x}},{\mathtt{N}}}\}_{{{}\mathbb{D}}}$ are known.
{\eqref{eqnMbyMcrack} can be also conveniently written in terms of the familiar matrix-column format.}
Indeed, with the notation $\mathfrak{C}_{{{\mu}}}(p)$ to denote the coefficient of $z^{-{{\mu}}}$ for polynomials $p$ of the form {$$\mathfrak{C}_0+\mathfrak{C}_1{z}^{-1}+\mathfrak{C}_2{z}^{-2}+\dotsc,$$} 
it is easy to see that
\begin{subequations}
\begin{equation}\begin{split}
{\sum_{\nu=1}^{{\mathtt{M}}}}a_{{{\mu}}\nu}\chi_\nu=b_{{\mu}}\quad\quad({{\mu}}=1, \dotsc, {\mathtt{M}})
\label{aknueqnPcrack}
\end{split}\end{equation}
where {(for ${{\mu}}, \nu=1, \dotsc, {\mathtt{M}}$)}
\begin{equation}\begin{split}
\begin{aligned}
a_{{{\mu}}\nu}&=\mathfrak{C}_{{{\mu}}-1}\big(\mathfrak{P}_{{}\mathbb{D}}(\mathcal{A}_{\nu-1})\big), \quad
\chi_\nu&={\mathrm{v}}^{{t}}_{\nu-1,{\mathtt{N}}}, \quad
b_{{\mu}}&=\mathfrak{C}_{{{\mu}}-1}(\mathfrak{P}_{{}\mathbb{D}}\big(\mathcal{F}^{{\mathrm{inc}}}\big)).
\end{aligned}
\end{split}\end{equation}
\end{subequations}
Let ${\recip{a}_{\nu{{\mu}}}}$ denote the components of the inverse of $[a_{{{\mu}}\nu}]_{{{\mu}}, \nu=1, \dotsc, {\mathtt{M}}}$. Then
\begin{equation}\begin{split}
{\mathrm{v}}^{{t}}_{{\mathtt{x}},{\mathtt{N}}}={\sum_{{{\mu}}=1}^{{\mathtt{M}}}}{\recip{a}_{({\mathtt{x}}+1){{\mu}}}}\mathfrak{C}_{{{\mu}}-1}(\mathfrak{P}_{{}\mathbb{D}}\big(\mathcal{F}^{{\mathrm{inc}}}\big)), {\mathtt{x}}\in{{}\mathbb{D}}.
\label{exactvmsol}
\end{split}\end{equation}
In view of Remark \ref{finiteremark1}, the reduced problem \eqref{eqnMbyMcrack} 
{needs} an evaluation of only a {{\em finite}} number of Fourier coefficients  {beyond the $\mathtt{M}=0$ case (due to appearance of $\phi_{{\mathtt{x}}}^{+}$ and $\psi_{{\mathtt{x}}}^{+}$ in $\mathcal{A}_{{\mathtt{x}}}$)}.

\section{{Pair of rigid constraints}}
\label{WHeqformC}
{In this section, the discrete scattering problem associated with a pair of staggered rigid constraints is analyzed within a Wiener--Hopf formulation by pursuing closely the symbolic manipulations and definitions of the previous section. The scatterer (with discrete Dirichlet condition) is located at the sites stated in \eqref{defectC}, shown schematically in Fig. \ref{Fig1}. 
The details of the equation of motion in the intact lattice and the incident wave are same as those provided earlier in \S\ref{sqLattFT}. Moreover, the application of the Fourier transform \eqref{discreteFT} brings out the specification of the exact equations for the unknown functions as described in the statements immediately following \eqref{ubulkC}}.

{On the lattice rows ${\mathtt{y}}=0$ (referred as the `lower' rigid constraint sometimes) and ${\mathtt{y}}={\mathtt{N}}$ (referred as the `upper' rigid constraint sometimes), as schematically shown in Fig. \ref{Fig1}(b), the sites with ${\mathtt{x}}\in{\mathbb{Z}^+}$ and ${\mathtt{x}}-{\mathtt{M}}\in{\mathbb{Z}^+}$ (recall the definition ${\mathbb{Z}^+}$ of in \eqref{ZpZn}) are restrained to admit only zero value of total displacement, i.e.,
\begin{equation}\begin{split}
{\su}^{t}_{{{\mathtt{x}}}, {{\mathtt{y}}}}=0.
\label{rigidconstraintdef}
\end{split}\end{equation}
{Combining the discrete Helmholtz equation \eqref{dHelmholtz} for ${{\mathtt{x}}}\in{\mathbb{Z}^-}$ (recall the definition ${\mathbb{Z}^-}$ of in \eqref{ZpZn}) and the Dirichlet condition \eqref{rigidconstraintdef}, i.e., ${\su}^{t}_{{{\mathtt{x}}}, 0}=0,$ for ${{\mathtt{x}}}\in{\mathbb{Z}^+}$ in the same equation at ${{\mathtt{y}}}=0$, it is found that the relevant}
equation of motion is
\begin{equation}\begin{split}
-{\upomega}^2{\su}_{{{\mathtt{x}}}, 0}={\triangle}{\su}_{{{\mathtt{x}}}, {\mathtt{y}}}|_{{\mathtt{y}}=0}{{\mathcal{H}}}(-{\mathtt{x}}-1)+{\upomega}^2 {\su}^{{\mathrm{inc}}}_{{{\mathtt{x}}}, {0}}{{\mathcal{H}}}({\mathtt{x}}),
\label{rigideq1}
\end{split}\end{equation}
{where the notation $|_{{\mathtt{y}}=0}$ has been used because of the presence of shift operators in the discrete Laplacian \eqref{dimnewtoneq}${}_2$.
To obtain above equation \eqref{rigideq1} from \eqref{dHelmholtz} and \eqref{rigidconstraintdef}, recall} that the total displacement ${\su}^{{t}}$ of an arbitrary particle in the lattice ${\mathfrak{S}}$ is a sum of the incident wave displacement ${\su}^{{\mathrm{inc}}}$ \eqref{uinc} and the scattered wave displacement ${\su}$ ({including} the reflected waves), {i.e., \eqref{utsplit} holds},
on the other hand the incident wave satisfies \eqref{dHelmholtz} due to the dispersion relation \eqref{dispersion}.}

{By an application of} the Fourier transform \eqref{discreteFT} {to \eqref{rigideq1}, using the elementary shift property of the half-infinite Fourier transform, it is found that}
$-{\upomega}^2{\su}_{0}^{{\mathrm{F}}}={\upomega}^2 {{\su}^{{\mathrm{inc}}}_{0}{}^+}+{{z}}({\su}_{0}^{-}+{\su}_{0, 0})+{{z}}^{-1}({\su}_{0}^{-}-{{z}} {\su}_{-1, 0})+{\su}_{1}^{-}+{\su}_{-1}^{-}-4{\su}_{0}^{-},$ {is the equation satisfied by the scattered field at ${\mathtt{y}}=0$.}
{This equation can be re-arranged (using \eqref{delDp} and \eqref{zP} to express ${\su}^{{\mathrm{inc}}}_{0}{}^+$) so that \eqref{rigideq1} is equivalently stated as
\begin{equation}\begin{split}
{{\mathtt{Q}}}{\su}_0^{-}-{\upomega}^2{\su}_0^{+}={{\saux}_0+{\su}_{1}^{-}+{\su}_{-1}^{-}}+{\upomega}^2{\su}^{{\mathrm{inc}}}_{0,0}\delta_{D}^{+}({{z}} {z}_{{P}}^{-1}),
\label{preu0expn}
\end{split}\end{equation}
\begin{equation}\begin{split}
\text{where }
{\saux}_0{({{z}}):=}-{\su}_{-1, 0}+{{z}} {\su}_{0, 0},\\
{\text{and }{{\mathtt{Q}}}({{z}}){\,:=}4-{{z}}-{{z}}^{-1}-{\upomega}^2.}
\label{defW0Q}
\end{split}\end{equation}
}
{As a result of the assumption \eqref{complexfreq}, it is found that ${\mathtt{Q}}$ does not have any zeros on the annulus ${{\mathscr{A}}}$ (described in Appendix \ref{applyFT}, see also \cite{Bls1}). Thus, recognizing the rigid constraint \eqref{rigidconstraintdef} at ${\mathtt{x}}\in{\mathbb{Z}^+}, {{\mathtt{y}}}=0$, that is ${{\su}_{0}^+}+{{\su}^{{\mathrm{inc}}}_{0}{}^+}=0$, \eqref{rigideq1} leads to the final form}
\begin{equation}\begin{split}
{\su}_0^{{\mathrm{F}}}{({{z}})}=\frac{{\saux}_0{({{z}})}+{\su}_{1}^{-}{({{z}})}+{\su}_{-1}^{-}{({{z}})}}{{\mathtt{Q}}{({{z}})}}-{{\su}^{{\mathrm{inc}}}_{0,0}}\delta_{D}^{+}({{z}} {z}_{{P}}^{-1}).
\label{u0expn}
\end{split}\end{equation}

\begin{remark}
Note that ${\su}_{-1, 0}$ is an unknown complex number while ${\su}_{0, 0}$ is known {(and equals $-{\su}_{0, 0}^{{\mathrm{inc}}}$ due to the `lower' rigid constraint ${\su}^{t}_{0, 0}=0$, to be substituted henceforth)}. 
{In particular, in \eqref{defW0Q}, ${\saux}_0({{z}})=-{\su}_{-1, 0}-{{z}} {\su}^{{\mathrm{inc}}}_{0, 0}.$}
\label{u00remark}
\end{remark}
{In terms of the definition of ${{\mathtt{H}}}$ provided in \eqref{h2} as part of Appendix \ref{applyFT}, it is a mild observation that ${{\mathtt{Q}}}={{\mathtt{H}}}+2$.}

For the row ${{\mathtt{y}}}={\mathtt{N}}$, {the horizontal stagger or offset is} ${{\mathtt{M}}}\in\mathbb{Z}$, {i.e. an integer}. {In this paper, the interesting case concerns those configurations when the offset ${\mathtt{M}}$ is negative or positive, since the zero offset, i.e.,} ${\mathtt{M}}=0$, is the case of aligned parallel constraint edges which {appears in a separate exposition} \cite{Bls8pair1}. 

\subsection{{${\mathtt{M}}>0$}}
\label{slitMP}
For the row ${{\mathtt{y}}}={\mathtt{N}}$, {using the schematic} shown in Fig. \ref{Fig1}(b),
{it is clear that the discrete Helmholtz equation \eqref{dHelmholtz} holds for ${{\mathtt{x}}}-{\mathtt{M}}\in{\mathbb{Z}^-}$ whereas the Dirichlet condition ${\su}^{t}_{{{\mathtt{x}}}, 0}=0$ holds for ${{\mathtt{x}}}-{\mathtt{M}}\in{\mathbb{Z}^+}$.
In terms of the splitting \eqref{utsplit} of the total field between the incident and scattered components, it is found that, as a counterpart of \eqref{rigideq1},}
the equation satisfied is
{
\begin{equation}\begin{split}
-{\upomega}^2{\su}_{{{\mathtt{x}}}, {\mathtt{N}}}={\triangle}{\su}_{{{\mathtt{x}}}, {\mathtt{y}}}|_{{\mathtt{y}}={\mathtt{N}}}{{\mathcal{H}}}(-{\mathtt{x}}+{\mathtt{M}}-1)+{\upomega}^2 {\su}^{{\mathrm{inc}}}_{{{\mathtt{x}}}, {{\mathtt{N}}}}{{\mathcal{H}}}({\mathtt{x}}-{\mathtt{M}}).
\label{rigideq2}
\end{split}\end{equation}
By adding and subtracting certain terms (in analogy with the case of cracks \eqref{crackt2}), \eqref{rigideq2} can be re-arranged so that it captures the same form rigid constraint as \eqref{rigideq1} but now for $\mathtt{y}={\mathtt{N}}$ (modulo the natural appearance of certain source terms in square brackets)}
\begin{equation}\begin{split}
-{\upomega}^2{\su}_{{{\mathtt{x}}}, {{\mathtt{N}}}}
&={\triangle}{\su}_{{{\mathtt{x}}}, {\mathtt{y}}}{|_{{\mathtt{y}}={\mathtt{N}}}}{{\mathcal{H}}}(-{\mathtt{x}}-1)+{\upomega}^2 {\su}^{{\mathrm{inc}}}_{{{\mathtt{x}}}, {{\mathtt{N}}}}{{\mathcal{H}}}({\mathtt{x}})\\
&+[({\triangle}{\su}_{{{\mathtt{x}}}, {\mathtt{y}}}{|_{{\mathtt{y}}={\mathtt{N}}}}+{\upomega}^2{\su}_{{{\mathtt{x}}}, {\mathtt{N}}}
-{\upomega}^2{\su}_{{{\mathtt{x}}}, {\mathtt{N}}}-{\upomega}^2 {\su}^{{\mathrm{inc}}}_{{{\mathtt{x}}}, {{\mathtt{N}}}}){{\mathcal{H}}}({\mathtt{M}}-1-{\mathtt{x}}){{\mathcal{H}}}({\mathtt{x}})].
\label{geneqPy0}
\end{split}\end{equation}
{Indeed, for ${\mathtt{x}}-{{\mathtt{M}}}\in{\mathbb{Z}^-},$ \eqref{geneqPy0} (evidently, identical to \eqref{rigideq2}) reduces to the discrete Helmholtz equation (with ${\su}$ replacing ${\su}^{{t}}$ in \eqref{dHelmholtz}) as in \eqref{rigideq2}.
In this manner, the equation of motion for the rows with rigid constraints have been written in terms of zero offset, i.e., `aligned' tips (the corresponding scattering problem is graced with the availability of an exact solution \cite{Bls8pair1}). This manner of re-writing \eqref{rigideq2} plays an important role in carrying out the reduction technique which is stated in the title of this paper.}

{
Catering to the need of convenience, due to repetition of similar expressions, in the remainder of this section, let 
\begin{equation}\begin{split}
{\mathtt{B}}_{\mathtt{x}}(m, n){\,:=}({\triangle}{\su}_{{{\mathtt{x}}}, {\mathtt{y}}}|_{{\mathtt{y}}={\mathtt{N}}}+{\upomega}^2{\su}_{{{\mathtt{x}}}, {\mathtt{N}}}){{\mathcal{H}}}({\mathtt{x}}-m){{\mathcal{H}}}(n-{{\mathtt{x}}}), {{\mathtt{x}}}\in\mathbb{Z}.
\label{anidentityB}
\end{split}\end{equation}
Evidently,
the sum of first and second term inside the square brackets of \eqref{geneqPy0} becomes ${\mathtt{B}}_{\mathtt{x}}(0, {\mathtt{M}}-1)$;
in view of the later occurrence 
as well, the following is established as an identity.}
\begin{claim}
{The Fourier transform \eqref{discreteFT} of $\{{\mathtt{B}}_{\mathtt{x}}(m, n)\}_{{\mathtt{x}}\in\mathbb{Z}}$, i.e., ${\mathtt{B}}^{{\mathrm{F}}}(m, n;{z})=\sum\limits_{{{\mathtt{x}}}\in\mathbb{Z}}{{z}}^{-{{\mathtt{x}}}}{\mathtt{B}}_{\mathtt{x}}(m, n),$
is given by
\begin{equation}\begin{split}
{\mathtt{B}}^{{\mathrm{F}}}(m, n;{z})&=-{\mathtt{Q}}{({{z}})} \sum\limits_{{{\mathtt{x}}}\in\mathbb{Z}_m^n}{{z}}^{-{{\mathtt{x}}}}{\su}_{{{\mathtt{x}}}, {\mathtt{N}}}
+\sum\limits_{{{\mathtt{x}}}\in\mathbb{Z}_m^n}{{z}}^{-{{\mathtt{x}}}}({\su}_{{{\mathtt{x}}}, {\mathtt{N}}+1}+{\su}_{{{\mathtt{x}}}, {\mathtt{N}}-1})\\
&
+{z}^{-m}(-{z}\su_{m, {\mathtt{N}}}+\su_{m-1, {\mathtt{N}}})+{z}^{-n}(\su_{n+1, {\mathtt{N}}}-{z}^{-1}\su_{n, {\mathtt{N}}})
,
\label{anidentitygen}
\end{split}\end{equation}
where ${\mathtt{Q}}$ is defined by \eqref{defW0Q}${}_2$.}
\label{claimone}
\end{claim}
\begin{proof}
{
Using \eqref{anidentityB} and \eqref{dimnewtoneq}${}_2$,
\begin{equation}\begin{split}
\sum\limits_{{{\mathtt{x}}}\in\mathbb{Z}}{{z}}^{-{{\mathtt{x}}}}{\mathtt{B}}_{\mathtt{x}}(m, n)&=\sum\limits_{{{\mathtt{x}}}\in\mathbb{Z}_m^n}{{z}}^{-{{\mathtt{x}}}}({\triangle}{\su}_{{{\mathtt{x}}}, {\mathtt{y}}}|_{{\mathtt{y}}={\mathtt{N}}}+{\upomega}^2{\su}_{{{\mathtt{x}}}, {\mathtt{N}}})\\
&=\sum\limits_{{{\mathtt{x}}}\in\mathbb{Z}_m^n}{{z}}^{-{{\mathtt{x}}}}({\su}_{{{\mathtt{x}}}+1, {\mathtt{N}}}+{\su}_{{{\mathtt{x}}}-1, {\mathtt{N}}}+({\upomega}^2-4){\su}_{{{\mathtt{x}}}, {\mathtt{N}}})\\
&+\sum\limits_{{{\mathtt{x}}}\in\mathbb{Z}_m^n}{{z}}^{-{{\mathtt{x}}}}({\su}_{{{\mathtt{x}}}, {\mathtt{N}}+1}+{\su}_{{{\mathtt{x}}}, {\mathtt{N}}-1}),
\label{claimeq}
\end{split}\end{equation}
where the first and second term in the first sum (i.e., in \eqref{claimeq}${}_3$) can be re-written as 
$$\sum\limits_{{{\mathtt{x}}}\in\mathbb{Z}_m^n}{{z}}^{-{{\mathtt{x}}}}{\su}_{{{\mathtt{x}}}+1, {\mathtt{N}}}={{z}}\sum\limits_{{{\mathtt{x}}}\in\mathbb{Z}_m^n}{{z}}^{-{{\mathtt{x}}}}{\su}_{{{\mathtt{x}}}, {\mathtt{N}}}-{z}^{-m+1}{\su}_{m, {\mathtt{N}}}+{z}^{-n}{\su}_{n+1, {\mathtt{N}}}
$$ and
$$\sum\limits_{{{\mathtt{x}}}\in\mathbb{Z}_m^n}{{z}}^{-{{\mathtt{x}}}}{\su}_{{{\mathtt{x}}}-1, {\mathtt{N}}}={{z}}\sum\limits_{{{\mathtt{x}}}\in\mathbb{Z}_m^n}{{z}}^{-{{\mathtt{x}}}}{\su}_{{{\mathtt{x}}}, {\mathtt{N}}}+{z}^{-m}{\su}_{m-1, {\mathtt{N}}}-{z}^{-n-1}{\su}_{n, {\mathtt{N}}}.$$
Using the definition of ${\mathtt{Q}}$ stated in \eqref{defW0Q}${}_2$, the claim follows.
}
\end{proof}
{By inspection of the terms on the right hand side of \eqref{anidentitygen} for ${\mathtt{B}}_{\mathtt{x}}(0, {\mathtt{M}}-1)$, for brevity, it is useful \cite{Bls1,Bls3} to consider the notation},
\begin{equation}\begin{split}
{\thicktilde{\saux}}_{{\mathtt{N}}}{({z})}{\,:=}-{\su}_{-1, {\mathtt{N}}}+{{z}} {\su}_{0, {\mathtt{N}}},\quad
\text{and }
{\widehat{\saux}}_{{\mathtt{N}}}{({z})}{\,:=}-{\su}_{{\mathtt{M}}-1, {\mathtt{N}}}+{{z}} {\su}_{{\mathtt{M}}, {\mathtt{N}}},
\label{WhP}
\end{split}\end{equation}
{so that, ${\mathtt{B}}^{{\mathrm{F}}}(0, {\mathtt{M}}-1;{z})$ can be expressed as}
\begin{subequations}
\begin{equation}\begin{split}
{{{\mathtt{B}}^{{\mathrm{F}}}(0, {\mathtt{M}}-1;{z})}=-{\mathtt{Q}}({z}) {{\mathtt{g}}}^{+}({z})-{\thicktilde{\saux}}_{{\mathtt{N}}}({z})+{z}^{-{\mathtt{M}}}{\widehat{\saux}}_{{\mathtt{N}}}({z})+{{{\mathtt{f}}}}^{+}({z}),}
\label{anidentity1}\end{split}\end{equation}
\begin{equation}\begin{split}
\text{where }
{{\mathtt{g}}}^{+}({z})&{\,:=}+\sum\limits_{{{\mathtt{x}}}\in\mathbb{Z}_0^{{\mathtt{M}}-1}}{{z}}^{-{{\mathtt{x}}}}{\su}_{{{\mathtt{x}}}, {{\mathtt{N}}}},\\
{{{\mathtt{f}}}}^{+}({z})&{\,:=}+\sum\limits_{{{\mathtt{x}}}\in\mathbb{Z}_0^{{\mathtt{M}}-1}}{{z}}^{-{{\mathtt{x}}}}({\su}_{{{\mathtt{x}}}, {\mathtt{N}}+1}+{\su}_{{{\mathtt{x}}}, {\mathtt{N}}-1}).
\label{defp10}
\end{split}\end{equation}
\end{subequations}
{The notation $\mathtt{g}^{+}$ and $\mathtt{f}^{+}$ has been adopted as $\mathtt{M}>0$ implies that these complex functions are polynomials in the variable ${z}^{-1}$; recall again the statement following \eqref{discreteFT}}.
\begin{remark}
{Similar to the observation stated below \eqref{defW0Q}, note that} ${\su}_{{\mathtt{M}}-1, {\mathtt{N}}}$ is an unknown complex number while ${\su}_{{\mathtt{M}}, {\mathtt{N}}}$ is known {(and equals $-{\su}_{{\mathtt{M}}, {\mathtt{N}}}^{{\mathrm{inc}}}$ by virtue of the `upper' rigid constraint ${\su}^{t}_{{\mathtt{M}}, {\mathtt{N}}}=0$, to be substituted henceforth)}.
{In particular, in \eqref{WhP}, $\widehat{\saux}_{\mathtt{N}}(z)=-{\su}_{{\mathtt{M}}-1, {\mathtt{N}}}-{{z}} {\su}^{{\mathrm{inc}}}_{{\mathtt{M}}, {\mathtt{N}}}.$}
\label{uMNremark}
\end{remark}

{
Employing the Fourier transform \eqref{discreteFT} to \eqref{geneqPy0}, in view of the identity \eqref{anidentity1}, it is found that (with a resemblance to \eqref{preu0expn})
\begin{equation}\begin{split}
{{\mathtt{Q}}}{\su}_{\mathtt{N}}^{-}-{\upomega}^2{\su}_{\mathtt{N}}^{+}&={\thicktilde{\saux}}_{\mathtt{N}}+{\su}_{{\mathtt{N}}+1}^{-}+{\su}_{{\mathtt{N}}-1}^{-}+{\upomega}^2{\su}^{{\mathrm{inc}}}_{0, {\mathtt{N}}}\delta_{D}^{+}({{z}} {z}_{{P}}^{-1})\\
&+{{\mathtt{B}}^{{\mathrm{F}}}(0, {\mathtt{M}}-1;{z})}-{\upomega}^2({{\mathtt{g}}}^{+}({z})+{{\mathtt{g}}}^{{{\mathrm{inc}}}+}({z})),
\label{geneqPy0FT1}
\end{split}\end{equation}
\begin{equation}\begin{split}
\text{where }
{{\mathtt{g}}}^{{{\mathrm{inc}}}+}({z}){\,:=}+\sum\limits_{{{\mathtt{x}}}\in\mathbb{Z}_0^{{\mathtt{M}}-1}}{{z}}^{-{{\mathtt{x}}}}{\su}^{{\mathrm{inc}}}_{{{\mathtt{x}}}, {{\mathtt{N}}}}.\label{defq1inc}
\end{split}\end{equation}
Using the Fourier transform \eqref{discreteFT}, the part corresponding to ${\mathtt{x}}\in{\mathbb{Z}^+}$, 
in conjunction with the rigid constraint \eqref{rigidconstraintdef} for ${\mathtt{x}}-{\mathtt{M}}\in{\mathbb{Z}^+}$, i.e., ${\su}_{{{\mathtt{x}}}, {{\mathtt{N}}}}+{\su}^{{\mathrm{inc}}}_{{{\mathtt{x}}}, {{\mathtt{N}}}}=0$ (the same also is part of \eqref{geneqPy0}), leads to the equation
\begin{equation}\begin{split}
{\su}_{{{\mathtt{N}}}}^{+}({z})=\sum\limits_{{{\mathtt{x}}}\in\mathbb{Z}^+}{z}^{-{\mathtt{x}}}{\su}_{{\mathtt{x}},{{\mathtt{N}}}}=-{\su}^{{\mathrm{inc}}}_{0, {\mathtt{N}}}\delta_{D}^{+}({{z}} {{z}}_{{P}}^{-1})
+({{\mathtt{g}}}^{+}({z})+{{\mathtt{g}}}^{{{\mathrm{inc}}}+}({z})).
\label{uNpC}
\end{split}\end{equation}
In above, recall that ${z}_{{P}}$ is defined by \eqref{zP}.}

{Substitution of $(-{\upomega}^2)\times$\eqref{uNpC} and the expression of ${\mathtt{B}}^{{\mathrm{F}}}(0, {\mathtt{M}}-1;{z})$ from \eqref{anidentity1} into \eqref{geneqPy0FT1} leads to a statement which is equivalent to \eqref{geneqPy0} (and \eqref{geneqPy0FT1}), namely,
\begin{equation}\begin{split}
{\mathtt{Q}}{\su}_{{{\mathtt{N}}}}^{-}({z})&={z}^{-{\mathtt{M}}}{\widehat{\saux}}_{{\mathtt{N}}}({z})+{\su}_{{{\mathtt{N}}}+1}^{-}({z})+{\su}_{{{\mathtt{N}}}-1}^{-}({z})-{\mathtt{Q}}({z}){{\mathtt{g}}}^{+}({z})+{{{\mathtt{f}}}}^{+}({z}),
\label{geneqPy0FT2}
\end{split}\end{equation}
where ${\mathtt{g}}^+$ and ${\mathtt{f}}^+$ are defined in \eqref{defq1inc} while ${\widehat{\saux}}_{{\mathtt{N}}}$ is defined in \eqref{WhP}.
Similar to the reason stated before \eqref{u0expn},
that ${\mathtt{Q}}$ does not have any zeros on the annulus ${{\mathscr{A}}}$,
adding ${\mathtt{Q}}\times$\eqref{uNpC} to \eqref{geneqPy0FT2} leads to
\begin{equation}\begin{split}
{\mathtt{Q}}({\su}_{{{\mathtt{N}}}}^{-}+{\su}_{{{\mathtt{N}}}}^{+})&={z}^{-{\mathtt{M}}}{\widehat{\saux}}_{{\mathtt{N}}}+{\su}_{{{\mathtt{N}}}+1}^{-}+{\su}_{{{\mathtt{N}}}-1}^{-}-{\su}^{{\mathrm{inc}}}_{0, {\mathtt{N}}}{\mathtt{Q}}\delta_{D}^{+}({{z}} {z}_{{P}}^{-1})\\
&+{\mathtt{Q}}{{\mathtt{g}}}^{{{\mathrm{inc}}}+}({z})+{{{\mathtt{f}}}}^{+}({z}).
\label{geneqPy0FT3}
\end{split}\end{equation}
}
{It is useful to simplify further the expression of the 
first term in the right hand side of \eqref{anidentity1inc}, i.e.,
${\mathtt{Q}} {{\mathtt{g}}}^{{\mathrm{inc}} {+}}({z})$ (see \eqref{defq1inc} for the definition of ${{\mathtt{g}}}^{{\mathrm{inc}} {+}}$). By an application of the claim \ref{claimone} (i.e., the identity \eqref{anidentity1}) to the incident wave},
\begin{subequations}
\begin{equation}\begin{split}
&0=\sum\limits_{{{\mathtt{x}}}\in\mathbb{Z}}{{z}}^{-{{\mathtt{x}}}}({\triangle}{\su}^{{\mathrm{inc}}}_{{{\mathtt{x}}}, {\mathtt{y}}}+{\upomega}^2{\su}^{{\mathrm{inc}}}_{{{\mathtt{x}}}, {\mathtt{y}}}){{\mathcal{H}}}({\mathtt{x}}){{\mathcal{H}}}({\mathtt{M}}-1-{{\mathtt{x}}})\\
&=-{\mathtt{Q}} \sum\limits_{{{\mathtt{x}}}\in\mathbb{Z}_0^{{\mathtt{M}}-1}}{{z}}^{-{{\mathtt{x}}}}{\su}^{{\mathrm{inc}}}_{{{\mathtt{x}}}, {\mathtt{N}}}-{z}\su^{{\mathrm{inc}}}_{0, {\mathtt{N}}}+\su^{{\mathrm{inc}}}_{-1, {\mathtt{N}}}+{z}^{-{\mathtt{M}}+1}\su^{{\mathrm{inc}}}_{{\mathtt{M}}, {\mathtt{N}}}-{z}^{-{\mathtt{M}}}\su^{{\mathrm{inc}}}_{{\mathtt{M}}-1, {\mathtt{N}}}+{{{\mathtt{f}}}}^{{\mathrm{inc}}+}({z}),
\label{anidentity1inc}
\end{split}\end{equation}
\begin{equation}\begin{split}
\text{where }
{{{\mathtt{f}}}}^{{\mathrm{inc}}+}({z}){\,:=}+\sum\limits_{{{\mathtt{x}}}\in\mathbb{Z}_0^{{\mathtt{M}}-1}}{{z}}^{-{{\mathtt{x}}}}({\su}^{{\mathrm{inc}}}_{{{\mathtt{x}}}, {\mathtt{N}}+1}+{\su}^{{\mathrm{inc}}}_{{{\mathtt{x}}}, {\mathtt{N}}-1}).
\label{defp10inc}
\end{split}\end{equation}
\end{subequations}
{Note that the left hand side in \eqref{anidentity1inc} is zero by virtue of the fact that the incident wave satisfies the discrete Helmholtz equation \eqref{dHelmholtz}}.
Therefore, 
{the fifth term in the right hand side of \eqref{geneqPy0FT3}, which is same as the negative of the first term in the right hand side of \eqref{anidentity1inc}${}_2$, can be re-written so that}
${\mathtt{Q}}({z}){{\mathtt{g}}}^{{\mathrm{inc}} {+}}({z})=-{\thicktilde{\saux}}^{{\mathrm{inc}}}_{{\mathtt{N}}}+{z}^{-{\mathtt{M}}}{\widehat{\saux}}^{{\mathrm{inc}}}_{{\mathtt{N}}}+{{{\mathtt{f}}}}^{{\mathrm{inc}}+}({z}),$
where
\begin{equation}\begin{split}
{\thicktilde{\saux}}^{{\mathrm{inc}}}_{{\mathtt{N}}}{\,:=}-{\su}^{{\mathrm{inc}}}_{-1, {\mathtt{N}}}+{{z}} {\su}^{{\mathrm{inc}}}_{0, {\mathtt{N}}},\quad
\text{and }
{\widehat{\saux}}^{{\mathrm{inc}}}_{{\mathtt{N}}}{\,:=}-{\su}^{{\mathrm{inc}}}_{{\mathtt{M}}-1, {\mathtt{N}}}+{{z}} {\su}^{{\mathrm{inc}}}_{{\mathtt{M}}, {\mathtt{N}}}.
\label{WhincP}
\end{split}\end{equation}
Hence, {the equation \eqref{geneqPy0FT3} can be replaced by an equivalent form}
\begin{equation}\begin{split}
{\mathtt{Q}}({\su}_{{{\mathtt{N}}}}^{-}+{\su}_{{{\mathtt{N}}}}^{+})&={z}^{-{\mathtt{M}}}{\widehat{\saux}}_{{\mathtt{N}}}+{\su}_{{{\mathtt{N}}}+1}^{-}+{\su}_{{{\mathtt{N}}}-1}^{-}-{\su}^{{\mathrm{inc}}}_{0, {\mathtt{N}}}{\mathtt{Q}}\delta_{D}^{+}({{z}} {z}_{{P}}^{-1})\\
&-{\thicktilde{\saux}}^{{\mathrm{inc}}}_{{\mathtt{N}}}+{z}^{-{\mathtt{M}}}{\widehat{\saux}}^{{\mathrm{inc}}}_{{\mathtt{N}}}+{{{\mathtt{f}}}}^{{\mathrm{inc}}+}({z})+{{{\mathtt{f}}}}^{+}({z}).
\label{masterCeqnP}
\end{split}\end{equation}
{Keeping in mind upcoming manipulation of expressions,} let the scattered and incident component of the sum of displacement field at the rows adjacent to the constrained rows be defined by
\begin{subequations}\begin{eqnarray}
{\mathrm{w}}_{{\mathtt{x}}, 0}&{\,:=}&{\su}_{{\mathtt{x}}, 1}+{\su}_{{\mathtt{x}}, -1}, \quad
{\mathrm{w}}^{{\mathrm{inc}}}_{{\mathtt{x}}, 0}{\,:=}{\su}^{{\mathrm{inc}}}_{{\mathtt{x}}, 1}+{\su}^{{\mathrm{inc}}}_{{\mathtt{x}}, -1}, \label{w0defAbra}\\
{\mathrm{w}}_{{\mathtt{x}}, {{\mathtt{N}}}}&{\,:=}&{\su}_{{\mathtt{x}}, {{\mathtt{N}}}+1}+{\su}_{{\mathtt{x}}, {{\mathtt{N}}}-1},
\quad
{\mathrm{w}}^{{\mathrm{inc}}}_{{\mathtt{x}}, {{\mathtt{N}}}}{\,:=}{\su}^{{\mathrm{inc}}}_{{\mathtt{x}}, {{\mathtt{N}}}+1}+{\su}^{{\mathrm{inc}}}_{{\mathtt{x}}, {{\mathtt{N}}}-1}.
\label{wNdefAbra}
\end{eqnarray}\label{wdefAbrafull}\end{subequations}
Thus, \eqref{defp10} and \eqref{defp10inc} can be written as, respectively,
\begin{equation}\begin{split}
{{{\mathtt{f}}}}^{+}({z})=+\sum\limits_{{{\mathtt{x}}}\in\mathbb{Z}_0^{{\mathtt{M}}-1}}{{z}}^{-{{\mathtt{x}}}}{\mathrm{w}}_{{\mathtt{x}}, {\mathtt{N}}},\quad {{{\mathtt{f}}}}^{{\mathrm{inc}}}{}^{+}({z})=+\sum\limits_{{{\mathtt{x}}}\in\mathbb{Z}_0^{{\mathtt{M}}-1}}{{z}}^{-{{\mathtt{x}}}}{\mathrm{w}}^{{\mathrm{inc}}}_{{\mathtt{x}}, {\mathtt{N}}}.
\label{defp1}
\end{split}\end{equation}
{The definition of expressions \eqref{wdefAbrafull} is utilized mainly during the application of Wiener--Hopf technique after a derivation of the Wiener--Hopf equation; however, in the following, the manner of writing sums on the right hand sides is deployed for a few paragraphs in the sequel.}

\subsection{{${\mathtt{M}}<0$}}
\label{slitMN}
{As remarked in the case of cracks, the manipulations of negative offset in the case of a pair of rigid constraints also follow closely those for positive offset as detailed above in \S\ref{slitMP}.}
{Analogous to the manner of writing \eqref{geneqPy0}},
for the row ${{\mathtt{y}}}={\mathtt{N}}$, the equation {\eqref{rigideq2} can re-arranged by adding and subtracting terms so that}
\begin{equation}\begin{split} 
-{\upomega}^2{\su}_{{{\mathtt{x}}}, {{\mathtt{N}}}}&={\triangle}{\su}_{{{\mathtt{x}}}, {\mathtt{y}}}|_{{\mathtt{y}}={\mathtt{N}}}{{\mathcal{H}}}(-1-{\mathtt{x}})+{\upomega}^2 {\su}^{{\mathrm{inc}}}_{{{\mathtt{x}}}, {{\mathtt{N}}}}{{\mathcal{H}}}({\mathtt{x}})\\
&+{[}(-{\triangle}{\su}_{{{\mathtt{x}}}, {\mathtt{y}}}|_{{\mathtt{y}}={\mathtt{N}}}-{{\upomega}^2 {\su}_{{{\mathtt{x}}}, {{\mathtt{N}}}}+{\upomega}^2 {\su}_{{{\mathtt{x}}}, {{\mathtt{N}}}}}+{\upomega}^2 {\su}^{{\mathrm{inc}}}_{{{\mathtt{x}}}, {{\mathtt{N}}}}){{\mathcal{H}}}(-1-{\mathtt{x}}){{\mathcal{H}}}({\mathtt{x}}-{\mathtt{M}}){]}. 
\label{geneqNy0}
\end{split}\end{equation}
{It is readily verified that for ${\mathtt{x}}-{{\mathtt{M}}}\in{\mathbb{Z}^-},$ \eqref{geneqNy0} reduces to \eqref{dHelmholtz} (with ${\su}$ replacing ${\su}^{{t}}$).}
{By an application of the claim \ref{claimone} to ${\mathtt{B}}^{{\mathrm{F}}}({\mathtt{M}}, -1;{z})$, while taking the Fourier transform \eqref{discreteFT} of the sum of the first and second term in the square brackets of \eqref{geneqNy0}, i.e., $-(\sum\nolimits_{{{\mathtt{x}}}\in\mathbb{Z}}{{z}}^{-{{\mathtt{x}}}}{\triangle}{\su}_{{{\mathtt{x}}}, {\mathtt{y}}}|_{{\mathtt{y}}={\mathtt{N}}}+\sum\nolimits_{{{\mathtt{x}}}\in\mathbb{Z}}{{z}}^{-{{\mathtt{x}}}}{\upomega}^2 {\su}^{{\mathrm{inc}}}_{{{\mathtt{x}}}, {{\mathtt{N}}}}){{\mathcal{H}}}({\mathtt{x}}-{\mathtt{M}}){{\mathcal{H}}}(-{\mathtt{x}}-1)$, yields an identity (using \eqref{wNdefAbra})}
\begin{subequations}
\begin{equation}\begin{split}
{-{\mathtt{B}}^{{\mathrm{F}}}({\mathtt{M}}, -1;{z})}&=-{\mathtt{Q}} {{\mathtt{g}}}^{{-}}({z})-{z}\su_{0, {\mathtt{N}}}+\su_{-1, {\mathtt{N}}}+{z}^{-{\mathtt{M}}}({-{z}\su^{\mathrm{inc}}_{{\mathtt{M}}, {\mathtt{N}}}}-\su_{{\mathtt{M}}-1, {\mathtt{N}}})
+{{{\mathtt{f}}}}^-({z}),
\end{split}\end{equation}
\label{BFMN}
\begin{equation}\begin{split}
\text{where }
{{\mathtt{g}}}^{{-}}({z}){\,:=}-\sum\limits_{{{\mathtt{x}}}\in\mathbb{Z}_{{\mathtt{M}}}^{-1}}{{z}}^{-{{\mathtt{x}}}}{\su}_{{{\mathtt{x}}}, {\mathtt{N}}}, \quad
{{{\mathtt{f}}}}^-({z}){\,:=}-\sum\limits_{{{\mathtt{x}}}\in\mathbb{Z}_{{\mathtt{M}}}^{-1}}{{z}}^{-{{\mathtt{x}}}}{\mathrm{w}}_{{\mathtt{x}}, {{\mathtt{N}}}}.
\label{defp1N}
\end{split}\end{equation}
\end{subequations}
{The notation $\mathtt{g}^{-}$ and $\mathtt{f}^{-}$ has been adopted as $\mathtt{M}<0$ implies that these complex functions are polynomials in ${z}$; recall the statement following \eqref{discreteFT}}.

{Therefore,} taking the Fourier transform \eqref{discreteFT} of {the entire equation \eqref{geneqNy0} yields an equation similar to \eqref{geneqPy0FT1},
\begin{equation}\begin{split}
{{\mathtt{Q}}}{\su}_{\mathtt{N}}^{-}-{\upomega}^2{\su}_{\mathtt{N}}^{+}&={\thicktilde{\saux}}_{\mathtt{N}}+{\su}_{{\mathtt{N}}+1}^{-}+{\su}_{{\mathtt{N}}-1}^{-}+{\upomega}^2{\su}^{{\mathrm{inc}}}_{0, {\mathtt{N}}}\delta_{D}^{+}({{z}} {z}_{{P}}^{-1})\\
&-{{\mathtt{B}}^{{\mathrm{F}}}({\mathtt{M}}, -1;{z})}-{\upomega}^2({{\mathtt{g}}}^{-}({z})+{{\mathtt{g}}}^{{{\mathrm{inc}}}-}({z})),
\label{geneqNy0FT1}
\end{split}\end{equation}
\begin{equation}\begin{split}
\text{where }
{{\mathtt{g}}}^{{{\mathrm{inc}}}-}({z}){\,:=}-\sum\limits_{{{\mathtt{x}}}\in\mathbb{Z}_{\mathtt{M}}^{-1}}{{z}}^{-{{\mathtt{x}}}}{\su}^{{\mathrm{inc}}}_{{{\mathtt{x}}}, {{\mathtt{N}}}}.\label{defq2inc}
\end{split}\end{equation}
In above, recall that ${z}_{{P}}$ is defined by \eqref{zP}.
Due to the rigid constraint \eqref{rigidconstraintdef} for ${\mathtt{x}}-{\mathtt{M}}\in{\mathbb{Z}^+}$, i.e., ${\su}_{{{\mathtt{x}}}, {{\mathtt{N}}}}+{\su}^{{\mathrm{inc}}}_{{{\mathtt{x}}}, {{\mathtt{N}}}}=0$ (the same also is part of \eqref{geneqPy0}), unlike \eqref{uNpC}, here,
\begin{subequations}
\begin{equation}\begin{split}
{\su}_{{{\mathtt{N}}}}^{+}({z})=\sum\limits_{{{\mathtt{x}}}\in\mathbb{Z}^+}{z}^{-{\mathtt{x}}}{\su}_{{\mathtt{x}},{{\mathtt{N}}}}=-{\su}^{{\mathrm{inc}}}_{0, {\mathtt{N}}}\delta_{D}^{+}({{z}} {{z}}_{{P}}^{-1}),
\label{uNnC}
\end{split}\end{equation}
\begin{equation}\begin{split}
\text{and }{{\mathtt{g}}}^{-}({z})+{{\mathtt{g}}}^{{{\mathrm{inc}}}-}({z})=0.
\label{partofCeq}
\end{split}\end{equation}
\label{uNnCall}
\end{subequations}
}
Note that ${{\mathtt{g}}}^{{-}}$, ${{\mathtt{g}}}^{{\mathrm{inc}} {-}}$, and ${{{\mathtt{f}}}}^-$ as given by 
\eqref{defp1N} and \eqref{defq2inc}, are counterparts to \eqref{defq1inc} and \eqref{defp1}, respectively, while {the latter also} uses \eqref{wNdefAbra}.
Hence, {by substitution of \eqref{BFMN} and \eqref{uNnCall} to \eqref{geneqNy0FT1} and an addition of ${\mathtt{Q}}\times$\eqref{uNnC} (since ${\mathtt{Q}}$ does not vanish on the annulus ${{\mathscr{A}}}$) to the result, it is found that}
\begin{equation}\begin{split}
{\mathtt{Q}}({\su}_{{{\mathtt{N}}}}^{-}+{\su}_{{{\mathtt{N}}}}^{+})&={z}^{-{\mathtt{M}}}{\widehat{\saux}}_{{\mathtt{N}}}+{\su}_{{{\mathtt{N}}}+1}^{-}+{\su}_{{{\mathtt{N}}}-1}^{-}-{\su}^{{\mathrm{inc}}}_{0, {\mathtt{N}}}{\mathtt{Q}}\delta_{D}^{+}({{z}} {z}_{{P}}^{-1})\\&+{{\mathtt{Q}} {{\mathtt{g}}}^{{\mathrm{inc}}{-}}({z})}+{{{\mathtt{f}}}}^-({z}).
\label{masterCeqN}
\end{split}\end{equation}
As ${\su}^{{\mathrm{inc}}}_{{{\mathtt{x}}}, {\mathtt{N}}}$ satisfies the discrete Helmholtz equation, the expression ${\mathtt{Q}} {{\mathtt{g}}}^{{\mathrm{inc}} {-}}({z})$ {can be expanded as before in \eqref{anidentity1inc}. In particular, by the claim \ref{claimone} (i.e., identity \eqref{anidentity1}),}
\begin{equation}\begin{split}
{0}&=-\sum\limits_{{{\mathtt{x}}}\in\mathbb{Z}}{{z}}^{-{{\mathtt{x}}}}({\triangle}{\su}^{{\mathrm{inc}}}_{{{\mathtt{x}}}, {\mathtt{y}}}+{\upomega}^2{\su}^{{\mathrm{inc}}}_{{{\mathtt{x}}}, {\mathtt{y}}}){{\mathcal{H}}}({\mathtt{x}}-{\mathtt{M}}){{\mathcal{H}}}(-1-{{\mathtt{x}}})\\
&={\mathtt{Q}} \sum\limits_{{{\mathtt{x}}}\in\mathbb{Z}_{{\mathtt{M}}}^{-1}}{{z}}^{-{{\mathtt{x}}}}{\su}^{{\mathrm{inc}}}_{{{\mathtt{x}}}, {\mathtt{N}}}-{z}\su^{{\mathrm{inc}}}_{0, {\mathtt{N}}}+\su^{{\mathrm{inc}}}_{-1, {\mathtt{N}}}+{z}^{-{\mathtt{M}}+1}\su^{{\mathrm{inc}}}_{{\mathtt{M}}, {\mathtt{N}}}-{z}^{-{\mathtt{M}}}\su^{{\mathrm{inc}}}_{{\mathtt{M}}-1, {\mathtt{N}}}+{{{\mathtt{f}}}}^{{\mathrm{inc}}-}({z}),
\end{split}\end{equation}
\begin{equation}\begin{split}
\text{where }
{{{\mathtt{f}}}}^{{\mathrm{inc}}-}({z}){\,:=}-\sum\limits_{{{\mathtt{x}}}\in\mathbb{Z}_{{\mathtt{M}}}^{-1}}{{z}}^{-{{\mathtt{x}}}}{\mathrm{w}}^{{\mathrm{inc}}}_{{\mathtt{x}}, {{\mathtt{N}}}}.
\label{defp1incN}
\end{split}\end{equation}
Therefore,
\begin{equation}\begin{split}
{\mathtt{Q}}({z}){{\mathtt{g}}}^{{\mathrm{inc}} {-}}({z})&=-{\thicktilde{\saux}}^{{\mathrm{inc}}}_{{\mathtt{N}}}+{z}^{-{\mathtt{M}}}{\widehat{\saux}}^{{\mathrm{inc}}}_{{\mathtt{N}}}+{{{\mathtt{f}}}}^{{\mathrm{inc}}-}({z}),
\end{split}\end{equation}
where
${\thicktilde{\saux}}^{{\mathrm{inc}}}_{{\mathtt{N}}}$
and
${\widehat{\saux}}^{{\mathrm{inc}}}_{{\mathtt{N}}}$
are given by 
\eqref{WhincP}, respectively.
Hence, \eqref{masterCeqN} becomes
\begin{equation}\begin{split}
{\mathtt{Q}}({\su}_{{{\mathtt{N}}}}^{-}+{\su}_{{{\mathtt{N}}}}^{+})&={z}^{-{\mathtt{M}}}{\widehat{\saux}}_{{\mathtt{N}}}+{\su}_{{{\mathtt{N}}}+1}^{-}+{\su}_{{{\mathtt{N}}}-1}^{-}-{\su}^{{\mathrm{inc}}}_{0, {\mathtt{N}}}{\mathtt{Q}}\delta_{D}^{+}({{z}} {z}_{{P}}^{-1})\\
&-{\thicktilde{\saux}}^{{\mathrm{inc}}}_{{\mathtt{N}}}+{z}^{-{\mathtt{M}}}{\widehat{\saux}}^{{\mathrm{inc}}}_{{\mathtt{N}}}+{{{\mathtt{f}}}}^{{\mathrm{inc}}-}({z})+{{{\mathtt{f}}}}^{-}({z}).
\label{masterCeqnN}
\end{split}\end{equation}
{Indeed, it has been an inspired attempt to bring above into a form similar to \eqref{masterCeqnP}.}

\subsection{{{{Wiener--Hopf}} equation}}
\label{WHconstrainteqns}
In {line with the program attained in the unification of both signs of stagger for the case of cracks in \S\ref{WHconstrainteqns}, the same can be carried out for the pair of rigid constraints}.
The resulting equation for the two cases {(namely, \eqref{masterCeqnP} and \eqref{masterCeqnN})} in {\S\ref{slitMP}} and {\S\ref{slitMN}}, with $\sgnM$ denoting the sign of ${\mathtt{M}}$, can be written as
\begin{equation}\begin{split}
{\su}_{{{\mathtt{N}}}}^{{\mathrm{F}}}&=\recip{{\mathtt{Q}}}({\su}_{{\mathtt{N}}+1}^{-}+{\su}_{{{\mathtt{N}}}-1}^{-})+\recip{{\mathtt{Q}}}{\saux}_{{\mathtt{N}}}+\recip{{\mathtt{Q}}}({{{\mathtt{f}}}}^{{\mathrm{inc}}{\sgnM}}+{{{\mathtt{f}}}}^{{\sgnM}})-{\su}^{{\mathrm{inc}}}_{0, {\mathtt{N}}}\delta_{D}^{+}({{z}} {z}_{{P}}^{-1}),
\label{uNexpn}
\end{split}\end{equation}
\begin{subequations}\begin{eqnarray}
\text{where }{\saux}_{{\mathtt{N}}}({z})&=&{z}^{-{\mathtt{M}}}({\widehat{\saux}}_{{\mathtt{N}}}+{\widehat{\saux}}^{{\mathrm{inc}}}_{{\mathtt{N}}})-{\thicktilde{\saux}}^{{\mathrm{inc}}}_{{\mathtt{N}}},\label{WcN}\\
{{{\mathtt{f}}}}^{+}({z})&=&+\sum\limits_{{{\mathtt{x}}}\in\mathbb{Z}_0^{{\mathtt{M}}-1}}{{z}}^{-{{\mathtt{x}}}}{\mathrm{w}}_{{\mathtt{x}}, {{\mathtt{N}}}},\quad {{{\mathtt{f}}}}^{-}({z})=-\sum\limits_{{{\mathtt{x}}}\in\mathbb{Z}_{{\mathtt{M}}}^{-1}}{{z}}^{-{{\mathtt{x}}}}{\mathrm{w}}_{{\mathtt{x}}, {{\mathtt{N}}}}, \label{ppnC}\\
\text{and }
{{{\mathtt{f}}}}^{{\mathrm{inc}}+}({z})&=&+\sum\limits_{{{\mathtt{x}}}\in\mathbb{Z}_0^{{\mathtt{M}}-1}}{{z}}^{-{{\mathtt{x}}}}{\mathrm{w}}^{{\mathrm{inc}}}_{{\mathtt{x}}, {{\mathtt{N}}}},\quad
{{{\mathtt{f}}}}^{{\mathrm{inc}}-}({z})=-\sum\limits_{{{\mathtt{x}}}\in\mathbb{Z}_{{\mathtt{M}}}^{-1}}{{z}}^{-{{\mathtt{x}}}}{\mathrm{w}}^{{\mathrm{inc}}}_{{\mathtt{x}}, {{\mathtt{N}}}},\label{ppnincC}
\end{eqnarray}\label{uNexpnfull}\end{subequations}
along with {\eqref{WhP}${}_2$, \eqref{WhincP}${}_2$, and \eqref{WhincP}${}_1$, } i.e.,
{$
\widehat{\saux}_{{\mathtt{N}}}=-{\su}_{{\mathtt{M}}-1, {\mathtt{N}}}-{{{z}} {\su}^{\mathrm{inc}}_{{\mathtt{M}}, {\mathtt{N}}}},
\widehat{\saux}^{{\mathrm{inc}}}_{{\mathtt{N}}}=-{\su}^{{\mathrm{inc}}}_{{\mathtt{M}}-1, {\mathtt{N}}}+{{z}} {\su}^{{\mathrm{inc}}}_{{\mathtt{M}}, {\mathtt{N}}},
\text{ and }
\thicktilde{\saux}^{{\mathrm{inc}}}_{{\mathtt{N}}}=-{\su}^{{\mathrm{inc}}}_{-1, {\mathtt{N}}}+{{z}} {\su}^{{\mathrm{inc}}}_{0, {\mathtt{N}}}.$}

{
\begin{remark}
In an expanded form, ${\saux}_{{\mathtt{N}}}={z}^{-{\mathtt{M}}}(-{\su}^{t}_{{\mathtt{M}}-1, {\mathtt{N}}})+{\su}^{{\mathrm{inc}}}_{-1, {\mathtt{N}}}-{{z}} {\su}^{{\mathrm{inc}}}_{0, {\mathtt{N}}}$.
When $\mathtt{M}=0,$ the expression of ${\saux}_{{\mathtt{N}}}$ according to \eqref{WcN} reduces to that for the zero offset, i.e., aligned rigid constraints, as discussed by \cite{Bls8pair1}, namely, ${\saux}_{{\mathtt{N}}}=-{\su}_{-1, {\mathtt{N}}}-{{z}} {\su}^{\mathrm{inc}}_{0, {\mathtt{N}}}$.
\label{zeroM}
\end{remark}
}

\begin{figure}[h!]
\centering
\includegraphics[width=\textwidth]{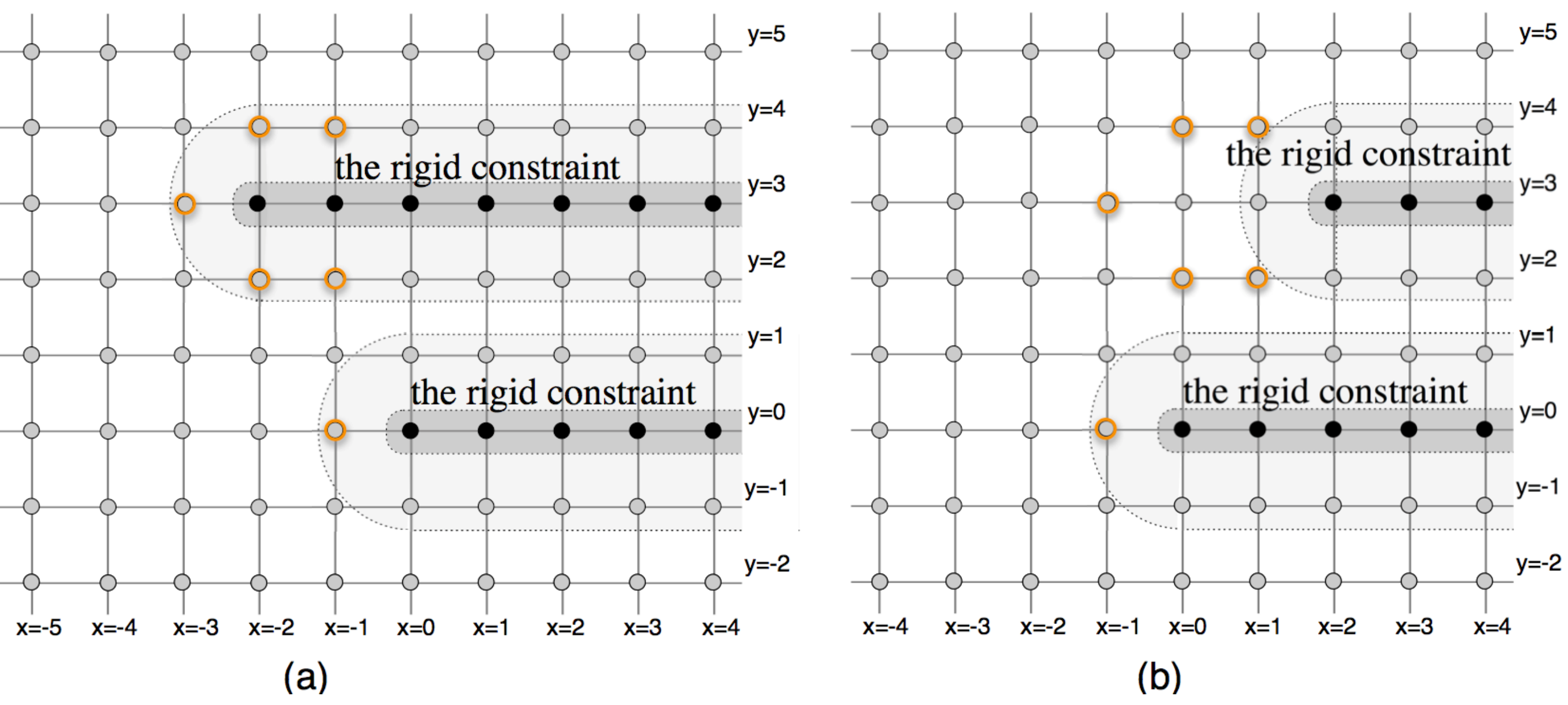}
\caption{
{
For the reduced (algebraic) problem, the unknowns are marked in orange, corresponding to Fig. \ref{Fig1}(b).
The gray portion represents 
the set of all lattice sites in ${\mathfrak{S}}$ that index those particles coupled to atleast one constrained site, can be interpreted as a `boundary layer' around the constrained sites ${{\Sigma}}_{{c}}$ \cite{Bls1, Bls3}.
(a) ${{\mathtt{M}}}=-2$ and (b) ${{\mathtt{M}}}=2$.
}
}
\label{Fig1_newunknowns2}
\end{figure}

{In above, it is useful to recall that ${z}_{{P}}$ is defined by \eqref{zP}.}
Notice that ${\su}_{{\mathtt{M}}-1, {\mathtt{N}}}$ is an unknown as well as $\{{\mathrm{w}}_{{\mathtt{x}}, {{\mathtt{N}}}}\}_{{{\mathtt{x}}}=0}^{{\mathtt{M}}-1}$ or $\{{\mathrm{w}}_{{\mathtt{x}}, {{\mathtt{N}}}}\}_{{{\mathtt{x}}}={\mathtt{M}}}^{-1}$. {Also it is worth a recollection} that $\su_{-1, 0}$ is also unknown {(Remark \ref{u00remark})}.
For example, the set of unknowns in case of ${\mathtt{M}}=-2$ {and ${\mathtt{M}}=+2$} has four elements, see {respectively, Fig. \ref{Fig1_newunknowns2}(a) and Fig. \ref{Fig1_newunknowns2}(b)}.

According to \eqref{ubulkC}, ${\su}_{-1}^{{\mathrm{F}}}=\su^{{\mathrm{F}}}_{0}{{\lambda}}, {\su}_{1}^{{\mathrm{F}}}={{\mathscr{F}}}_1{\su}_0+{{\mathscr{G}}}_1{\su}_{{\mathtt{N}}}, {\su}_{{{\mathtt{N}}}-1}^{{\mathrm{F}}}={{\mathscr{F}}}_{{{\mathtt{N}}}-1}{\su}_0+{{\mathscr{G}}}_{{{\mathtt{N}}}-1}{\su}_{{\mathtt{N}}}, {\su}_{{{\mathtt{N}}}+1}^{{\mathrm{F}}}=\su^{{\mathrm{F}}}_{{\mathtt{N}}}{{\lambda}}.$ Writing the same in expanded form {(recall the additive splitting \eqref{discreteFT})},
\begin{equation}\begin{split}
{\su}_{-1}^{-}
+{\su}_{-1}^{+}&=({\su}_{0}^{-}+{\su}_{0}^{+}){{\lambda}}, \\
{\su}_{1}^{-}+
{\su}_{1}^{+}&={{\mathscr{F}}}_1({\su}_{0}^{-}+{\su}_{0}^{+})+{{\mathscr{G}}}_1({\su}_{{{\mathtt{N}}}}^{-}+{\su}_{{{\mathtt{N}}}}^{+}), \\
{\su}_{{{\mathtt{N}}}-1}^{-}+{\su}_{{{\mathtt{N}}}-1}^{+}&={{\mathscr{F}}}_{{{\mathtt{N}}}-1}({\su}_{0}^{-}+{\su}_{0}^{+})+{{\mathscr{G}}}_{{{\mathtt{N}}}-1}({\su}_{{{\mathtt{N}}}}^{-}+{\su}_{{{\mathtt{N}}}}^{+}), \\
{\su}_{{{\mathtt{N}}}+1}^{-}+{\su}_{{{\mathtt{N}}}+1}^{+}&=({\su}_{{{\mathtt{N}}}}^{-}+{\su}_{{{\mathtt{N}}}}^{+}){{\lambda}}.
\label{foureqns}
\end{split}\end{equation}
The particular ${\mathscr{F}}$s and ${\mathscr{G}}$s in above can be read out from \eqref{ubulkC}.
Using \eqref{u0expn} and \eqref{uNexpn}, the system \eqref{foureqns} of four equations, involving ${\su}_{\pm1; \pm}, {\su}_{{{\mathtt{N}}}\pm1; \pm}$ as $\pm$ parts of four unknown complex functions, can be written as (use \eqref{ubulkC})
\begin{equation}\begin{split}
\mathbf{I}_{4\times 4}\boldsymbol{u}^++\mathbf{K}_{4\times 4}\boldsymbol{u}^-=\widehat{\boldsymbol{c}}, 
\label{4WHCmatrix1}
\end{split}\end{equation}
\begin{equation}\begin{split}
\text{where }\boldsymbol{u}^{\pm}&{\,:=}\begin{bmatrix}
{\su}_{-1}^{\pm}\\
{\su}_{1}^{\pm}\\
{\su}_{{{\mathtt{N}}}-1}^{\pm}\\
{\su}_{{{\mathtt{N}}}+1}^{\pm}
\end{bmatrix}, \mathbf{K}_{4\times 4}
{\,:=}\mathbf{I}_{4\times 4}-\recip{{\mathtt{Q}}}\begin{bmatrix}
{{\lambda}}&{{\lambda}}&0&0\\
{{{\mathscr{F}}}_1}&{{{\mathscr{F}}}_1}&{{{\mathscr{F}}}_{{{\mathtt{N}}}-1}}&{{{\mathscr{F}}}_{{{\mathtt{N}}}-1}}\\
{{{\mathscr{F}}}_{{{\mathtt{N}}}-1}}&{{{\mathscr{F}}}_{{{\mathtt{N}}}-1}}&{{{\mathscr{F}}}_1}&{{{\mathscr{F}}}_1}\\
0&0&{{\lambda}}&{{\lambda}}
\end{bmatrix}, \\
\widehat{\boldsymbol{c}}&{\,:=}\begin{bmatrix}
(\dfrac{{\saux}_0}{{\mathtt{Q}}}+{\su}_{0}^{+}){{\lambda}}\\
{{\mathscr{F}}}_1(\dfrac{{\saux}_0}{{\mathtt{Q}}}+{\su}_{0}^{+})+{{\mathscr{F}}}_{{{\mathtt{N}}}-1}(\dfrac{{\saux}_{{\mathtt{N}}}+{{{\mathtt{f}}}}^{{\mathrm{inc}}{\sgnM}}+{{{\mathtt{f}}}}^{{\sgnM}}}{{\mathtt{Q}}}+{\su}_{{{\mathtt{N}}}}^{+})\\
{{\mathscr{F}}}_1(\dfrac{{\saux}_{{\mathtt{N}}}+{{{\mathtt{f}}}}^{{\mathrm{inc}}{\sgnM}}+{{{\mathtt{f}}}}^{{\sgnM}}}{{\mathtt{Q}}}+{\su}_{{{\mathtt{N}}}}^{+})+{{\mathscr{F}}}_{{{\mathtt{N}}}-1}(\dfrac{{\saux}_0}{{\mathtt{Q}}}+{\su}_{0}^{+})\\
(\dfrac{{\saux}_{{\mathtt{N}}}+{{{\mathtt{f}}}}^{{\mathrm{inc}}{\sgnM}}+{{{\mathtt{f}}}}^{{\sgnM}}}{{\mathtt{Q}}}+{\su}_{{{\mathtt{N}}}}^{+}){{\lambda}}
\end{bmatrix}.
\label{4WHCmatrix2}
\end{split}\end{equation}
Noticing {an obviously peculiar} structure of the matrix kernel $\mathbf{K}_{4\times 4}$ in \eqref{4WHCmatrix1} and \eqref{4WHCmatrix2}, 
{by} adding first and second as well as third and fourth equations, above system also implies a certain system of two {{Wiener--Hopf}} equations involving ${\su}_{-1}^{ \pm}+{\su}_{1}^{ \pm}, {\su}_{{{\mathtt{N}}}-1}^{ \pm}+{\su}_{{{\mathtt{N}}}+1}^{ \pm}$ as $\pm$ counterparts of two unknown functions as components of $\boldsymbol{{\mathrm{w}}}$ (recall {\eqref{w0defAbra} and \eqref{wNdefAbra}}), {namely,
\begin{equation}\begin{split}
{{\mathrm{w}}}_{0}^{ \pm}&=\sum_{{\mathtt{x}}\in\mathbb{Z}^\pm}{z}^{-{\mathtt{x}}}{\mathrm{w}}_{{\mathtt{x}}, 0}={\su}_{-1}^{ \pm}+{\su}_{1}^{ \pm},
\quad
{{\mathrm{w}}}_{{\mathtt{N}}}^{ \pm}=\sum_{{\mathtt{x}}\in\mathbb{Z}^\pm}{z}^{-{\mathtt{x}}}{\mathrm{w}}_{{\mathtt{x}}, {\mathtt{N}}}={\su}_{{{\mathtt{N}}}-1}^{ \pm}+{\su}_{{{\mathtt{N}}}+1}^{ \pm}.
\label{FTw0Npn}
\end{split}\end{equation}
} 
{
\begin{remark}
As stated in the context of \eqref{ubulkC}, now it is clear that the two functions ${\su}_0^{{\mathrm{F}}}$ and ${\su}_{{\mathtt{N}}}^{{\mathrm{F}}}$ can be found in terms of ${\mathrm{w}}_0^{{\mathrm{F}}}$ and ${\mathrm{w}}_{\mathtt{N}}^{{\mathrm{F}}}$ via \eqref{u0expn} and \eqref{uNexpn}, namely,
${\su}_0^{{\mathrm{F}}}({{z}})
=\recip{\mathtt{Q}}({\saux}_0({{z}})+{\mathrm{w}}_0^{-})-{{\su}^{{\mathrm{inc}}}_{0,0}}\delta_{D}^{+}({{z}} {z}_{{P}}^{-1})$ and
${\su}_{{{\mathtt{N}}}}^{{\mathrm{F}}}
=\recip{{\mathtt{Q}}}({\mathrm{w}}_{\mathtt{N}}^{-}+{\saux}_{{\mathtt{N}}}+{{\mathtt{f}}}^{{\mathrm{inc}}{\sgnM}}+{{\mathtt{f}}}^{{\sgnM}})-{\su}^{{\mathrm{inc}}}_{0, {\mathtt{N}}}\delta_{D}^{+}({{z}} {z}_{{P}}^{-1}).$
\label{u0uNFrem}
\end{remark}
}

After 
simplifying {\eqref{4WHCmatrix1} as described in preceding paragraph}, using the definitions of ${\mathscr{F}}$s and ${\mathscr{G}}$s in \eqref{ubulkC},
it is found that
\begin{subequations}
\begin{equation}\begin{split}
\boldsymbol{{\mathrm{w}}}^-+\mathbf{L}\boldsymbol{{\mathrm{w}}}^+=\widetilde{\boldsymbol{c}}, 
\label{WHCAbra}
\end{split}\end{equation}
\begin{equation}\begin{split}
\text{where }\mathbf{L}&{\,:=}\frac{{{\lambda}}^{-1}+{{\lambda}}}{{{\lambda}}^{-1}-{{\lambda}}}(\mathbf{I}+{{\lambda}}^{{\mathtt{N}}}\begin{bmatrix}0&1\\1&0\end{bmatrix})
={{\mathcal{L}}}_{{c}}\begin{bmatrix}
1&{{\lambda}}^{{\mathtt{N}}}\\
{{\lambda}}^{{\mathtt{N}}}&1
\end{bmatrix}, {{\mathcal{L}}}_{{c}}{\,:=}\frac{{\mathtt{Q}}}{{\mathtt{r}}{\mathtt{h}}},
\label{WHCkernel}
\end{split}\end{equation}
\begin{equation}\begin{split}
\text{and }
\boldsymbol{{\mathrm{w}}}^\pm&{\,:=}
\begin{bmatrix}
{{\mathrm{w}}}_{0}^{ \pm}\\
{{\mathrm{w}}}_{{\mathtt{N}}}^{ \pm}
\end{bmatrix}, \widetilde{\boldsymbol{c}}{\,:=}-(\mathbf{I}-\mathbf{L})({\mathtt{Q}}\boldsymbol{q}^{{\mathrm{inc}}}{}^++\boldsymbol{p}^{\saux}+\boldsymbol{p}^{\sgnM}),
\label{ctildeslits}
\end{split}\end{equation}
with 
\begin{equation}\begin{split}
\boldsymbol{q}^{{\mathrm{inc}}}{}^+&{\,:=-\begin{bmatrix}
{\su}^{{\mathrm{inc}}}_{{0}}{}^{+}\\
{\su}^{{\mathrm{inc}}}_{\mathtt{N}}{}^{+}
\end{bmatrix}=}-{\su}^{{\mathrm{inc}}}_{{0}}{}^{+}\begin{bmatrix}
1\\
e^{i\upkappa_y{\mathtt{N}}}
\end{bmatrix}=-{{\mathrm{A}}}\delta_{D}^{+}({{z}} {z}_{{P}}^{-1})\begin{bmatrix}
1\\
e^{i\upkappa_y{\mathtt{N}}}
\end{bmatrix}, \\
\boldsymbol{p}^{\saux}&{\,:=}\begin{bmatrix}
{\saux}_{0}\\
{\saux}_{{\mathtt{N}}}
\end{bmatrix},
\boldsymbol{p}^{\sgnM}{\,:=}\begin{bmatrix}
0\\
{{{\mathtt{f}}}}^{{\mathrm{inc}}{\sgnM}}+{{{\mathtt{f}}}}^{{\sgnM}}
\end{bmatrix}
\label{qPcAbra}
\end{split}\end{equation}
\label{WHCequations}
\end{subequations}
{In the context of the definition of ${{\mathcal{L}}}_{{c}}$ \eqref{WHCkernel}${}_3$, recall that ${\mathtt{Q}}$ is defined by \eqref{defW0Q}${}_2$, while the definitions of ${\mathtt{h}}$ and ${\mathtt{r}}$ provided in \eqref{lamL}. In order to obtain \eqref{WHCkernel}${}_2$, an identity is used, that is, $({{\lambda}}^{-1}+{{\lambda}})/({{\lambda}}^{-1}-{{\lambda}})={\mathtt{Q}}/({\mathtt{r}}{\mathtt{h}})$ which utilizes the definition of ${\lambda}$ in terms of ${\mathtt{h}}$ and ${\mathtt{r}}$ as provided in \eqref{lamL} and the definition of ${\mathtt{Q}}$ in \eqref{defW0Q}${}_2$ (i.e., ${\mathtt{Q}}={\mathtt{H}}+2={\mathtt{h}}^2+2$).}

{In the expression of $\widetilde{\boldsymbol{c}}$ provided in \eqref{ctildeslits}${}_2$, it is emphasized that} $\boldsymbol{p}^{\sgnM}$ is an unknown polynomial in ${z}^{-1}$ (resp. ${z}$) for $\sgnM=1$, i.e., ${\mathtt{M}}>0$ (resp. $\sgnM=-1$, i.e., ${\mathtt{M}}<0$) given by \eqref{defp1}, \eqref{defp1N}, and \eqref{defp1incN}.
For example, the set of unknowns in case of ${\mathtt{M}}=-2$ has four elements, see Fig. \ref{Fig1_newunknowns2}(a).

\begin{remark}
It is noteworthy that the zero offset case \cite{Bls8pair1} is a special case of the presented formulation in a sense that $\boldsymbol{p}^{\sgnM}\equiv\boldsymbol{0}$, however, the effect of non-zero ${\mathtt{M}}$ is not only captured via $\boldsymbol{p}^{\sgnM}$ but also $\boldsymbol{p}^{\saux}$.
\label{perturbrigid}
\end{remark}

\subsection{{Reduction of matrix Wiener--Hopf problem}}
\label{reductionC}
From the multiplicative factorization of the kernel $\mathbf{L}$ \eqref{WHCkernel}, i.e., $\mathbf{L}=\mathbf{L}_-\mathbf{L}_+,$ it is found that the {{Wiener--Hopf}} equation \eqref{WHCAbra} becomes 
\begin{equation}\begin{split}
\recip{\mathbf{L}}_-\boldsymbol{{\mathrm{w}}}^-+\mathbf{L}_+\boldsymbol{{\mathrm{w}}}^+=\boldsymbol{c}{\,:=}-(\recip{\mathbf{L}}_--\mathbf{L}_+)({\mathtt{Q}}\boldsymbol{q}^{{\mathrm{inc}}}{}^++\boldsymbol{p}^{\saux}+\boldsymbol{p}^{\sgnM}), 
\label{WHCn}
\end{split}\end{equation}
where the second component of $\boldsymbol{p}^{\sgnM}$ is an unknown polynomial in ${z}^{-1}$ (resp. ${z}$) for ${\mathtt{M}}>0$ (resp. ${\mathtt{M}}<0$) given by \eqref{qPcAbra}; also {two more unknown values, via the expression of $\boldsymbol{p}^{\saux}$, }
are present in the same equation.
Analogous to the detailed expressions provided earlier (note that ${{\mathcal{L}}}_{{c}}={{{\mathcal{L}}}_{c+}}{{{\mathcal{L}}}_{c-}}$, ${{{\mathcal{L}}}_{c\pm}}={{{\mathtt{Q}}_\pm}/({{\mathtt{r}}_\pm{\mathtt{h}}_\pm})}$ \cite{Bls1}) in {\eqref{Lfactorsfull}}, in this case,
\begin{equation}\begin{split}
(\mathbf{L}_\pm)^{\pm1}&={({{\mathbf{D}}}_\pm)^{\pm1}({\mathbf{J}}_\pm)^{\pm1}}=({{\mathcal{L}}}_{c\pm})^{\pm1}(\tilde{{\mathbf{D}}}_\pm)^{\pm1}({\mathbf{J}}_\pm)^{\pm1},\\
\text{with }
{{\mathbf{D}}}_\pm&=({{\mathcal{L}}}_{{c}{\pm}})^{\pm1}\begin{bmatrix}
(1-{{\lambda}}^{{\mathtt{N}}})_\pm&0\\
0&(1+{{\lambda}}^{{\mathtt{N}}})_\pm
\end{bmatrix}=\begin{bmatrix}
{\upalpha}_\pm&0\\
0&{\upbeta}_\pm
\end{bmatrix}.
\label{Dfactorsrigid}
\end{split}\end{equation}
Therefore,
in \eqref{WHCn}, the right hand side $\boldsymbol{c}$ becomes
\begin{equation}\begin{split}
\boldsymbol{c}
&=-(\recip{{{\mathbf{D}}}}_-{\mathbf{J}}-{{\mathbf{D}}}_+{\mathbf{J}})({\mathtt{Q}}\boldsymbol{q}^{{\mathrm{inc}}}{}^++\boldsymbol{p}^{\saux}+\boldsymbol{p}^{\sgnM}).
\label{cformC}
\end{split}\end{equation}
The additive factorization $\boldsymbol{c}=\boldsymbol{c}^++\boldsymbol{c}^-$ is needed. {To reach this stage, it is useful to introdice a convenient splitting provided by the definitions}
\begin{equation}\begin{split}
\boldsymbol{c}&=\boldsymbol{c}^{\old}+\boldsymbol{c}^{\saux}+\boldsymbol{c}^{\redc},\\
\text{where }\boldsymbol{c}^{\old}&{\,:=}[\cdot]{\mathtt{Q}}\boldsymbol{q}^{{\mathrm{inc}}}{}^+,\quad \boldsymbol{c}^{\saux}{\,:=}[\cdot]\boldsymbol{p}^{{\saux}},\quad \boldsymbol{c}^{\redc}{\,:=}[\cdot]\boldsymbol{p}^{\sgnM},\\
\text{and }[\cdot]&{\,:=}-\frac{1}{\sqrt{2}}(\begin{bmatrix}
\recip{{\upalpha}}_-&-\recip{{\upalpha}}_-\\
\recip{{\upbeta}}_-&\recip{{\upbeta}}_-
\end{bmatrix}-\begin{bmatrix}
{\upalpha}_+&-{\upalpha}_+\\
{\upbeta}_+&{\upbeta}_+
\end{bmatrix}).
\label{ccincC}
\end{split}\end{equation}

Finally, the {(formal)} solution can {be} written in terms of Fourier transforms {\eqref{discreteFT}} as
\begin{equation}\begin{split}
\recip{\mathbf{L}}_-\boldsymbol{{\mathrm{w}}}^-=\boldsymbol{c}^-,
\quad
\mathbf{L}_+\boldsymbol{{\mathrm{w}}}^+=\boldsymbol{c}^+.
\label{WHsolC}
\end{split}\end{equation}
{As a result,} \eqref{u0expn} and \eqref{uNexpn} yield the expressions for $\su_{0}$ and $\su_{{\mathtt{N}}},$ respectively. Eventually, the field is determined by \eqref{ubulkC}. 
However, the problem is far from solved yet due to presence of unknowns in the right hand side {through the presence} of $\boldsymbol{p}^{\sgnM}$ and $\boldsymbol{p}^{\saux}$.
In view of the Remark \ref{perturbrigid}, it is noted that the sum $\boldsymbol{c}^{\saux}+\boldsymbol{c}^{\redc}$ describes the effect of a perturbation introduced by offset ${\mathtt{M}}.$

{\bf Additive factorization of $\boldsymbol{c}^{\old}$:}
Note that 
\begin{equation}\begin{split}
\boldsymbol{c}^{\old}
&=-\big({\mathtt{Q}}\recip{{{\mathbf{D}}}}_--{\mathtt{Q}}({z}_{{P}})\recip{{{\mathbf{D}}}}_-({z}_{{P}})+({z}^{-1}-{z}_{{P}}^{-1})\recip{{{\mathbf{D}}}}_-(0)+({z}-{z}_{{P}}){{\mathbf{D}}}_+(\infty)\big){\mathbf{J}}\boldsymbol{q}^{{\mathrm{inc}}}{}^+\\
&-\big(-{\mathtt{Q}}{{\mathbf{D}}}_++{\mathtt{Q}}({z}_{{P}})\recip{{{\mathbf{D}}}}_-({z}_{{P}})+({z}_{{P}}^{-1}-{z}^{-1})\recip{{{\mathbf{D}}}}_-(0)+({z}_{{P}}-{z}){{\mathbf{D}}}_+(\infty)
\big){\mathbf{J}}\boldsymbol{q}^{{\mathrm{inc}}}{}^+\\
&=\boldsymbol{c}^{\old}{}^-+\boldsymbol{c}^{\old}{}^+.
\label{cincfacsC}
\end{split}\end{equation}
{In above, recall that ${z}_{{P}}$ is defined by \eqref{zP}.}
{Note that ${{\mathbf{D}}}^-(0)$ \eqref{Dfactorsrigid} equals the diagonal matrix $\lim_{{z}\to0}\text{diag}({\upalpha}^-({z}),{\upbeta}^-({z}))$,
while ${{\mathbf{D}}}_+(\infty)$ \eqref{Dfactorsrigid} equals the diagonal matrix $\lim_{{z}\to\infty}\text{diag}({\upalpha}_+({z}),{\upbeta}_+({z}))$.}

Depending on the sign of ${\mathtt{M}}$, the terms 
\begin{equation}\begin{split}
-{\recip{{\mathbf{D}}}_-({z})}{\mathbf{J}}\boldsymbol{p}^{+}({z})\text{ and }{{\mathbf{D}}}_+({z}){\mathbf{J}}\boldsymbol{p}^{-}({z})
\label{difficultC}
\end{split}\end{equation}
in $\boldsymbol{c}^{\redc}$, as well as {a similar term in} $\boldsymbol{c}^{\saux}$, need to be factorized in a manner {analogous to that following \eqref{difficultK};
this is discussed below.}

For illustration, consider the case ${\mathtt{M}}>0$, the details for the other case {(${\mathtt{M}}<0$)} are provided in Appendix \ref{ApprigidMN}.
{As evident from Fig. \ref{Fig1_newunknowns2}(a) and Fig. \ref{Fig1_newunknowns2}(b), the difference between the two cases lies in the appearance of (shown as sites colored in orange) the unknown field at designated sites. {It is clear now that} ${\mathtt{M}}<0$ is associated with the imposition of the condition that the total field needs to be zero on an `extra' finite set of sites whereas ${\mathtt{M}}>0$ leads to `relaxation' of the constraint of zero total wave field on a finite set of sites.}

{\bf Additive factorization of $\boldsymbol{c}^{\saux}$:}
{For the task of relevant factorization, it is useful to recall} the definition of $\boldsymbol{c}^{\saux}$ \eqref{ccincC} and $\boldsymbol{p}^{\saux}$ \eqref{qPcAbra} with $\saux$s given by \eqref{u0expn} and \eqref{WcN} along with \eqref{WhP} and \eqref{WhincP};
specifically,
\begin{equation}\begin{split}
\boldsymbol{c}^{\saux}&=-(\recip{{{\mathbf{D}}}}_-{\mathbf{J}}-{{\mathbf{D}}}_+{\mathbf{J}})\boldsymbol{p}^{\saux},\\
\boldsymbol{p}^{\saux}
&=\begin{bmatrix}
{-{\su}_{-1, 0}-{{z}} {\su}^{\mathrm{inc}}_{0, 0}}\\
{-{z}^{-{\mathtt{M}}}{\su}^{{t}}_{{\mathtt{M}}-1, {\mathtt{N}}}}
+{\su}^{{\mathrm{inc}}}_{-1, {\mathtt{N}}}-{{z}} {\su}^{{\mathrm{inc}}}_{0, {\mathtt{N}}}
\end{bmatrix},
\label{caux}
\end{split}\end{equation}
{while recalling Remark \ref{u00remark} and Remark \ref{uMNremark}.}

Using \eqref{caux} {(which is first defined in \eqref{qPcAbra})},
\begin{equation}\begin{split}
\boldsymbol{p}^{\saux}&=\boldsymbol{p}^{\saux0}+{{z}} \boldsymbol{p}^{\saux1}-{z}^{-{\mathtt{M}}}{\su}^{{t}}_{{\mathtt{M}}-1, {\mathtt{N}}}\ensuremath{\hat{\mathbf{e}}}_2,\\
\boldsymbol{p}^{\saux0}&{\,:=}\begin{bmatrix}
-{\su}_{-1, 0}\\
{\su}^{{\mathrm{inc}}}_{-1, {\mathtt{N}}}
\end{bmatrix}{=\begin{bmatrix}
-{\su}^{t}_{-1, 0}+{\su}^{{\mathrm{inc}}}_{-1, 0}\\
{\su}^{{\mathrm{inc}}}_{-1, {\mathtt{N}}}
\end{bmatrix}}, \quad 
\boldsymbol{p}^{\saux1}{\,:=}\begin{bmatrix}
{-{\su}^{{\mathrm{inc}}}_{0, 0}}\\
-{\su}^{{\mathrm{inc}}}_{0, {\mathtt{N}}}
\end{bmatrix},
\label{formpW}
\end{split}\end{equation}
then {(in connection with \eqref{cformC} and \eqref{ccincC} as the only term remaining after factors presented in \eqref{cincfacsC} and \eqref{cpfacsC})} by the splitting suggested in \eqref{deffminus} and according to \eqref{phipsirigid},
it is easy to see that 
\begin{equation}\begin{split}
\boldsymbol{c}^{\saux}&={\boldsymbol{c}^{\saux-}+\boldsymbol{c}^{\saux+}}\\
&=-(\recip{{{\mathbf{D}}}}_--\recip{{{\mathbf{D}}}}_-(0)){\mathbf{J}}\boldsymbol{p}^{\saux0}-{z}(\recip{{{\mathbf{D}}}}_--{{\mathbf{D}}}_+(\infty)){\mathbf{J}}\boldsymbol{p}^{\saux1}\\
&+{z}({{\mathbf{D}}}_+-{{\mathbf{D}}}_+(\infty)){\mathbf{J}}\boldsymbol{p}^{\saux1}+({{\mathbf{D}}}_+-\recip{{{\mathbf{D}}}}_-(0)){\mathbf{J}}\boldsymbol{p}^{\saux0}\\
&{+(-{z}^{-{\mathtt{M}}}{\su}^{{t}}_{{\mathtt{M}}-1, {\mathtt{N}}}){{\mathbf{D}}}_+{\mathbf{J}}\ensuremath{\hat{\mathbf{e}}}_2-\frac{1}{\sqrt{2}}{\su}^{{t}}_{{\mathtt{M}}-1, {\mathtt{N}}}\begin{bmatrix}
\phi_{{\mathtt{M}}}^-+\phi_{{\mathtt{M}}}^+\\
-\psi_{{\mathtt{M}}}^--\psi_{{\mathtt{M}}}^+
\end{bmatrix},}
\label{cauxfacs}
\end{split}\end{equation}
{
\begin{equation}\begin{split}
\text{where }\phi_{{\mathtt{M}}}=\recip{{\upalpha}}_-{z}^{-{\mathtt{M}}}&=f^-{z}^{-{\mathtt{M}}}
=\phi_{{\mathtt{M}}}^{+}+\phi_{{\mathtt{M}}}^{-},\\
\psi_{{\mathtt{M}}}=\recip{{\upbeta}}_-{z}^{-{\mathtt{M}}}&=g^-{z}^{-{\mathtt{M}}}
=\psi_{{\mathtt{M}}}^{+}+\psi_{{\mathtt{M}}}^{-}.
\label{cauxfacsphiM}
\end{split}\end{equation}
}
{
\begin{remark}
According to the definition of $\boldsymbol{p}^{\saux1}$ in \eqref{formpW}, a straightforward calculation reveals that it equals the limit of $\boldsymbol{q}^{{\mathrm{inc}}}{}^+$ stated in \eqref{qPcAbra} as ${z}\to\infty$. Thus,
$\boldsymbol{p}^{\saux1}=\boldsymbol{q}^{{\mathrm{inc}}}{}^+(\infty)$.
\label{pW1rem}
\end{remark}
}

{
\begin{remark}
Following up the observation stated in Remark \ref{perturbrigid}, using a form similar to ${\saux}_0$ from \eqref{defW0Q}${}_1$ (recall Remark \ref{u00remark}, i.e., ${\saux}_0({{z}})=-{\su}_{-1, 0}-{{z}} {\su}^{{\mathrm{inc}}}_{0, 0}$), it is easy to see that ${\saux}_{{\mathtt{N}}}$ from \eqref{WcN} can be expressed as 
\begin{equation}\begin{split}
{\saux}_{{\mathtt{N}}}
=-{\su}_{-1, {\mathtt{N}}}-{{z}} {\su}^{{\mathrm{inc}}}_{0, {\mathtt{N}}}+\pertbn{\saux}_{{\mathtt{N}}},
\text{where }\pertbn{\saux}_{{\mathtt{N}}}=-{z}^{-{\mathtt{M}}}{\su}^{t}_{{\mathtt{M}}-1, {\mathtt{N}}}+{\su}^{t}_{-1, {\mathtt{N}}},
\label{Wauxperturb}
\end{split}\end{equation}
so that $\boldsymbol{p}^{\saux}$ from \eqref{qPcAbra}${}_2$ can be re-written as
$$
\boldsymbol{p}^{\saux}=\begin{bmatrix}
{\saux}_{0}\\
{\saux}_{{\mathtt{N}}}
\end{bmatrix}=\begin{bmatrix}
-{\su}_{-1, 0}-{{z}} {\su}^{{\mathrm{inc}}}_{0, 0}\\
-{\su}_{-1, {\mathtt{N}}}-{{z}} {\su}^{{\mathrm{inc}}}_{0, {\mathtt{N}}}
\end{bmatrix}
+\pertbn{\saux}_{{\mathtt{N}}}\ensuremath{\hat{\mathbf{e}}}_2.
$$
In light of this, the expression $\boldsymbol{c}^{\saux}$ from \eqref{caux} motivates the definitions
\begin{equation}\begin{split}
\mathring{\boldsymbol{c}}^{\saux}&=-(\recip{{{\mathbf{D}}}}_-{\mathbf{J}}-{{\mathbf{D}}}_+{\mathbf{J}})\begin{bmatrix}
-{\su}_{-1, 0}-{{z}} {\su}^{{\mathrm{inc}}}_{0, 0}\\
-{\su}_{-1, {\mathtt{N}}}-{{z}} {\su}^{{\mathrm{inc}}}_{0, {\mathtt{N}}}
\end{bmatrix}, \\
\pertbn{\boldsymbol{c}}^{\saux}{\,}&:=-
\pertbn{\saux}_{{\mathtt{N}}}(\recip{{{\mathbf{D}}}}_-{\mathbf{J}}-{{\mathbf{D}}}_+{\mathbf{J}})\ensuremath{\hat{\mathbf{e}}}_2,
\label{cauxperturb}
\end{split}\end{equation}
so that
$$
\boldsymbol{c}^{\saux}=\mathring{\boldsymbol{c}}^{\saux}+\pertbn{\boldsymbol{c}}^{\saux}.
$$
It is noted that $\pertbn{\boldsymbol{c}}^{\saux}+\boldsymbol{c}^{\redc}$ describes the effect of a perturbation introduced by offset ${\mathtt{M}}.$ 
A re-look at \eqref{WHCequations} confirms that the exact solution for $\mathtt{M}\ne0$ can be expressed as a superposition of that for the case $\mathtt{M}=0$ \cite{Bls8pair1} as it corresponds to the vanishing of $\pertbn{\boldsymbol{c}}^{\saux}+\boldsymbol{c}^{\redc}$, i.e., the result when the right hand side in the Wiener--Hopf equation only involves 
\begin{equation}\begin{split}
\mathring{\boldsymbol{c}}:=\boldsymbol{c}^{\old}+\mathring{\boldsymbol{c}}^{\saux}.
\label{calignedC}
\end{split}\end{equation}
The exact nature of the effect of stagger is thus brought out by $\pertbn{\boldsymbol{c}}^{\saux}+\boldsymbol{c}^{\redc}$ in which $\boldsymbol{c}^{\redc}$ is further studied below.
By addition of $\pm$ functions in \eqref{WHsolC}, 
$$\boldsymbol{{\mathrm{w}}}^{\mathrm{F}}=\boldsymbol{{\mathrm{w}}}^-+\boldsymbol{{\mathrm{w}}}^+={\mathbf{L}_-}\boldsymbol{c}^-+\recip{\mathbf{L}}_+\boldsymbol{c}^+=\mathring{\boldsymbol{{\mathrm{w}}}}^{\mathrm{F}}+\pertbn{\boldsymbol{{\mathrm{w}}}}^{\mathrm{F}},
$$
$$
\text{where }
\mathring{\boldsymbol{{\mathrm{w}}}}^{\mathrm{F}}
={\mathbf{L}_-}\mathring{\boldsymbol{c}}^{-}+\recip{\mathbf{L}}_+\mathring{\boldsymbol{c}}^{+},
$$
\begin{equation}\begin{split}
\text{and }
\pertbn{\boldsymbol{{\mathrm{w}}}}^{\mathrm{F}}={\mathbf{L}_-}(\pertbn{\boldsymbol{c}}^{\saux-}+\boldsymbol{c}^{\redc-})+\recip{\mathbf{L}}_+(\pertbn{\boldsymbol{c}}^{\saux+}+\boldsymbol{c}^{\redc+}).
\label{whatK}
\end{split}\end{equation}
\label{perturbrigid2}
\end{remark}
}

{\bf Additive factorization of $\boldsymbol{c}^{\redc}$:}
Let ${{}\mathbb{D}}$ denote the set $\mathbb{Z}_0^{{\mathtt{M}}-1}$ {(same as the definition \eqref{defnDcrack})}.
Let
\begin{equation}\begin{split}
\poly{P}^+({z}){\,:=}{{{\mathtt{f}}}}^{+}({z})+{{{\mathtt{f}}}}^{{\mathrm{inc}} {+}}({z})
=\sum\limits_{{\mathtt{x}}\in{{}\mathbb{D}}}{\mathrm{w}}^{{t}}_{{{\mathtt{x}}},{\mathtt{N}}}{z}^{-{\mathtt{x}}},
\label{polyPC}
\end{split}\end{equation}
{so that by \eqref{qPcAbra}}
{
$$
\boldsymbol{p}^{\sgnM}=\boldsymbol{p}^{+}=\begin{bmatrix}
0\\
\poly{P}^+({z})
\end{bmatrix},
$$
}
according to \eqref{ppnC} and \eqref{ppnincC}.
{In relation to the term involving $\boldsymbol{p}^{+}$ in \eqref{cformC} (for preparing additive factors of $\boldsymbol{c}^{\redc}$ defined in \eqref{ccincC}), it is useful to expand}
\begin{equation}\begin{split}
-{\recip{{\mathbf{D}}}_-({z})}{\mathbf{J}}\boldsymbol{p}^{+}({z})
=\frac{1}{\sqrt{2}}\begin{bmatrix}
\recip{{\upalpha}}_-\poly{P}^+({z})\\
-\recip{{\upbeta}}_-\poly{P}^+({z})
\end{bmatrix},
\label{difficulttermMp1C}
\end{split}\end{equation}
which has right hand side of the same form as \eqref{difficulttermMp1}.
Using the splitting suggested in \eqref{deffminus}, 
\begin{equation}\begin{split}\recip{{\upalpha}}_-\poly{P}^+&=f^-\poly{P}^+
=\sum\limits_{{\mathtt{x}}\in{{}\mathbb{D}}}{\mathrm{w}}^{{t}}_{{{\mathtt{x}}},{\mathtt{N}}}\phi_{{\mathtt{x}}}=\sum\limits_{{\mathtt{x}}\in{{}\mathbb{D}}}{\mathrm{w}}^{{t}}_{{{\mathtt{x}}},{\mathtt{N}}}(\phi_{{\mathtt{x}}}^{+}+\phi_{{\mathtt{x}}}^{-}),\\
\recip{{\upbeta}}_-\poly{P}^+&=g^-\poly{P}^+
=\sum\limits_{{\mathtt{x}}\in{{}\mathbb{D}}}{\mathrm{w}}^{{t}}_{{{\mathtt{x}}},{\mathtt{N}}}\psi_{{\mathtt{x}}}=\sum\limits_{{\mathtt{x}}\in{{}\mathbb{D}}}{\mathrm{w}}^{{t}}_{{{\mathtt{x}}},{\mathtt{N}}}(\psi_{{\mathtt{x}}}^{+}+\psi_{{\mathtt{x}}}^{-}).\label{phipsirigid}\end{split}\end{equation}
Thus, \eqref{cpfacs} holds except that $+{{\mathcal{L}}}_{{c}}^+$ replaces $-{{\mathcal{L}}}_{{k}}^+$ and ${\mathrm{w}}$ replaces ${\mathrm{v}}$.
{In particular, in terms of the definition of $\boldsymbol{c}^{\redc}$ in \eqref{ccincC} (with $\sgnM=+1$),
\begin{equation}\begin{split}
{\sqrt{2}}\boldsymbol{c}^{\redc}&={\sqrt{2}}\boldsymbol{c}^{\redc}{}^-+{\sqrt{2}}\boldsymbol{c}^{\redc}{}^+=\begin{bmatrix}
\sum\limits_{{\mathtt{x}}\in{{}\mathbb{D}}}{\mathrm{w}}^{{t}}_{{{\mathtt{x}}},{\mathtt{N}}}\phi_{{\mathtt{x}}}^{-}\\
-\sum\limits_{{\mathtt{x}}\in{{}\mathbb{D}}}{\mathrm{w}}^{{t}}_{{{\mathtt{x}}},{\mathtt{N}}}\psi_{{\mathtt{x}}}^{-}\\
\end{bmatrix}+
{\sqrt{2}}\boldsymbol{\tau}^+
+{{{\mathbf{D}}}_+{\mathbf{J}}}
\boldsymbol{p}^+({z}),
\label{cpfacsC}
\end{split}\end{equation}
\begin{equation}\begin{split}
\text{where }
\boldsymbol{\tau}^+=\frac{1}{\sqrt{2}}\begin{bmatrix}
\sum\limits_{{\mathtt{x}}\in{{}\mathbb{D}}}{\mathrm{w}}^{{t}}_{{{\mathtt{x}}},{\mathtt{N}}}\phi_{{\mathtt{x}}}^{+}\\
-\sum\limits_{{\mathtt{x}}\in{{}\mathbb{D}}}{\mathrm{w}}^{{t}}_{{{\mathtt{x}}},{\mathtt{N}}}\psi_{{\mathtt{x}}}^{+}
\end{bmatrix},
\label{taurigidP1}
\end{split}\end{equation}
}

{\bf Equation for $\{{\mathrm{w}}_{{{\mathtt{x}}},{\mathtt{N}}}\}_{{\mathtt{x}}\in{{}\mathbb{D}}}$:}
{In order to finally arrive at the equation determining the coefficients of the polynomial \eqref{polyPC}}, by virtue of \eqref{WHsolC} the set of first ${\mathtt{M}}$ Fourier coefficients of the second component of $\boldsymbol{{\mathrm{w}}}^+$ {(which equals $\ensuremath{\hat{\mathbf{e}}}_2\cdot\recip{\mathbf{L}}_+\boldsymbol{c}^+$)}, namely, {$\mathfrak{P}_{{}\mathbb{D}}(\ensuremath{\hat{\mathbf{e}}}_2\cdot\boldsymbol{{\mathrm{w}}}^+)$ need to be evaluated, where the projection operator defined in \eqref{defdomainproj} is utlized. Thus, the condition}
\begin{equation}\begin{split}
{\mathfrak{P}_{{}\mathbb{D}}}(\ensuremath{\hat{\mathbf{e}}}_2\cdot\boldsymbol{{\mathrm{w}}}^+)={\mathfrak{P}_{{}\mathbb{D}}}(\ensuremath{\hat{\mathbf{e}}}_2\cdot\recip{\mathbf{L}}_+(\boldsymbol{c}^{\old}{}^++\boldsymbol{c}^{\saux}{}^++\boldsymbol{c}^{\redc}{}^+))
\label{reducedrigidP}
\end{split}\end{equation}
is {precisely the foundation of an ${\mathtt{M}}\times{\mathtt{M}}$ system of linear algebraic equations for $\{{\mathrm{w}}_{{{\mathtt{x}}},{\mathtt{N}}}\}_{{\mathtt{x}}\in{{}\mathbb{D}}}$ (recall \eqref{defnDcrack}); however, there remain two more unknowns for two more equations are needed}.

{
From \eqref{cincfacsC}, \eqref{cauxfacs}, and \eqref{cpfacsC}, respectively,
with \eqref{taurigidP1}
(recall \eqref{Dfactorsrigid} and \eqref{Lfactorsfull} too)
\begin{equation}\begin{split}
\recip{\mathbf{L}}_+\boldsymbol{c}^{\old}{}^+
&=-\big(-{\mathtt{Q}}{\mathbf{I}}+{\mathtt{Q}}({z}_{{P}})\recip{{\mathbf{J}}}\recip{{{\mathbf{D}}}}_+\recip{{{\mathbf{D}}}}_-({z}_{{P}}){\mathbf{J}}\\
&+({z}_{{P}}^{-1}-{z}^{-1})\recip{{\mathbf{J}}}\recip{{{\mathbf{D}}}}_+\recip{{{\mathbf{D}}}}_-(0){\mathbf{J}}+({z}_{{P}}-{z})\recip{{\mathbf{J}}}\recip{{{\mathbf{D}}}}_+{{\mathbf{D}}}_+(\infty){\mathbf{J}}\big)\boldsymbol{q}^{{\mathrm{inc}}}{}^+,\\
\recip{\mathbf{L}}_+\boldsymbol{c}^{\saux+}
&=({\mathbf{I}}-\recip{{\mathbf{J}}}\recip{{{\mathbf{D}}}}_+\recip{{{\mathbf{D}}}}_-(0){\mathbf{J}})\boldsymbol{p}^{\saux0}+{z}({\mathbf{I}}-\recip{{\mathbf{J}}}\recip{{{\mathbf{D}}}}_+{{\mathbf{D}}}_+(\infty){\mathbf{J}})\boldsymbol{p}^{\saux1}\\
&+(-{z}^{-{\mathtt{M}}}{\su}^{{t}}_{{\mathtt{M}}-1, {\mathtt{N}}})\ensuremath{\hat{\mathbf{e}}}_2-\frac{1}{\sqrt{2}}{\su}^{{t}}_{{\mathtt{M}}-1, {\mathtt{N}}}
\recip{{\mathbf{J}}}\recip{{{\mathbf{D}}}}_+\begin{bmatrix}
\phi_{{\mathtt{M}}}^+\\
-\psi_{{\mathtt{M}}}^+
\end{bmatrix},\\
\recip{\mathbf{L}}_+\boldsymbol{c}^{\redc}{}^+
&=\recip{{\mathbf{J}}}\recip{{{\mathbf{D}}}}_+\boldsymbol{\tau}^+
+\boldsymbol{p}^+({z}).
\label{LcincWcf}
\end{split}\end{equation}
}
{The sum of above gives the $\boldsymbol{{\mathrm{w}}}^+$ (i.e., $\recip{\mathbf{L}}_+\boldsymbol{c}^+$) in \eqref{reducedrigidP} while the same equals $\recip{{\mathbf{J}}^+}\recip{{{\mathbf{D}}}}_+(\boldsymbol{c}^{\old}{}^++\boldsymbol{c}^{\saux}{}^++\boldsymbol{c}^{\redc}{}^+)$.
Using the definition of ${\mathtt{Q}}$, a relevant expansion is ${\mathtt{Q}}{\su}^{\mathrm{inc}}_{{\mathtt{N}}}{}^{+}({z})=(4-z-z^{-1}-\upomega^2)\sum_{z\in\mathbb{Z}^+}{z}^{-\mathtt{x}}{\su}^{\mathrm{inc}}_{\mathtt{x}, {\mathtt{N}}}=-{z}\su^{{\mathrm{inc}}}_{0, {\mathtt{N}}}+\su^{{\mathrm{inc}}}_{-1, {\mathtt{N}}}+{{\mathtt{w}}}_{{\mathtt{N}}}^{{\mathrm{inc}}+}({z})$ (the same can be found using \eqref{anidentity1inc} for the choice ${\mathtt{M}}\to\infty$).
Hence, taking a part of the sum of the three expressions in \eqref{LcincWcf},
\begin{equation}\begin{split}
\ensuremath{\hat{\mathbf{e}}}_2\cdot({\mathtt{Q}}\boldsymbol{q}^{{\mathrm{inc}}}{}^++\boldsymbol{p}^{\saux0}+{z}\boldsymbol{p}^{\saux1}+\boldsymbol{p}^+({z}))&=-{\mathtt{Q}}{\su}^{{\mathrm{inc}}}_{{\mathtt{N}}}{}^{+}+{\su}^{{\mathrm{inc}}}_{-1, {\mathtt{N}}}-{z}{\su}^{{\mathrm{inc}}}_{0, {\mathtt{N}}}+{{\mathtt{f}}}^+({z})+{{\mathtt{f}}}^{{\mathrm{inc}}}{}^+({z})\\
&=-{{\mathtt{w}}}_{{\mathtt{N}}}^{{\mathrm{inc}}+}({z})+{{\mathtt{f}}}^+({z})+{{\mathtt{f}}}^{{\mathrm{inc}}}{}^+({z}).
\label{relatedtoidentity1inc}
\end{split}\end{equation}
It can be observed that $\mathfrak{P}_{{}\mathbb{D}}(-{{\mathtt{w}}}_{{\mathtt{N}}}^{{\mathrm{inc}}+}({z})+{{\mathtt{f}}}^{{\mathrm{inc}}}{}^+({z}))=0$ in \eqref{relatedtoidentity1inc}.
Indeed, using above and expanding and re-arranging the terms in \eqref{reducedrigidP} further,
\begin{subequations}\begin{equation}\begin{split}
&\mathfrak{P}_{{}\mathbb{D}}({{\mathrm{w}}}_{{\mathtt{N}}}^+({z})-{{\mathtt{f}}}^+({z}))\\
&=\mathfrak{P}_{{}\mathbb{D}}({\frac{1}{2}}{\mathtt{Q}}({z}_{{P}})\mathcal{F}^{{\mathrm{inc}}}({z})-{\frac{1}{2}}\sum\limits_{{\mathtt{x}}\in{{}\mathbb{D}}}{\mathrm{w}}^{{t}}_{{{\mathtt{x}}},{\mathtt{N}}}\mathcal{A}_{{\mathtt{x}}}({z})\\
&-{\su}^{{t}}_{-1, 0}\mathcal{J}(z)-{\su}^{{t}}_{{\mathtt{M}}-1, {\mathtt{N}}}(\mathcal{K}(z)+{z}^{-{\mathtt{M}}})+\mathcal{G}^{{\mathrm{inc}}}(z)),
\label{eqnMbyMrigidpre}
\end{split}\end{equation}
\begin{equation}\begin{split}
\text{where }
\mathcal{A}_{{\mathtt{x}}}({z})&{\,:=}\big(\phi_{{\mathtt{x}}}^{+}({z})\recip{{\upalpha}}_+({z})+\psi_{{\mathtt{x}}}^{+}({z})\recip{{\upbeta}}_+({z})\big),
\label{Axrigid}
\end{split}\end{equation}
\begin{equation}\begin{split}
\mathcal{F}^{{\mathrm{inc}}}({z})&{\,:=}\big(- (1-e^{i\upkappa_y{\mathtt{N}}})\recip{{\upalpha}}_-({z}_{{P}})\recip{{\upalpha}}_+({z})\\
&+ (1+e^{i\upkappa_y{\mathtt{N}}})\recip{{\upbeta}}_-({z}_{{P}})\recip{{\upbeta}}_+({z})\big)e^{-i\upkappa_y{\mathtt{N}}}{\su}_{{\mathtt{N}}}^{{\mathrm{inc}}}{}^+({z}),
\label{Fincrigid}
\end{split}\end{equation}
\begin{equation}\begin{split}
\mathcal{J}(z)&{\,:=}-\frac{1}{2}\big(-\recip{{\upalpha}}_-({0})\recip{{\upalpha}}_+({z})+\recip{{\upbeta}}_-({0})\recip{{\upbeta}}_+({z})\big),
\label{KrigidP}
\end{split}\end{equation}
\begin{equation}\begin{split}
\mathcal{K}(z)&{\,:=}-\frac{1}{2}(\phi_{{\mathtt{M}}}^+\recip{{\upalpha}}_+({z})+\psi_{{\mathtt{M}}}^+\recip{{\upbeta}}_+({z})),
\label{Grigid}
\end{split}\end{equation}
\text{and }
\begin{equation}\begin{split}
\mathcal{G}^{{\mathrm{inc}}}(z)&{\,:=}\frac{1}{2}\big(-(1-e^{i\upkappa_y{\mathtt{N}}})\recip{{\upalpha}}_-({0})\recip{{\upalpha}}_+({z})\\
&+(1+e^{i\upkappa_y{\mathtt{N}}})\recip{{\upbeta}}_-({0})\recip{{\upbeta}}_+({z})\big)e^{-i\upkappa_y{\mathtt{N}}}\{({z}_{{P}}^{-1}-{z}^{-1}){\su}_{{\mathtt{N}}}^{{\mathrm{inc}}}{}^+({z})-{\su}_{-1,{\mathtt{N}}}^{{\mathrm{inc}}}\}\\
&+\frac{1}{2}\big(-(1-e^{i\upkappa_y{\mathtt{N}}}){{\upalpha}_+}({\infty})\recip{{\upalpha}}_+({z})\\
&+(1+e^{i\upkappa_y{\mathtt{N}}}){{\upbeta}_+}({\infty})\recip{{\upbeta}}_+({z})\big)e^{-i\upkappa_y{\mathtt{N}}}\{({z}_{{P}}-{z}){\su}_{{\mathtt{N}}}^{{\mathrm{inc}}}{}^+({z})+{z}{\su}^{{\mathrm{inc}}}_{0, {\mathtt{N}}}\}.
\label{Gincrigid}
\end{split}\end{equation}
\label{eqnMbyMrigidprefull} 
\end{subequations}
Using the expression of ${\su}_{{\mathtt{N}}}^{{\mathrm{inc}}}{}^+$, after a sequence of elementary manipulations it is found that the term in curly brackets is zero, so that $\mathcal{G}^{{\mathrm{inc}}}$ is further simplified to be zero. 
In view of the definitions of ${{{\mathtt{f}}}}^+$ and ${{{\mathtt{f}}}}^{{\mathrm{inc}}}{}^+$ given by \eqref{defp1}, it is easy to see that above equation \eqref{eqnMbyMrigidpre} leads to
\begin{equation}\begin{split}
\sum\limits_{{\mathtt{x}}\in{{}\mathbb{D}}}{\mathrm{w}}^{{t}}_{{\mathtt{x}},{\mathtt{N}}}\mathfrak{P}_{{}\mathbb{D}}(\mathcal{A}_{{\mathtt{x}}})
&={\mathtt{Q}}({z}_{{P}})\mathfrak{P}_{{}\mathbb{D}}\big(\mathcal{F}^{{\mathrm{inc}}}\big)-2{\su}^{{t}}_{-1, 0}\mathfrak{P}_{{}\mathbb{D}}(\mathcal{J})-2{\su}^{{t}}_{{\mathtt{M}}-1, {\mathtt{N}}}\mathfrak{P}_{{}\mathbb{D}}(\mathcal{K}),
\label{eqnMbyMrigid}
\end{split}\end{equation}
which is a ${\mathtt{M}}\times{\mathtt{M}}$ system of linear algebraic equations for $\{{\mathrm{w}}^{{t}}_{{\mathtt{x}},{\mathtt{N}}}\}_{{\mathtt{x}}\in{{}\mathbb{D}}}$, i.e., 
$\{{\mathrm{w}}_{{\mathtt{x}},{\mathtt{N}}}\}_{{{}\mathbb{D}}}$ since $\{{\mathrm{w}}^{{\mathrm{inc}}}_{{\mathtt{x}},{\mathtt{N}}}\}_{{{}\mathbb{D}}}$ are known.}
\begin{remark}
{It is a neat result that the coefficient matrix $\mathcal{A}_{{\mathtt{x}}}$ retains the same form in both types of defects as seen from \eqref{Axcrack} vs \eqref{Axrigid}; also compare \eqref{eqnMbyMcrack} and \eqref{eqnMbyMrigid}. This is not surprising in view of the fact that both kinds of defects leads to effectively the same kernel (modulo a scalar factor). However, it also points towards a possibility of preparing an exact matrix Wiener--Hopf factors of such special kernels \cite{GMthesis,Bls8staggerpair_asymp}; the task is deferred to another forum in future.}
\end{remark}

{\bf Equation based on evaluation of $\su^{t}_{-1, 0}$ and $\su^{t}_{{\mathtt{M}}-1, {\mathtt{N}}}$:}
{The equation \eqref{reducedrigidP}, an analogue of \eqref{reducedcrackP} for cracks, is {\em the reduced algebraic problem} for the rigid constraints except that it needs to be supplemented by two more equations. In other words, the analogue of the equation \eqref{eqnMbyMcrack} for the rigid constraints needs to be coupled to two more equations.}
{Revisiting the expression of the right hand side of the Wiener--Hopf equation in \eqref{cformC} (mainly the term involving $\boldsymbol{p}^{\saux}$ \eqref{formpW}),} there still remain two unknowns in $\boldsymbol{p}^{+}$, i.e., $\su_{-1, 0}$ and $\su_{{\mathtt{M}}-1, {\mathtt{N}}}$, as the total number of unknowns in the reduced problem is ${{\mathtt{M}}}+2$.
The equations \eqref{u0expn} and \eqref{uNexpn} {encapsulate the scattered field in the half-row facing the rigid constraint, and thereby} yield the conditions relevant for the last two unknowns. In particular,
by virtue of \eqref{WHsolC} {(as $\boldsymbol{{\mathrm{w}}}^{{\mathrm{F}}}=\boldsymbol{{\mathrm{w}}}^++\boldsymbol{{\mathrm{w}}}^-$)} {and the inverse Fourier transform \cite{Bls0,Bls1,Bls2,Bls3}},
\begin{equation}\begin{split}
\su_{-1, 0}&=\frac{1}{2\pi i}\int_{{\mathcal{C}}}(\recip{{\mathtt{Q}}}({\saux}_0+\ensuremath{\hat{\mathbf{e}}}_1\cdot\mathbf{L}^-(\boldsymbol{c}^{\old}{}^-+\boldsymbol{c}^{\saux}{}^-+\boldsymbol{c}^{\redc}{}^-))\\
&-{{\su}^{{\mathrm{inc}}}_{0, 0}}\delta_{D}^{+}({{z}} {z}_{{P}}^{-1})){z}^{-2}d{z},
\label{un10pre}
\end{split}\end{equation}
\begin{equation}\begin{split}
\text{and }
\su_{{\mathtt{M}}-1, {\mathtt{N}}}&=\frac{1}{2\pi i}\int_{{\mathcal{C}}}(\recip{{\mathtt{Q}}}({\saux}_{{\mathtt{N}}}+\ensuremath{\hat{\mathbf{e}}}_2\cdot\mathbf{L}^-(\boldsymbol{c}^{\old}{}^-+\boldsymbol{c}^{\saux}{}^-+\boldsymbol{c}^{\redc}{}^-))\\
&+\recip{{\mathtt{Q}}}\poly{P}^{+}-{\su}^{{\mathrm{inc}}}_{0, {\mathtt{N}}}\delta_{D}^{+}({{z}} {z}_{{P}}^{-1})){z}^{{\mathtt{M}}-2}d{z},
\label{uMn10Ppre}
\end{split}\end{equation}
where ${\mathcal{C}}$ is a counter-clockwise contour in the annulus ${\mathscr{A}}$. Above give the required equations that need to be solved for $\su_{-1, 0}$ and $\su_{{\mathtt{M}}-1, {\mathtt{N}}}$ in conjunction with that obtained from the reduced equation \eqref{reducedrigidP} via the projection operator defined in \eqref{defdomainproj}.
{The remaining symbolism leading to a matrix formulation follows that presented for two cracks with positive offset, i.e., \eqref{aknueqnPcrack}.}

{With the details provided in Appendix \ref{ApprigidMPextra}, it takes some effort to further simplify \eqref{un10pre} and \eqref{uMn10Ppre} so that \eqref{eqnMbyMrigid} leads to
\begin{subequations}
\begin{equation}\begin{split}
\sum_{\nu=1}^{{\mathtt{M}}}a_{{{\mu}}\nu}\chi_\nu+a_{{{\mu}}({\mathtt{M}}+1)}\chi_{\mathtt{M}+1}+a_{{{\mu}}({\mathtt{M}}+2)}\chi_{\mathtt{M}+2}=b_{{\mu}}\quad\quad({{\mu}}, \dotsc, {\mathtt{M}})
\label{aknueqnPrigid}
\end{split}\end{equation}
\begin{equation}\begin{split}
\begin{aligned}
\text{where }
a_{{{\mu}}\nu}&=\mathfrak{C}_{{{\mu}}-1}\big(\mathfrak{P}_{{}\mathbb{D}}(\mathcal{A}_{\nu-1})\big)\\
a_{{{\mu}}({\mathtt{M}}+1)}&=2\mathfrak{C}_{{{\mu}}-1}(\mathfrak{P}_{{}\mathbb{D}}(\mathcal{J})),\quad
a_{{{\mu}}({\mathtt{M}}+2)}=2\mathfrak{C}_{{{\mu}}-1}(\mathfrak{P}_{{}\mathbb{D}}(\mathcal{K})),\\
\chi_\nu&={\mathrm{w}}^{{t}}_{\nu-1,{\mathtt{N}}},\quad \chi_{\mathtt{M}+1}={\su}^{{t}}_{-1, 0},\quad \chi_{\mathtt{M}+2}={\su}^{{t}}_{{\mathtt{M}}-1, {\mathtt{N}}},\\
b_{{\mu}}&={\mathtt{Q}}({z}_{{P}})\mathfrak{C}_{{{\mu}}-1}(\mathfrak{P}_{{}\mathbb{D}}\big(\mathcal{F}^{{\mathrm{inc}}}\big)),
\end{aligned}
\end{split}\end{equation}
while,
\eqref{un10} and \eqref{uMn10P}, lead to, respectively,
\begin{equation}\begin{split}
\sum_{\nu=1}^{{\mathtt{M}}}a_{({\mathtt{M}}+1)\nu}\chi_\nu+a_{({\mathtt{M}}+1)({\mathtt{M}}+1)}\chi_{\mathtt{M}+1}+a_{({\mathtt{M}}+1)({\mathtt{M}}+2)}\chi_{\mathtt{M}+2}=b_{({\mathtt{M}}+1)},
\end{split}\end{equation}
\begin{equation}\begin{split}
\text{and }
\sum_{\nu=1}^{{\mathtt{M}}}a_{({\mathtt{M}}+2)\nu}\chi_\nu+a_{({\mathtt{M}}+2)({\mathtt{M}}+1)}\chi_{\mathtt{M}+1}+a_{({\mathtt{M}}+2)({\mathtt{M}}+2)}\chi_{\mathtt{M}+2}=b_{({\mathtt{M}}+2)},
\end{split}\end{equation}
\begin{equation}\begin{split}
\begin{aligned}
\text{with }
a_{({\mathtt{M}}+1)\nu}&=-\frac{1}{2}\big(\phi_{{\nu-1}}^{-}({z}_q){{\upalpha}_-}({z}_q)-\psi_{{\nu-1}}^{-}({z}_q){{\upbeta}_-}({z}_q)\big), \\
a_{({\mathtt{M}}+1)({\mathtt{M}}+1)}&=\frac{1}{2}\big(\recip{{\upalpha}}_-(0){{\upalpha}_-}({z}_q)+\recip{{\upbeta}}_-(0){{\upbeta}_-}({z}_q)\big),\\
a_{({\mathtt{M}}+1)({\mathtt{M}}+2)}&=\frac{1}{2}\big(\phi_{{\mathtt{M}}}^{-}({z}_q){{\upalpha}_-}({z}_q)-\psi_{{\mathtt{M}}}^{-}({z}_q){{\upbeta}_-}({z}_q)\big),
\end{aligned}
\end{split}\end{equation}
\begin{equation}\begin{split}
\begin{aligned}
a_{({\mathtt{M}}+2)\nu}&=-({z}_q^{-{({\nu-1})}}-\frac{1}{2}\big(\phi_{{{\nu-1}}}^{-}({z}_q){{\upalpha}_-}({z}_q)+\psi_{{{\nu-1}}}^{-}({z}_q){{\upbeta}_-}({z}_q)\big)),\\
a_{({\mathtt{M}}+2)({\mathtt{M}}+1)}&=\frac{1}{2}\big(-\recip{{\upalpha}}_-(0){{\upalpha}_-}({z}_q)+\recip{{\upbeta}}_-(0){{\upbeta}_-}({z}_q)\big),\\
a_{({\mathtt{M}}+2)({\mathtt{M}}+2)}&=({z}_q^{-{\mathtt{M}}}-\frac{1}{2}\big(\phi_{{\mathtt{M}}}^{-}({z}_q){{\upalpha}_-}({z}_q)+\psi_{{\mathtt{M}}}^{-}({z}_q){{\upbeta}_-}({z}_q)\big)),
\end{aligned}
\end{split}\end{equation}
\begin{equation}\begin{split}
\begin{aligned}
b_{({\mathtt{M}}+1)}&=-\frac{{z}_q}{{z}_q-{z}_{P}}{\frac{1}{2}}{\mathtt{Q}}({z}_{{P}})\big((1-e^{i\upkappa_y{\mathtt{N}}})\recip{{\upalpha}}_-({z}_{{P}}){{\upalpha}_-}({z}_q)\\
&+ (1+e^{i\upkappa_y{\mathtt{N}}})\recip{{\upbeta}}_-({z}_{{P}}){{\upbeta}_-}({z}_q)\big)e^{-i\upkappa_y{\mathtt{N}}}{\su}_{0,{\mathtt{N}}}^{{\mathrm{inc}}},\\
b_{({\mathtt{M}}+2)}&=-\frac{{z}_q}{{z}_q-{z}_{P}}{\frac{1}{2}}{\mathtt{Q}}({z}_{{P}})\big(- (1-e^{i\upkappa_y{\mathtt{N}}})\recip{{\upalpha}}_-({z}_{{P}}){{\upalpha}_-}({z}_q)\\
&+ (1+e^{i\upkappa_y{\mathtt{N}}})\recip{{\upbeta}}_-({z}_{{P}}){{\upbeta}_-}({z}_q)\big)e^{-i\upkappa_y{\mathtt{N}}}{\su}_{0,{\mathtt{N}}}^{{\mathrm{inc}}}.
\end{aligned}
\end{split}\end{equation}
\label{aknueqnPrigidcon}
\end{subequations}
Let $\recip{a}_{\nu{{\mu}}}$ denote the components of the inverse of $[a_{{{\mu}}\nu}]_{{{\mu}}, \nu=1, \dotsc, {\mathtt{M}}}$. Then
\begin{equation}\begin{split}
{\mathrm{w}}^{{t}}_{{\mathtt{x}},{\mathtt{N}}}&=\sum_{\nu=1}^{{\mathtt{M}}}{\mathtt{Q}}({z}_{{P}})\recip{a}_{({\mathtt{x}}+1){{\mu}}}\mathfrak{C}_{{{\mu}}-1}(\mathfrak{P}_{{}\mathbb{D}}\big(\mathcal{F}^{{\mathrm{inc}}}\big))-\frac{{z}_q}{{z}_q-{z}_{P}}{\frac{1}{2}}{\mathtt{Q}}({z}_{{P}})\\
&\times\big((\recip{a}_{({\mathtt{x}}+1)(\mathtt{M}+1)}-\recip{a}_{({\mathtt{x}}+1)(\mathtt{M}+2)})(1-e^{i\upkappa_y{\mathtt{N}}})\recip{{\upalpha}}_-({z}_{{P}}){{\upalpha}_-}({z}_q)\\
&+(\recip{a}_{({\mathtt{x}}+1)(\mathtt{M}+1)}+\recip{a}_{({\mathtt{x}}+1)(\mathtt{M}+2)}) (1+e^{i\upkappa_y{\mathtt{N}}})\recip{{\upbeta}}_-({z}_{{P}}){{\upbeta}_-}({z}_q)\big)e^{-i\upkappa_y{\mathtt{N}}}{\su}_{0,{\mathtt{N}}}^{{\mathrm{inc}}},
\label{exactwmsol}
\end{split}\end{equation}
provides the solution $\{{\mathrm{w}}^{{t}}_{{\mathtt{x}},{\mathtt{N}}}\}_{{\mathtt{x}}\in{{}\mathbb{D}}}$; similarly, the expression for ${\su}^{{t}}_{-1, 0}$, and ${\su}^{{t}}_{{\mathtt{M}}-1, {\mathtt{N}}}$ is obtained. 
}

\section{Numerical results and discussion}
\label{briefnumerics}

\begin{figure}[h!]
\centering
\includegraphics[width=\textwidth]{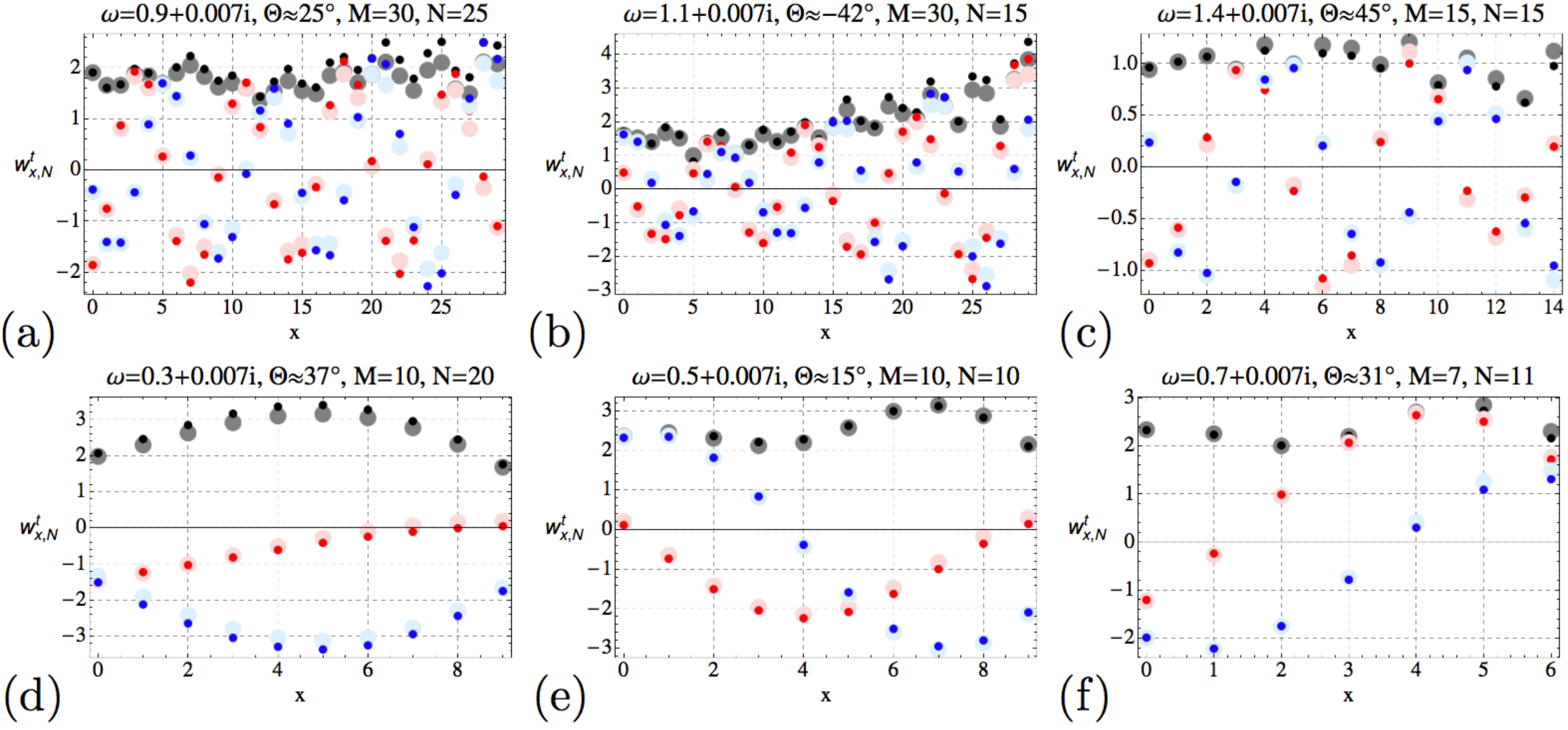}
\caption{Comparison of $\{{\mathrm{v}}^{{t}}_{{\mathtt{x}},{\mathtt{N}}}\}_{{\mathtt{x}}\in\mathbb{Z}_{{0}}^{{\mathtt{M}-1}}}$ based on analytical solution ($|v|$: gray dots, $\Re v$: light red dots, $\Im v$: light blue dots) vs numerical solution (resp. black dots, red dots, blue dots) {when $\mathtt{M}>0$}.}
\label{example1}
\end{figure}

\begin{figure}[h!]
\centering
\includegraphics[width=\textwidth]{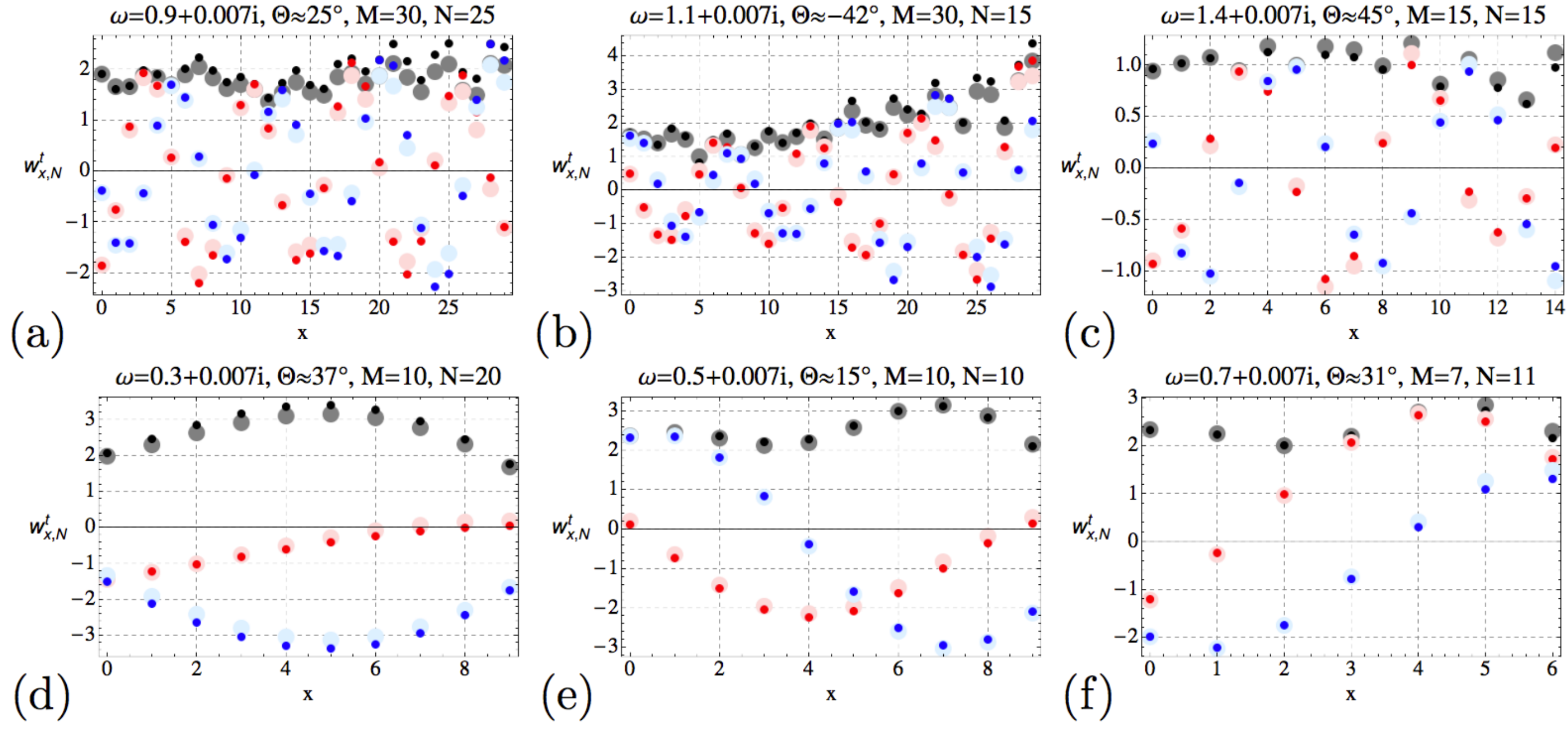}
\caption{Comparison of $\{{\mathrm{w}}^{{t}}_{{\mathtt{x}},{\mathtt{N}}}\}_{{\mathtt{x}}\in\mathbb{Z}_{{0}}^{{\mathtt{M}-1}}}$ based on analytical solution ($|w|$: gray dots, $\Re w$: light red dots, $\Im w$: light blue dots) vs numerical solution (resp. black dots, red dots, blue dots) {when $\mathtt{M}>0$}.}
\label{example2}
\end{figure}

\begin{figure}[h!]
\centering
\includegraphics[width=\textwidth]{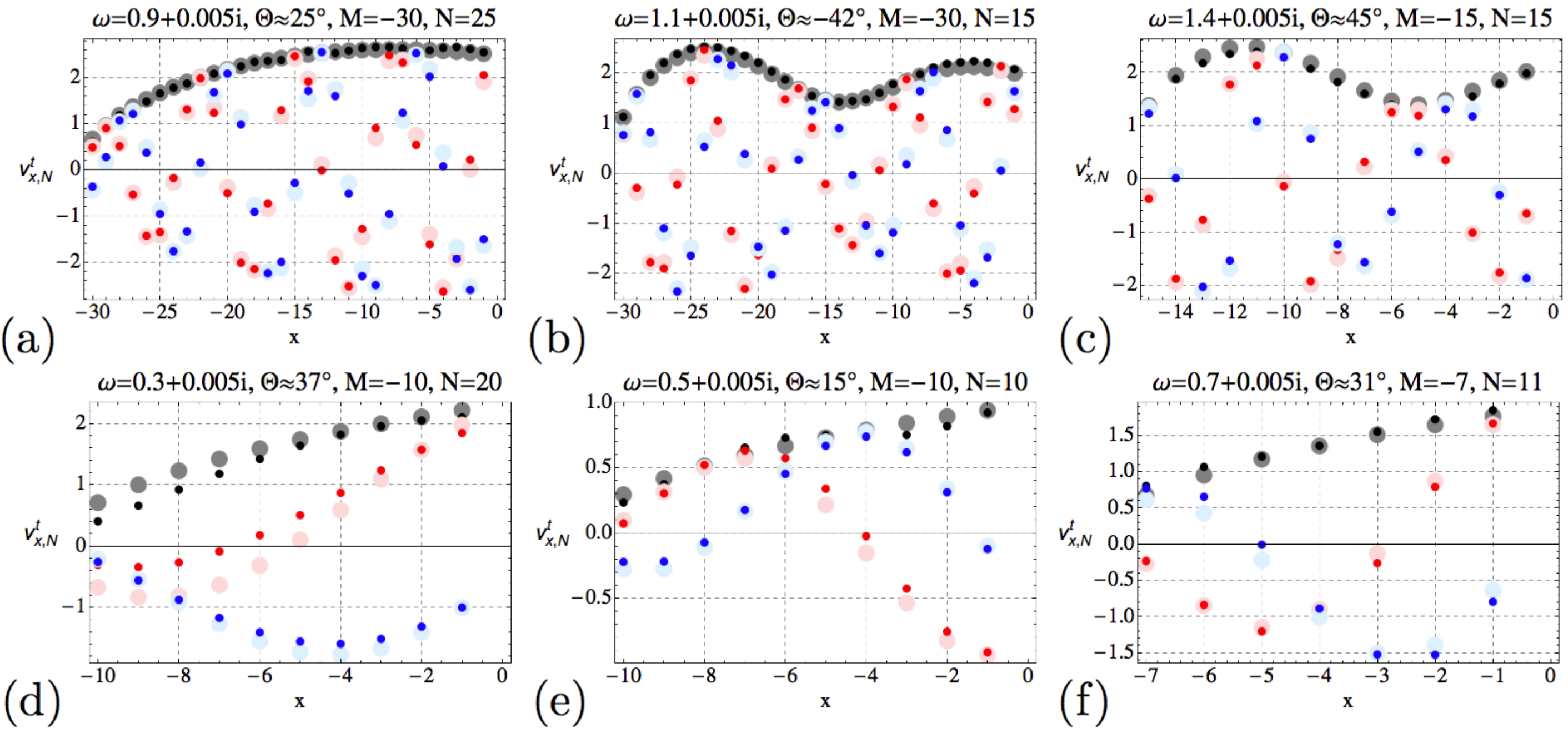}
\caption{Comparison of $\{{\mathrm{v}}^{{t}}_{{\mathtt{x}},{\mathtt{N}}}\}_{{\mathtt{x}}\in\mathbb{Z}_{{\mathtt{M}}}^{-1}}$ based on analytical solution ($|v|$: gray dots, $\Re v$: light red dots, $\Im v$: light blue dots) vs numerical solution (resp. black dots, red dots, blue dots) {when $\mathtt{M}<0$}.}
\label{example1p}
\end{figure}

\begin{figure}[h!]
\centering
\includegraphics[width=\textwidth]{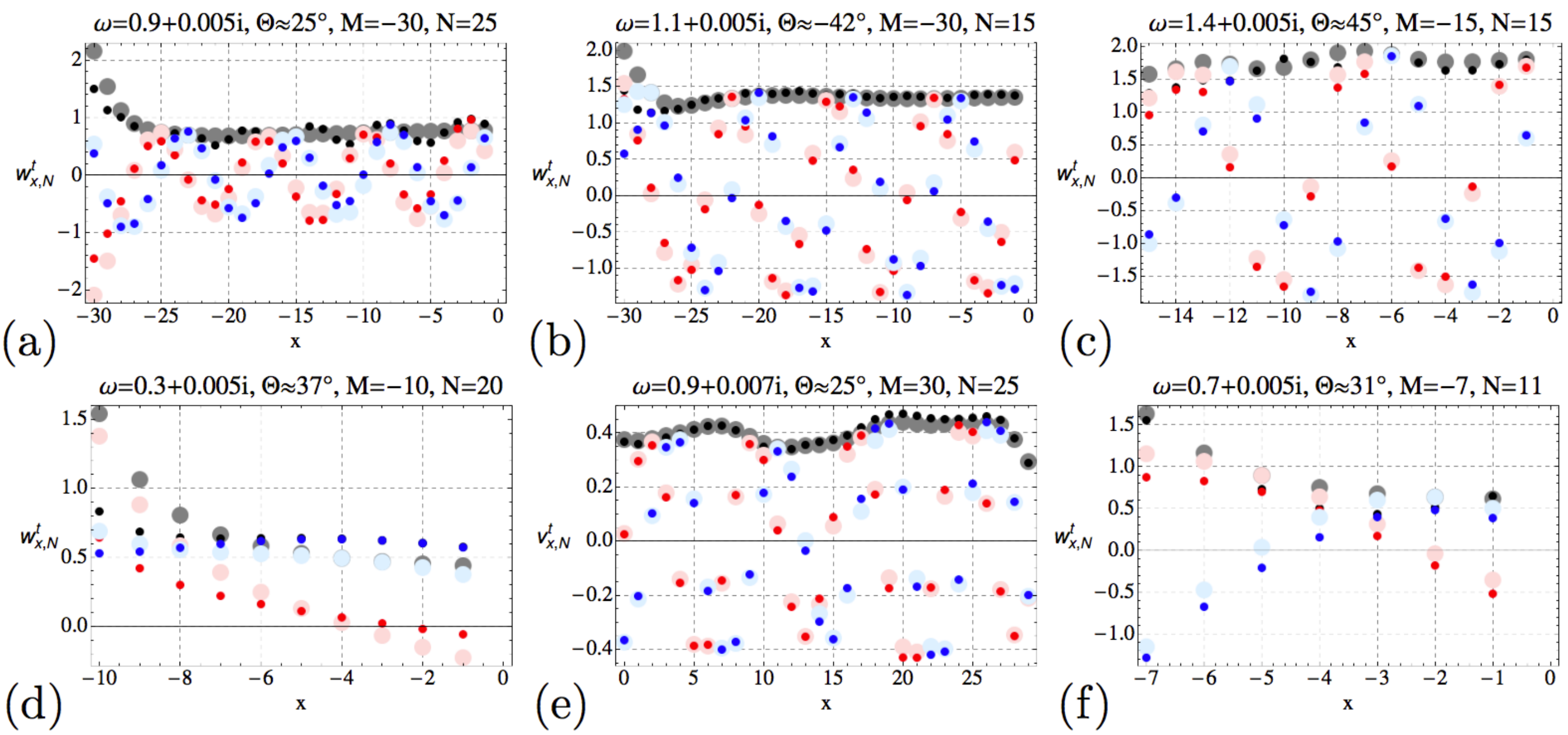}
\caption{Comparison of $\{{\mathrm{w}}^{{t}}_{{\mathtt{x}},{\mathtt{N}}}\}_{{\mathtt{x}}\in\mathbb{Z}_{{\mathtt{M}}}^{-1}}$ based on analytical solution ($|w|$: gray dots, $\Re w$: light red dots, $\Im w$: light blue dots) vs numerical solution (resp. black dots, red dots, blue dots) {when $\mathtt{M}<0$}.}
\label{example2p}
\end{figure}

{The pair of Fig. \ref{example1} and Fig. \ref{example2} (as well as Fig. \ref{example1p} and Fig. \ref{example2p}) provides examples of a comparison between numerical solution (i.e, using the scheme described in Appendix D of \cite{Bls0}) of the discrete scattering problem (on a grid described in caption of Fig. \ref{Fig1num}) and evaluation of the analytical solution for identical values of the parameters except 
for the 
imaginary part of $\upomega$.
As a part of a graphical illustration of the wave field on the lattice structure, the plots presented earlier in Fig. \ref{Fig1num}(a') and Fig. \ref{Fig1num}(b'), correspond to Fig. \ref{example1}(a) and Fig. \ref{example2}(a), respectively.}
{It is observed, with an allowance for small deviations in some cases for some sites, that the solution of the reduced algebraic problem coincides with a numerical solution (see Appendix D of \cite{Bls0}) of the discussed scattering problem.}

{The illustration in Fig. \ref{example1} and Fig. \ref{example1p} shows the plots of the solution for 
$\{{\mathrm{v}}^{{t}}_{{\mathtt{x}},{\mathtt{N}}}\}_{{\mathtt{x}}\in\mathbb{Z}_{{\mathtt{M}}}^{-1}}\text{ and }\{{\mathrm{v}}^{{t}}_{{\mathtt{x}},{\mathtt{N}}}\}_{{\mathtt{x}}\in\mathbb{Z}_0^{{\mathtt{M}}-1}}$ when ${\mathtt{M}}>0$ and ${\mathtt{M}}<0$, respectively, in case of crack scattering problem. Thus, the expression \eqref{exactvmsol} has been verified using numerical solution of the discrete Helmholtz equation.
}
{As an illustration of the numerical evaluation of the exact solution for the case of rigid constraints, Fig. \ref{example2} and Fig. \ref{example2p} show the solution for 
$\{{\mathrm{w}}^{{t}}_{{\mathtt{x}},{\mathtt{N}}}\}_{{\mathtt{x}}\in\mathbb{Z}_{{\mathtt{M}}}^{-1}}\text{ and }\{{\mathrm{w}}^{{t}}_{{\mathtt{x}},{\mathtt{N}}}\}_{{\mathtt{x}}\in\mathbb{Z}_0^{{\mathtt{M}}-1}}$ when ${\mathtt{M}}<0$ and ${\mathtt{M}}>0$, respectively.}
{The width of the annulus ${{\mathscr{A}}}$ (described in Appendix \ref{applyFT} and schematically illustrated in Fig. \ref{schematicannulus_2}) where the Wiener--Hopf problem is posed depends on the imaginary part of $\upomega$. In particular, the numerical accuracy of contribution of terms containing the result of scalar Wiener--Hopf factorization has been found to depend on it, specially when ${\mathtt{M}}$ or ${\mathtt{N}}$ are large compared to $1$.}

\subsection{Effect of stagger on the scattered wave field}
\label{briefperturb}

{Taking forward the statement of the Remark \ref{perturbcrack2},
using the expression of ${\su}_{0}^{{\mathrm{F}}}$ and ${\su}_{{\mathtt{N}}-1}^{{\mathrm{F}}}$ in \eqref{v0NAbraeqn},
the effect of stagger is given by 
\begin{equation}\begin{split}
\begin{bmatrix}\pertbn{\su}_{0}^{{\mathrm{F}}}\\\pertbn{\su}_{{\mathtt{N}}-1}^{{\mathrm{F}}}\end{bmatrix}&=\frac{1-{{\lambda}}^{-1}}{({{\mathscr{F}}}_1-{{\mathtt{H}}}-{{\lambda}}^{-1})^2- {{\mathscr{F}}}_{{{\mathtt{N}}}-2}^2}\begin{bmatrix}({{\mathscr{F}}}_1-{{\mathtt{H}}} -{{\lambda}}^{-1})&{{{\mathscr{F}}}_{{{\mathtt{N}}}-2}}\\-{{{\mathscr{F}}}_{{{\mathtt{N}}}-2}}&-({{\mathscr{F}}}_1-{{\mathtt{H}}} -{{\lambda}}^{-1})\end{bmatrix}\begin{bmatrix}\pertbn{\mathrm{v}}_{0}^{{\mathrm{F}}}\\{\pertbn{\mathrm{v}}}_{{\mathtt{N}}}^{{\mathrm{F}}}\end{bmatrix},\label{v0NAbraeqn2}
\end{split}\end{equation}
using \eqref{vhatK}. 
Note that $\boldsymbol{c}^{\redc\pm}$ in \eqref{vhatK} includes the contribution of \eqref{exactvmsol} as expected.
The expressions of
$\pertbn{\su}_{-1}^{{\mathrm{F}}}$ and $\pertbn{\su}_{{\mathtt{N}}}^{{\mathrm{F}}}$
corresponding to \eqref{v0NAbraeqn2} are given by \eqref{un1uNF}, i.e., $\pertbn{\su}_{-1}^{{\mathrm{F}}}=\pertbn{\su}_{0}^{{\mathrm{F}}}-\pertbn{\mathrm{v}}_{0}^{{\mathrm{F}}},$ $\pertbn{\su}_{{\mathtt{N}}}^{{\mathrm{F}}}=\pertbn{\su}_{{{\mathtt{N}}}-1}^{{\mathrm{F}}}-\pertbn{\mathrm{v}}^{{\mathrm{F}}}_{{\mathtt{N}}}.$
These functions determine the solution for scattered wave field everywhere through \eqref{ubulkK}.
}

{Analogously, in the background of the Remark \ref{u0uNFrem} and Remark \ref{perturbrigid2},
using \eqref{whatK}, it is convenient to split the expression of ${\su}_0^{{\mathrm{F}}}$ and ${\su}_{\mathtt{N}}^{{\mathrm{F}}}$ so that the effect of stagger is captured by (using \eqref{Wauxperturb} for the second expression)
\begin{equation}\begin{split}
\pertbn{\su}_0^{{\mathrm{F}}}
=\recip{\mathtt{Q}}\pertbn{\mathrm{w}}_0^{-},\quad
\pertbn{\su}_{{{\mathtt{N}}}}^{{\mathrm{F}}}
=\recip{{\mathtt{Q}}}(\pertbn{\mathrm{w}}_{\mathtt{N}}^{-}+\pertbn{\saux}_{{\mathtt{N}}}+{{\mathtt{f}}}^{{\mathrm{inc}}{\sgnM}}+{{\mathtt{f}}}^{{\sgnM}}).
\end{split}\end{equation}
These functions determine the solution for scattered wave field everywhere through \eqref{ubulkC}.
}

{
Eventually, the counterpart of \eqref{ubulkK} provides the Fourier transform of $\pertbn{\su}_{{\mathtt{x}},{\mathtt{y}}}$ everywhere in both instances of defects as well as for both signs of the offset. The inverse Fourier transform \eqref{uFTinv} based expression also admits the far-field approximation via the method of stationary phase following \cite{Bls0,Bls1,GMthesis,Bls8staggerpair_asymp} to allow interpretation relative to the effect on scattered field far from the cracks. 
The details are omitted and shall be presented elsewhere for the purpose of complete analysis of problem relative to the scattering parameters in the context of far-field and near-tip field behaviour.}

\subsection{Wave field on a portion of both edges}

\begin{figure}[h!]
\centering
\includegraphics[width=.7\textwidth]{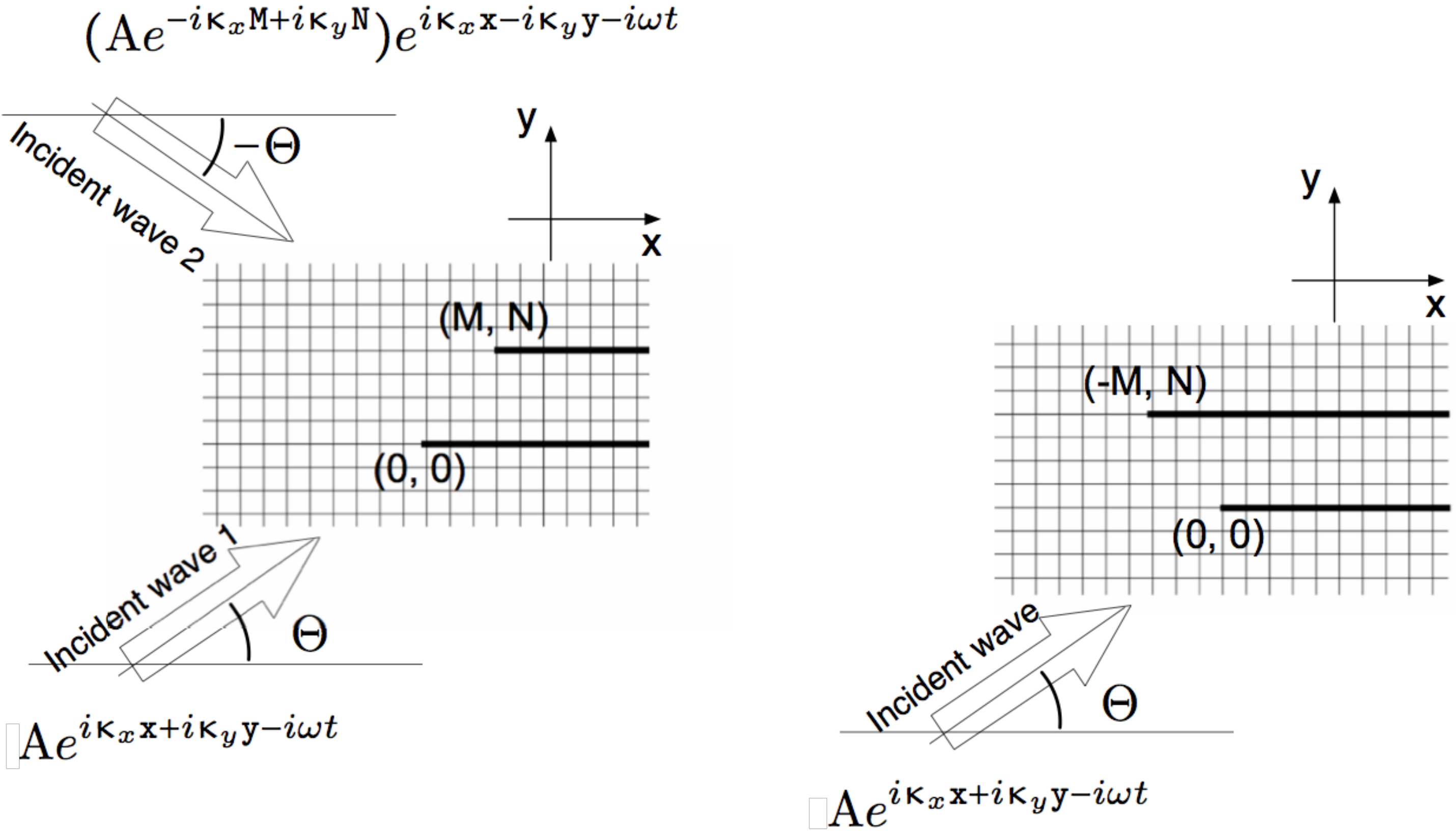}
\caption{{Schematic of the relation between different signs of stagger, i.e., positive and negative offset. The case of positive offset (with $\mathtt{M}>0$) on the left can be used to analyze the one on the right with negative offset as well.}}
\label{flipschematic}
\end{figure}

{Taking forward the last statement in the introduction to this paper (in the context of Fig. \ref{scalingschematic_2}(b)), it is worth pointing out some details concerning it.
Indeed, it is easily seen by flipping the lattice structure vertically that the mapping can be found by changing the angle of incidence and the phase of the incident wave appropriately, as shown schematically in Fig. \ref{flipschematic}; in other words, the analysis of the case with positive offset (here, the offset is denoted by ${\mathtt{M}}$) for an entire range of incidence angle from $-\pi$ to $\pi$ is sufficient to handle negative offset as well. On the other hand, with the analysis of both signs of offsets at disposal, a simultaneous evaluation of the relevant wave field is possible as hinted in Fig. \ref{scalingschematic_2}(b) earlier. This is stated in the form of claims as follows.}

{
Suppose that the list of $|\mathtt{M}|$ number of components of $\mathtt{v}_{\mathtt{N}}$ that satisfies 
\eqref{exactvmsol} for $\mathtt{M}>0$ and \eqref{exactvmsolN} for $\mathtt{M}<0$, and the corresponding scattered field $\mathtt{u}_{\mathtt{x},\mathtt{y}}$ is denoted by
$\mathfrak{V}(\mathtt{M},\mathtt{N},{\mathrm{A}},{\upkappa}_x,{\upkappa}_y)$
and
$\mathfrak{U}_{\mathtt{x},\mathtt{y}}(\mathtt{M},\mathtt{N},{\mathrm{A}},{\upkappa}_x,{\upkappa}_y)$, respectively.
The relevant claims are stated below without proof.
\begin{claim}
For $\mathtt{M}>0,$
$\mathfrak{V}_{-\mathtt{x}-1,{0}}(-\mathtt{M},\mathtt{N},{\mathrm{A}},{\upkappa}_x,{\upkappa}_y)=\mathfrak{U}_{\mathtt{x},1}(\mathtt{M},\mathtt{N},{{\mathrm{A}}}e^{-i{\upkappa}_x{\mathtt{M}}+i{\upkappa}_y{\mathtt{N}}},{\upkappa}_x,-{\upkappa}_y)
+\mathfrak{U}_{\mathtt{x},-1}(\mathtt{M},\mathtt{N},{{\mathrm{A}}}e^{-i{\upkappa}_x{\mathtt{M}}+i{\upkappa}_y{\mathtt{N}}},{\upkappa}_x,-{\upkappa}_y), \mathtt{x}\in\mathbb{Z}_0^{\mathtt{M}-1}.$
\label{claim21}
\end{claim}
\begin{claim}
For $\mathtt{M}<0,$
$\mathfrak{V}_{\mathtt{x}+1,{0}}(-\mathtt{M},\mathtt{N},{\mathrm{A}},{\upkappa}_x,{\upkappa}_y)=\mathfrak{U}_{\mathtt{x},1}(\mathtt{M},\mathtt{N},{{\mathrm{A}}}e^{-i{\upkappa}_x{\mathtt{M}}+i{\upkappa}_y{\mathtt{N}}},{\upkappa}_x,-{\upkappa}_y)
+\mathfrak{U}_{\mathtt{x},-1}(\mathtt{M},\mathtt{N},{{\mathrm{A}}}e^{-i{\upkappa}_x{\mathtt{M}}+i{\upkappa}_y{\mathtt{N}}},{\upkappa}_x,-{\upkappa}_y), \mathtt{x}\in\mathbb{Z}_{\mathtt{M}}^{-1}.$
\label{claim21p}
\end{claim}
Suppose that $\mathtt{w}_{\mathtt{N}}$ that satisfies 
\eqref{exactwmsol} for $\mathtt{M}>0$ and its counterpart for $\mathtt{M}<0$, and the corresponding scattered field $\mathtt{u}_{\mathtt{x},\mathtt{y}}$ is denoted by
$\mathfrak{W}(\mathtt{M},\mathtt{N},{\mathrm{A}},{\upkappa}_x,{\upkappa}_y)$
and
$\mathfrak{U}_{\mathtt{x},\mathtt{y}}(\mathtt{M},\mathtt{N},{\mathrm{A}},{\upkappa}_x,{\upkappa}_y)$, respectively.
The following are stated without proof.
\begin{claim}
For $\mathtt{M}>0,$
$\mathfrak{W}_{-\mathtt{x}-1,{0}}(-\mathtt{M},\mathtt{N},{\mathrm{A}},{\upkappa}_x,{\upkappa}_y)=\mathfrak{U}_{\mathtt{x},1}(\mathtt{M},\mathtt{N},{{\mathrm{A}}}e^{-i{\upkappa}_x{\mathtt{M}}+i{\upkappa}_y{\mathtt{N}}},{\upkappa}_x,-{\upkappa}_y)
+\mathfrak{U}_{\mathtt{x},-1}(\mathtt{M},\mathtt{N},{{\mathrm{A}}}e^{-i{\upkappa}_x{\mathtt{M}}+i{\upkappa}_y{\mathtt{N}}},{\upkappa}_x,-{\upkappa}_y), \mathtt{x}\in\mathbb{Z}_0^{\mathtt{M}-1}.$
\label{claim31}
\end{claim}
\begin{claim}
For $\mathtt{M}<0,$
$\mathfrak{W}_{\mathtt{x}+1,{0}}(-\mathtt{M},\mathtt{N},{\mathrm{A}},{\upkappa}_x,{\upkappa}_y)=\mathfrak{U}_{\mathtt{x},1}(\mathtt{M},\mathtt{N},{{\mathrm{A}}}e^{-i{\upkappa}_x{\mathtt{M}}+i{\upkappa}_y{\mathtt{N}}},{\upkappa}_x,-{\upkappa}_y)
+\mathfrak{U}_{\mathtt{x},-1}(\mathtt{M},\mathtt{N},{{\mathrm{A}}}e^{-i{\upkappa}_x{\mathtt{M}}+i{\upkappa}_y{\mathtt{N}}},{\upkappa}_x,-{\upkappa}_y), \mathtt{x}\in\mathbb{Z}_{\mathtt{M}}^{-1}.$
\label{claim31p}
\end{claim}
Above four claims have been verified numerically though the graphical plots have been omitted in the paper. The proof of these claims is possible by the execution of manipulations similar to those leading to the reduced algebraic equations.
}

\section{Concluding remarks}
As an extension of the analysis of a discrete analogue of Sommerfeld diffraction by a semi-infinite crack or a rigid constraint, and the zero-offset case that has been analyzed recently \cite{Bls8pair1}, this paper presents an analysis of the discrete scattering problem associated with a pair of semi-infinite cracks or rigid constraints on square lattice. 
The multiple-diffraction problem considered in the present paper is an analogue of acoustic wave scattering due to parallel semi-infinite screens, where the edges of the screens are offset or 'staggered'. 
${\mathtt{M}}$ is the `horizontal' offset between the edges, while ${\mathtt{N}}$ is the vertical spacing between the edges.
The stagger in the alignment of the defect edges in the discrete framework leads to a matrix {{Wiener--Hopf}} {kernel which belongs to a class that is well-known to} lie outside the solvable cases.
In the corresponding continuum models, {this class of problems are} known to possess certain exponentially growing elements on the diagonal. The paper provides a way to reduce the complexity of the discrete analogue of the same issue as the non-trivial $2\times2$ {{Wiener--Hopf}} kernel factorization is replaced with the inversion of $|{\mathtt{M}}|\times|{\mathtt{M}}|$ coefficient matrix for case of cracks and the inversion of $|{\mathtt{M}}|+2$ coefficient matrix for case of rigid constraints; the elements in the coefficient matrix depend on scalar {{Wiener--Hopf}} {factorization} of {certain} functions.
Some natural generalizations of the problem are under investigation \cite{sharma2019wienerhopf} {as evident by a recent application to the problem of discrete scattering due to a cohesive crack \cite{Bls32}}.
{Further analysis of the results in this paper shall be presented elsewhere for the purpose of exploring the choice of various parameters in the context of far-field and near-tip field behaviour as well as incidence from waveguide involved. }
{It is worth mentioning from a mathematical viewpoint, pending a rigorous investigation, that the numerical evaluation of the analytically reduced problem suggests that the two signs of offset behave differently sometimes in relation to the variations in the imaginary part of the frequency (possibly due to difference in the structure of scalar Wiener--Hopf factors involved that appear as coefficients of the final reduced equations).}

\section*{Acknowledgement}
The support of SERB MATRICS grant MTR/2017/000013 is gratefully acknowledged.
This work has been available free of peer review on the arXiv since 9/2019.
The author sincerely thanks all three anonymous reviewers for their constructive comments and useful suggestions.

\renewcommand*{\bibfont}{\footnotesize}
\printbibliography

\appendix
\section{Application of Fourier Transform}
\label{applyFT}
Let, ${\upkappa}$, the {\em lattice wave number} of incident lattice wave ${\su}^{{\mathrm{inc}}}$, and, ${\Theta}\in(-\pi, \pi]$, the {\em angle of incidence} of ${\su}^{{\mathrm{inc}}}$ be defined by the relations
\begin{equation}\begin{split}
{\upkappa}_x&={\upkappa}\cos{\Theta}, \quad
{\upkappa}_y={\upkappa}\sin{\Theta}, \quad
{\upkappa}={\upkappa}_1+i{\upkappa}_2, \quad
{\upkappa}_1\ge0.
\label{complexk}
\end{split}\end{equation}
As a consequence of \eqref{complexfreq}, ${\upkappa}_2>0$. 
It is stated without proof, reasoning is analogous to the case of single defect \cite{Bls0,Bls1}, that ${\su}_{{\mathtt{y}}}^{{\mathrm{F}}}$ ($={\su}_{{\mathtt{y}}}^{+}+{\su}_{{\mathtt{y}}}^{-}$), as defined by \eqref{discreteFT}), 
is analytic inside the annulus
\begin{equation}\begin{split}
{{\mathscr{A}}}_u{\,:=}\{{{z}}\in\mathbb{C}: {\mathrm{R}}_+< |{{z}}|< {\mathrm{R}}_-\}, 
\label{annAu}
\end{split}\end{equation}
\begin{equation}\begin{split}
\text{where }
{\mathrm{R}}_+=e^{-{\upkappa}_2}, {\mathrm{R}}_-=e^{{\upkappa}_2\cos{\Theta}}.
\label{Rpm}
\end{split}\end{equation}

Based on above discussion, the Fourier transform \eqref{discreteFT} ${\su}_{{\mathtt{y}}}^{{\mathrm{F}}}$ of the sequence $\{{\su}_{{\mathtt{x}}, {\mathtt{y}}}\}_{{\mathtt{x}}\in\mathbb{Z}}$ is well defined for all ${\mathtt{y}}\in\mathbb{Z}.$ Therefore, the discrete Helmholtz equation \eqref{dHelmholtz} is expressed as
\begin{equation}\begin{split}
({{\mathtt{H}}}({{z}})+2){\su}_{{\mathtt{y}}}^{{\mathrm{F}}}({{z}})-({\su}_{{\mathtt{y}}+1}^{{\mathrm{F}}}({{z}})+{\su}_{{\mathtt{y}}-1}^{{\mathrm{F}}}({{z}}))=0, 
\quad\quad{{z}}\in{{\mathscr{A}}}_u, 
\label{maineq}
\end{split}\end{equation}
for all ${\mathtt{y}}\in\mathbb{Z}$ with ${\mathtt{y}}\ne0, -1$. In \eqref{maineq} the complex function ${{\mathtt{H}}}$ is defined by
\begin{equation}\begin{split}
{{\mathtt{H}}}({{z}})&{\,:=}2-{{z}}-{{z}}^{-1}-{\upomega}^2, 
{{z}}\in\mathbb{C}, \text{ and further define }{{\mathtt{R}}}{\,:=}{{\mathtt{H}}}+4.
\label{h2}
\end{split}\end{equation}
Note that both functions, ${{\mathtt{H}}}$ and ${{\mathtt{R}}}$, are analytic on ${{\mathscr{A}}}_u$. 
The zeros of ${{\mathtt{H}}}$ are ${{z}}_{{\mathtt{h}}}$ and $1/{{z}}_{{\mathtt{h}}}$ while the zeros of ${{\mathtt{R}}}$ are ${{z}}_{{\mathtt{r}}}$ and $1/{{z}}_{{\mathtt{r}}}$.
Using an elementary technique of solving a second order difference equation \cite{Levy}, the general solution of \eqref{maineq} is given by the expression
\begin{equation}\begin{split}
{\su}_{{\mathtt{y}}}^{{\mathrm{F}}}({{z}})=P({{z}}){{\lambda}}({{z}})^{{\mathtt{y}}}+Q({{z}}){{\lambda}}({{z}})^{-{\mathtt{y}}}, 
\quad\quad{{z}}\in{{\mathscr{A}}}, 
\label{gensol}
\end{split}\end{equation}
where $P, Q$ are arbitrary analytic functions on ${{\mathscr{A}}}$. The annulus ${{\mathscr{A}}}$ is defined by
\begin{equation}\begin{split}
{{\mathscr{A}}}{\,:=}{{\mathscr{A}}}_u\cap{{\mathscr{A}}}_L, \quad{{\mathscr{A}}}_L{\,:=}\{{{z}}\in\mathbb{C} : {\mathrm{R}}_L< |{{z}}|< {\mathrm{R}}_L^{-1}\}, \quad
{\mathrm{R}}_L{\,:=}\max\{|{{z}}_{{\mathtt{h}}}|, |{{z}}_{{\mathtt{r}}}|\}.\label{annAAL}
\end{split}\end{equation}
In \eqref{gensol}, following \cite{Slepyanbook}, the function ${{\lambda}}$ is defined by
\begin{subequations}\begin{eqnarray}
{{\lambda}}({{z}})&{\,:=}&\frac{{\mathtt{r}}({{z}})-{\mathtt{h}}({{z}})}{{\mathtt{r}}({{z}})+{\mathtt{h}}({{z}})}, 
\quad\quad{{z}}\in\mathbb{C}\setminus{\mathscr{B}}, \label{lam12}\\
\text{where }&&{{\mathtt{h}}}({{z}}){\,:=}\sqrt{{{\mathtt{H}}}({{z}})}, {{\mathtt{r}}}({{z}}){\,:=}\sqrt{{{\mathtt{R}}}({{z}})}, 
\quad\quad{{z}}\in\mathbb{C}\setminus{\mathscr{B}},\label{Lz}
\end{eqnarray}\label{lamL}\end{subequations}
and ${\mathscr{B}}$ denotes the union of branch cuts for ${{\lambda}}$ (see below), borne out of the chosen branch \eqref{branch} 
for ${{\mathtt{h}}}$ and ${{\mathtt{r}}}$ such that 
$|{{\lambda}}({{z}})|\le1, {{z}}\in\mathbb{C}\setminus{\mathscr{B}}.$
Indeed, as discussed in \cite{Slepyanbook}, with ${\upomega}_2>0$, for all ${{z}}\in\mathbb{C}\setminus{\mathscr{B}}$, the conditions
\begin{equation}\begin{split}
-\pi<\arg {{\mathtt{H}}}({{z}}) <\pi, \Re {{\mathtt{h}}}({{z}})>0, \Re {{\mathtt{r}}}({{z}})>0, {\text{\rm sgn}} \Im {{\mathtt{h}}}({{z}})={\text{\rm sgn}} \Im {{\mathtt{r}}}({{z}}),
\label{branch}
\end{split}\end{equation}
are sufficient to conclude that $|{{\mathtt{r}}}({{z}})-{{\mathtt{h}}}({{z}})|<|{{\mathtt{r}}}({{z}})+{{\mathtt{h}}}({{z}})|, 
{{z}}\in\mathbb{C}\setminus{\mathscr{B}}.$ 

\section{{{Wiener--Hopf}} Factorization}
\subsection{General summary}
\label{appWHfacsgen}
As a major step in any application of the {{Wiener--Hopf}} technique \cite{Noble, Wiener}, the multiplicative factorization of the kernel function is required. The multiplicative factorization of a function $f$ is \cite{Noble}
\begin{equation}\begin{split}
f({{z}})=f_{+}({{z}})f_{-}({{z}}), 
 \quad\quad{{z}}\in{{\mathscr{A}}}_L, 
\label{Lfac}
\end{split}\end{equation}
where the factors $f_\pm$ are given by the Cauchy projectors \cite{Mikhlin, Noble}. Indeed,
\begin{equation}\begin{split}
f_\pm ({{z}})=\exp(\pm\frac{1}{2\pi i} \oint_{{{\mathcal{C}}}_{{{z}}}}\frac{\log f(\zeta)}{{{z}}-\zeta}d\zeta), 
 \quad\quad{{z}}\in\mathbb{C}\text{ such that }|{{z}}|\gtrless {\mathrm{R}}_L^{\pm1}, 
\label{Lpm}
\end{split}\end{equation}
where ${{\mathcal{C}}}_{{{z}}}$ is any rectifiable, closed, counterclockwise contour that lies in the annulus of analyticity for $f$, that is ${{\mathscr{A}}}_L$ defined by \eqref{annAAL}. In \eqref{Lpm}, it has been implicitly assumed that $f_\pm({{z}})=f_\mp({{z}}^{-1})$, which makes the representation unique. 
In the multiplicative factorization of $f$, described by \eqref{Lfac} and \eqref{Lpm}, the function $f_{+}$ (resp. $f_{-}$) is analytic, in fact it has neither poles nor zeros, in the exterior (resp. interior) of a disk centered at $0$ in $\mathbb{C}$ with radius ${\mathrm{R}}_L$ (resp. ${\mathrm{R}}_L^{-1}$). This means that $1/{f_{+}}$ (resp. $1/{f_{-}}$) is analytic in the same region as $f_{+}$ (resp. $f_{-}$). 

Clearly, $\delta_{D}^{+}({z})$ is analytic outside the unit disk in $\mathbb{C}$. Note that $|{{z}}_{{P}}|=e^{{\upkappa}_2\cos{\Theta}}={\mathrm{R}}_-$ and the only singularity of $\delta_{D}^{+}({{z}} {{z}}_{{P}}^{-1})$ is a simple pole at ${{z}}={{z}}_{{P}}$, which lies inside the annulus ${{\mathscr{A}}}$ \eqref{annAAL}.

\subsection{Factorization of ${\tilde{\mathbf{D}}}$}
\label{appDfacs}
In {both kinds of defects the common $2\times2$ matrix is $\tilde{\mathbf{D}}({z})$, where}
\begin{equation}\begin{split}
\tilde{\mathbf{D}}({z})=\begin{bmatrix}
1-{{\lambda}}^{{\mathtt{N}}}({z})&0\\
0&1+{{\lambda}}^{{\mathtt{N}}}({z})
\end{bmatrix}=\begin{bmatrix}
{\tilde{\upalpha}}&0\\
0&{\tilde{\upbeta}}
\end{bmatrix},
\end{split}\end{equation}
using the Chebyshev polynomials \cite{Bls9}, it can be shown that 
\begin{equation}\begin{split}
{\tilde{\upalpha}}
=\frac{{{\mathtt{h}}}}{{{\mathtt{r}}}}\frac{2^{{\mathtt{N}}}\prod_{j=1}^{\lfloor \frac{{{\mathtt{N}}}-1}{2}\rfloor}({{\mathtt{H}}}+4\sin^2\frac{j}{{\mathtt{N}}}\pi)}{({{\mathtt{H}}}+4)^{\lfloor \frac{{{\mathtt{N}}}-1}{2}\rfloor}(1+\frac{{{\mathtt{h}}}}{{{\mathtt{r}}}})^{{\mathtt{N}}}},
{\tilde{\upbeta}}=\frac{2^{{\mathtt{N}}}\prod_{j=1}^{\lfloor \frac{{\mathtt{N}}}{2}\rfloor}({{\mathtt{H}}}+4\sin^2\frac{2j-1}{2{\mathtt{N}}}\pi)}{({{\mathtt{H}}}+4)^{\lfloor \frac{{\mathtt{N}}}{2}\rfloor}(1+\frac{{{\mathtt{h}}}}{{{\mathtt{r}}}})^{{\mathtt{N}}}}.
\end{split}\end{equation}
The symbol $\lfloor\cdot\rfloor$ denotes the integer just less than or equal to the argument.
In fact, 
when ${\mathtt{N}}=2N$, ${\lfloor \frac{{{\mathtt{N}}}-1}{2}\rfloor}=N-1$, ${\lfloor \frac{{{\mathtt{N}}}}{2}\rfloor}=N$,
\begin{equation}\begin{split}
{\tilde{\upalpha}}
={\mathtt{r}}{\mathtt{h}}{\lambda}^N{\mathtt{U}}_{N-1}, \quad
{\tilde{\upbeta}}
=2{\lambda}^N{\mathtt{T}}_{N},
\end{split}\end{equation}
whereas for ${\mathtt{N}}=2N-1$, ${\lfloor \frac{{{\mathtt{N}}}-1}{2}\rfloor}=N-1$, ${\lfloor \frac{{{\mathtt{N}}}}{2}\rfloor}=N-1$,
\begin{equation}\begin{split}
{\tilde{\upalpha}}
=\frac{{{\mathtt{h}}}}{{{\mathtt{r}}}}{\lambda}^{N-1}(1+{\lambda}){\mathtt{W}}_{N-1}, \quad
{\tilde{\upbeta}}
={\lambda}^{N-1}(1+{\lambda}){\mathtt{V}}_{N-1}.
\end{split}\end{equation}
Let ${\tilde{\upalpha}}={\tilde{\upalpha}}_{-}{\tilde{\upalpha}}_{+}, {\tilde{\upbeta}}={\tilde{\upbeta}}_{-}{\tilde{\upbeta}}_{+}$. 

Let (in case of any such need, the sign in front of the radical is decided by the condition $|g^{(1)}_j|<1$)
\begin{equation}\begin{split}
g^{(1)}_{j}={\frac{1}{2}}(2+4\sin^2\frac{j}{{\mathtt{N}}}\pi-{\upomega}^2-\sqrt{(2+4\sin^2\frac{j}{{\mathtt{N}}}\pi-{\upomega}^2)^2-4}),
\end{split}\end{equation}
for all $j=1, \dotsc, {\lfloor \frac{{{\mathtt{N}}}-1}{2}\rfloor}.$
Then
\begin{equation}\begin{split}
{\tilde{\upalpha}}
&=\frac{{{\mathtt{h}}}}{{{\mathtt{r}}}}\frac{2^{{\mathtt{N}}}\prod_{j=1}^{\lfloor \frac{{{\mathtt{N}}}-1}{2}\rfloor}\frac{1}{g^{(1)}_{j}}(1-{{z}}^{-1}{g^{(1)}_{j}})(1-{{z}} g^{(1)}_{j})}{({{{z}}_{{\mathtt{r}}}^{-1}}(1-{{z}}^{-1}{{{z}}_{{\mathtt{r}}}})(1-{{z}} {{z}}_{{\mathtt{r}}}))^{\lfloor \frac{{{\mathtt{N}}}-1}{2}\rfloor}(1+\frac{{{\mathtt{h}}}}{{{\mathtt{r}}}})^{{\mathtt{N}}}},
\end{split}\end{equation}
and
\begin{equation}\begin{split}
{\tilde{\upalpha}}_{+}&=\frac{{{\mathtt{h}}}_+}{{{\mathtt{r}}}_+}\frac{2^{{\frac{1}{2}}{\mathtt{N}}}\prod_{j=1}^{\lfloor \frac{{{\mathtt{N}}}-1}{2}\rfloor}{\sqrt{g^{(1)}_{j}}}{{z}}^{-1}(1-{{z}}^{-1}{g^{(1)}_{j}})}{({\sqrt{{{z}}_{{\mathtt{r}}}}}{{z}}^{-1}(1-{{z}}^{-1}{{{z}}_{{\mathtt{r}}}}))^{\lfloor \frac{{{\mathtt{N}}}-1}{2}\rfloor}((1+\frac{{{\mathtt{h}}}}{{{\mathtt{r}}}})^+)^{{\mathtt{N}}}}, \\
{\tilde{\upalpha}}_{-}&=\frac{{{\mathtt{h}}}_-}{{{\mathtt{r}}}_-}\frac{2^{{\frac{1}{2}}{\mathtt{N}}}\prod_{j=1}^{\lfloor \frac{{{\mathtt{N}}}-1}{2}\rfloor}{\sqrt{g^{(1)}_{j}}}^{-1}(1-{g^{(1)}_{j}}{{z}})}{({\sqrt{{{z}}_{{\mathtt{r}}}}}^{-1}(1-{{{z}}_{{\mathtt{r}}}}{{z}}))^{\lfloor \frac{{{\mathtt{N}}}-1}{2}\rfloor}((1+\frac{{{\mathtt{h}}}}{{{\mathtt{r}}}})^-){}^{{\mathtt{N}}}}.
\end{split}\end{equation}
Let (in case of any such need, the sign in front of the radical is decided by the condition $|g^{(2)}_j|<1$)
\begin{equation}\begin{split}
g^{(2)}_{j}={\frac{1}{2}}(2+4\sin^2\frac{2j-1}{2{\mathtt{N}}}\pi-{\upomega}^2-\sqrt{(2+4\sin^2\frac{2j-1}{2{\mathtt{N}}}\pi-{\upomega}^2)^2-4}),
\end{split}\end{equation}
for all $ j=1, \dotsc, {\lfloor \frac{{\mathtt{N}}}{2}\rfloor}.$
Similarly, 
\begin{equation}\begin{split}
{\tilde{\upbeta}}
&=\frac{2^{{\mathtt{N}}}\prod_{j=1}^{\lfloor \frac{{\mathtt{N}}}{2}\rfloor}\frac{1}{g^{(2)}_{j}}(1-{{z}}^{-1}{g^{(2)}_{j}})(1-{{z}} g^{(2)}_{j})}{({{{z}}_{{\mathtt{r}}}}^{-1}(1-{{z}}^{-1}{{{z}}_{{\mathtt{r}}}})(1-{{z}} {{z}}_{{\mathtt{r}}}))^{\lfloor \frac{{\mathtt{N}}}{2}\rfloor}(1+\frac{{{\mathtt{h}}}}{{{\mathtt{r}}}})^{{\mathtt{N}}}},
\end{split}\end{equation}
and
\begin{equation}\begin{split}
{\tilde{\upbeta}}_{+}&=\frac{2^{{\frac{1}{2}}{\mathtt{N}}}\prod_{j=1}^{\lfloor \frac{{\mathtt{N}}}{2}\rfloor}{\sqrt{g^{(2)}_{j}}}^{-1}(1-{{z}}^{-1}{g^{(2)}_{j}})}{({\sqrt{{{z}}_{{\mathtt{r}}}}}^{-1}(1-{{z}}^{-1}{{{z}}_{{\mathtt{r}}}}))^{\lfloor \frac{{\mathtt{N}}}{2}\rfloor}((1+\frac{{{\mathtt{h}}}}{{{\mathtt{r}}}})^+){}^{{\mathtt{N}}}}, \\
{\tilde{\upbeta}}_{-}&=\frac{2^{{\frac{1}{2}}{\mathtt{N}}}\prod_{j=1}^{\lfloor \frac{{\mathtt{N}}}{2}\rfloor}{\sqrt{g^{(2)}_{j}}}^{-1}(1-{g^{(2)}_{j}}{{z}})}{(\frac{1}{\sqrt{{{z}}_{{\mathtt{r}}}}}(1-{{{z}}_{{\mathtt{r}}}}{{z}}))^{\lfloor \frac{{\mathtt{N}}}{2}\rfloor}((1+\frac{{{\mathtt{h}}}}{{{\mathtt{r}}}})^-){}^{{\mathtt{N}}}}.
\end{split}\end{equation}
Also,
$(1+\frac{{{\mathtt{h}}}}{{{\mathtt{r}}}})^{{\mathtt{N}}}=\sum\limits_{k=0}^{{\mathtt{N}}}\binom{{\mathtt{N}}}{k}(\frac{{{\mathtt{h}}}}{{{\mathtt{r}}}})^k
=F+\frac{{{\mathtt{h}}}}{{{\mathtt{r}}}}G,$
with $F$ and $G$ as rational functions. {However, despite all these expansions there is no closed form expression for the multiplicative factors $(1+\frac{{{\mathtt{h}}}}{{{\mathtt{r}}}})_\pm$ of $(1+\frac{{{\mathtt{h}}}}{{{\mathtt{r}}}})$ and they need to be evaluated numerically via contour integrals \eqref{Lpm}}.

\section{Reduction to algebraic equation for two cracks: ${\mathtt{M}}<0$}
\label{AppcrackMN}
\begin{equation}\begin{split}
{\text{Let ${\mathbb{D}}$ denote the set $\mathbb{Z}_{{\mathtt{M}}}^{-1}$.}}
\label{defnDcrack2}
\end{split}\end{equation}
\text{Above is analgous to the definition provided in \eqref{defnDcrack}}.

According to the expressions provided in \eqref{ppiNAbra}, let
\begin{equation}\begin{split}
-\poly{Q}^-({z})={{{\mathtt{f}}}}^{-}({z})+{{{\mathtt{f}}}}^{{\mathrm{inc}} {-}}({z})
=-\sum\limits_{{\mathtt{x}}\in{\mathbb{D}}}{\mathrm{v}}^{{t}}_{{{\mathtt{x}}},{\mathtt{N}}}{z}^{-{\mathtt{x}}},
\label{defQcrack}
\end{split}\end{equation}
and that by \eqref{qPAbra},
$\boldsymbol{p}^{\sgnM}=\boldsymbol{p}^{-}
=\begin{bmatrix}0\\\poly{Q}^-({z})\end{bmatrix}.$
So the second term in \eqref{difficultK} can be written as
\begin{equation}\begin{split}
-{{\mathbf{D}}}_+({z}){\mathbf{J}}\boldsymbol{p}^{-}({z})
=\frac{1}{\sqrt{2}}\begin{bmatrix}
{\upalpha}_+\poly{Q}^-({z})\\
-{\upbeta}_+\poly{Q}^-({z})
\end{bmatrix}
\end{split}\end{equation}
Using the splitting suggested in \eqref{defFplus},
\begin{equation}\begin{split}
{\upalpha}_+\poly{Q}^-&=F_+\poly{Q}^-
=\sum\limits_{{\mathtt{x}}\in{\mathbb{D}}}{\mathrm{v}}^{{t}}_{{{\mathtt{x}}},{\mathtt{N}}}\Phi_{{\mathtt{x}}}=\sum\limits_{{\mathtt{x}}\in{\mathbb{D}}}{\mathrm{v}}^{{t}}_{{{\mathtt{x}}},{\mathtt{N}}}(\Phi_{{\mathtt{x}}}^{+}+\Phi_{{\mathtt{x}}}^{-}),\\
{\upbeta}_+\poly{Q}^-&=G_+\poly{Q}^-
=\sum\limits_{{\mathtt{x}}\in{\mathbb{D}}}{\mathrm{v}}^{{t}}_{{{\mathtt{x}}},{\mathtt{N}}}\Psi_{{\mathtt{x}}}=\sum\limits_{{\mathtt{x}}\in{\mathbb{D}}}{\mathrm{v}}^{{t}}_{{{\mathtt{x}}},{\mathtt{N}}}(\Psi_{{\mathtt{x}}}^{+}+\Psi_{{\mathtt{x}}}^{-}).
\end{split}\end{equation}
Thus, using the definition of $\boldsymbol{c}^{\redc}$ \eqref{ccincp}, its additive factors are given by
\begin{equation}\begin{split}
{\sqrt{2}}\boldsymbol{c}^{\redc}&={\sqrt{2}}\boldsymbol{c}^{\redc}{}^{+}+{\sqrt{2}}\boldsymbol{c}^{\redc}{}^-\\
&=\begin{bmatrix}
\sum\limits_{{\mathtt{x}}\in{\mathbb{D}}}{\mathrm{v}}^{{t}}_{{{\mathtt{x}}},{\mathtt{N}}}\Phi_{{\mathtt{x}}}^{+}\\
-\sum\limits_{{\mathtt{x}}\in{\mathbb{D}}}{\mathrm{v}}^{{t}}_{{{\mathtt{x}}},{\mathtt{N}}}\Psi_{{\mathtt{x}}}^{+}\\
\end{bmatrix}
+\begin{bmatrix}
\sum\limits_{{\mathtt{x}}\in{\mathbb{D}}}{\mathrm{v}}^{{t}}_{{{\mathtt{x}}},{\mathtt{N}}}\Phi_{{\mathtt{x}}}^{-}\\
-\sum\limits_{{\mathtt{x}}\in{\mathbb{D}}}{\mathrm{v}}^{{t}}_{{{\mathtt{x}}},{\mathtt{N}}}\Psi_{{\mathtt{x}}}^{-}
\end{bmatrix}+\begin{bmatrix}
\recip{{\upalpha}}_-&-\recip{{\upalpha}}_-\\
\recip{{\upbeta}}_-&\recip{{\upbeta}}_-
\end{bmatrix}\boldsymbol{p}^-({z}).
\label{cpfacsN}
\end{split}\end{equation}
{Let $\mathfrak{P}_{\mathbb{D}}$ denote the projection of Fourier coefficients of a typical $f^-({z})$ for $|{z}|<{\mathrm{R}}_-$ to the set ${\mathbb{D}}$.
Thus,
\begin{equation}\begin{split}
\mathfrak{P}_{{}\mathbb{D}}(f^-)=\sum\limits_{{\mathtt{x}}\in{{}\mathbb{D}}}f_{{\mathtt{x}}}{z}^{-{\mathtt{x}}}, \forall f^-=\sum\limits_{{\mathtt{x}}\in{\mathbb{Z}^-}}f_{{\mathtt{x}}}{z}^{-{\mathtt{x}}}, |{z}|<{\mathrm{R}}_-.
\label{defdomainNproj}
\end{split}\end{equation}}
Then the exact solution \eqref{WHsolK}${}_1$ yields equation involving the set of first ${\mathtt{M}}$ Fourier coefficients of the second component of $\boldsymbol{{\mathrm{v}}}^-$, i.e.,
\begin{equation}\begin{split}
{\mathfrak{P}_{{}\mathbb{D}}}(\ensuremath{\hat{\mathbf{e}}}_2\cdot\boldsymbol{{\mathrm{v}}}^-)={\mathfrak{P}_{{}\mathbb{D}}}(\ensuremath{\hat{\mathbf{e}}}_2\cdot\mathbf{L}^-\boldsymbol{c}^-)={\mathfrak{P}_{{}\mathbb{D}}}(\ensuremath{\hat{\mathbf{e}}}_2\cdot\mathbf{L}^-(\boldsymbol{c}^{\old}{}^-+\boldsymbol{c}^{\redc}{}^-)).
\label{reducedcrackN}
\end{split}\end{equation}
Above can be written in terms of a $|{\mathtt{M}}|\times|{\mathtt{M}}|$ coefficient matrix for $\{{\mathrm{v}}_{{{\mathtt{x}}},{\mathtt{N}}}\}_{{\mathtt{x}}\in{\mathbb{D}}}$.
Indeed,
${{\mathrm{v}}}_{{\mathtt{N}}}^-=\ensuremath{\hat{\mathbf{e}}}_2\cdot{{\mathbf{J}}^-}{{{\mathbf{D}}}^-}(\boldsymbol{c}^{\old}{}^-+\boldsymbol{c}^{\redc}{}^-)$.
The equation \eqref{reducedcrackN} is {the reduced algebraic problem} for ${\mathtt{M}}<0$.
Using \eqref{cincfacs} and \eqref{cpfacsN}, with 
\begin{equation}\begin{split}
{\boldsymbol{\tau}}^-=\frac{1}{\sqrt{2}}\begin{bmatrix}
\sum\limits_{{\mathtt{x}}\in{\mathbb{D}}}{\mathrm{v}}^{{t}}_{{{\mathtt{x}}},{\mathtt{N}}}\Phi_{{\mathtt{x}}}^{-}\\
-\sum\limits_{{\mathtt{x}}\in{\mathbb{D}}}{\mathrm{v}}^{{t}}_{{{\mathtt{x}}},{\mathtt{N}}}\Psi_{{\mathtt{x}}}^{-}
\end{bmatrix},
\label{taucrackN}
\end{split}\end{equation}
it is found that
\begin{equation}\begin{split}
{\mathbf{L}^-(\boldsymbol{c}^{\old}{}^-+\boldsymbol{c}^{\redc}{}^-)}
&=-(\recip{{\mathbf{J}}}{{{\mathbf{D}}}^-}\recip{{{\mathbf{D}}}}_-({z}_{{P}}){\mathbf{J}}-\mathbf{I})\boldsymbol{q}^{{\mathrm{inc}}}{}^++\recip{{\mathbf{J}}}{{{\mathbf{D}}}^-}{\boldsymbol{\tau}}^-+\boldsymbol{p}^-.
\end{split}\end{equation}
Indeed, using above and expanding and re-arranging the terms in \eqref{reducedcrackN} further,
\begin{equation}\begin{split}
&{\mathfrak{P}_{{}\mathbb{D}}}({{\mathrm{v}}}_{{\mathtt{N}}}^-({z})+{{{\mathtt{f}}}}^-({z})+{{{\mathtt{f}}}}^{{\mathrm{inc}}}{}^-({z}))\\
&={\mathfrak{P}_{{}\mathbb{D}}}({{\frac{1}{2}}}\mathcal{F}^{{\mathrm{inc}}}({z})+{{{\mathtt{f}}}}^{{\mathrm{inc}}}{}^-({z})-{{\frac{1}{2}}}\sum\limits_{{\mathtt{x}}\in{{}\mathbb{D}}}{\mathrm{v}}^{{t}}_{{{\mathtt{x}}},{\mathtt{N}}}\mathcal{A}_{{\mathtt{x}}}({z})+{\{{{\mathrm{v}}}^{{\mathrm{inc}}}_{{\mathtt{N}}}{}^+({z})-{{{\mathtt{f}}}}^{{\mathrm{inc}}}{}^-({z})\}}),
\label{eqnNMbyMpre}
\end{split}\end{equation}
\begin{equation}\begin{split}
\text{where }
\mathcal{A}_{{\mathtt{x}}}({z})&{\,:=}\big(\Phi_{{\mathtt{x}}}^{-}({z}){{\upalpha}_-}({z})+\Psi_{{\mathtt{x}}}^{-}({z}){{\upbeta}_-}({z})\big),
\end{split}\end{equation}
\begin{equation}\begin{split}
\text{and }
\mathcal{F}^{{\mathrm{inc}}}({z})&{\,:=}-\big(- (1-e^{i\upkappa_y{\mathtt{N}}})\recip{{\upalpha}}_-({z}_{{P}}){{\upalpha}_-}({z})\\
&+ (1+e^{i\upkappa_y{\mathtt{N}}})\recip{{\upbeta}}_-({z}_{{P}}){{\upbeta}_-}({z})\big)e^{-i\upkappa_y{\mathtt{N}}}{{\mathrm{v}}}_{{\mathtt{N}}}^{{\mathrm{inc}}}{}^+({z}).
\end{split}\end{equation}
In view of the definitions of ${{{\mathtt{f}}}}^-$ and ${{{\mathtt{f}}}}^{{\mathrm{inc}}}{}^-$ given by \eqref{ppiNAbra}, it is easy to see that above equation {\eqref{eqnNMbyMpre}} leads to \eqref{eqnMbyMcrack}, i.e., 
$\sum\limits_{{\mathtt{x}}\in{{}\mathbb{D}}}{\mathrm{v}}^{{t}}_{{\mathtt{x}},{\mathtt{N}}}\mathfrak{P}_{{}\mathbb{D}}(\mathcal{A}_{{\mathtt{x}}})
=\mathfrak{P}_{{}\mathbb{D}}\big(\mathcal{F}^{{\mathrm{inc}}}\big),$ which yields a $|{\mathtt{M}}|\times|{\mathtt{M}}|$ system of linear algebraic equations for $\{{\mathrm{v}}^{{t}}_{{\mathtt{x}},{\mathtt{N}}}\}_{{\mathtt{x}}\in{{}\mathbb{D}}}$, i.e., 
$\{{\mathrm{v}}_{{\mathtt{x}},{\mathtt{N}}}\}_{{{}\mathbb{D}}}$ since $\{{\mathrm{v}}^{{\mathrm{inc}}}_{{\mathtt{x}},{\mathtt{N}}}\}_{{{}\mathbb{D}}}$ are known {(by simply expanding the terms, it can be shown that $\mathfrak{P}_{{}\mathbb{D}}$ operating on the expression within curly brackets in \eqref{eqnNMbyMpre} results in zero; here recall \eqref{ppiNAbra}${}_2$)}.
Indeed, with the notation $\mathfrak{C}_{{{\mu}}}(p)$ to denote the coefficient of $z^{{{\mu}}}$ for polynomials $p$ of the form $\mathfrak{C}_1{z}+\mathfrak{C}_2{z}^{2}+\dotsc,$ it is easy to see that
\begin{equation}\begin{split}
{\sum_{\nu=1}^{{\mathtt{M}}}}a_{{{\mu}}\nu}\chi_\nu=b_{{\mu}}\quad\quad({{\mu}}=1 \dotsc, {\mathtt{M}})
\label{aknueqnNcrack}
\end{split}\end{equation}
where {(for ${{\mu}}, \nu=1, \dotsc, {\mathtt{M}}$)}
\begin{equation}\begin{split}
\begin{aligned}
a_{{{\mu}}\nu}&=\mathfrak{C}_{{{\mu}}}\big(\mathfrak{P}_{{}\mathbb{D}}(\mathcal{A}_{-\nu})\big),\quad
\chi_\nu&={\mathrm{v}}^{{t}}_{-\nu,{\mathtt{N}}},\quad
b_{{\mu}}&=\mathfrak{C}_{{{\mu}}}(\mathfrak{P}_{{}\mathbb{D}}\big(\mathcal{F}^{{\mathrm{inc}}}\big)).
\end{aligned}
\end{split}\end{equation}
Let ${\recip{a}_{\nu{{\mu}}}}$ denote the components of the inverse of $[a_{{{\mu}}\nu}]_{{{\mu}}, \nu=1, \dotsc, {\mathtt{M}}+2}$. Then
\begin{equation}\begin{split}
{\mathrm{v}}^{{t}}_{{\mathtt{x}},{\mathtt{N}}}={\sum_{{{\mu}}=1}^{{\mathtt{M}}}}{\recip{a}_{(-{\mathtt{x}}){{\mu}}}}\mathfrak{C}_{{{\mu}}}(\mathfrak{P}_{{}\mathbb{D}}\big(\mathcal{F}^{{\mathrm{inc}}}\big)), {\mathtt{x}}\in{{}\mathbb{D}}.
\label{exactvmsolN}
\end{split}\end{equation}
The expression \eqref{exactvmsolN} has been verified using numerical solution of the discrete Helmholtz equation.

\section{Auxiliary details for the equations for two rigid constraints: ${\mathtt{M}}>0$}
\label{ApprigidMPextra}
{
Now,
with
\begin{equation}\begin{split}
\boldsymbol{\tau}^-=\frac{1}{\sqrt{2}}\begin{bmatrix}
\sum\limits_{{\mathtt{x}}\in{\mathbb{D}}}{\mathrm{w}}^{{t}}_{{{\mathtt{x}}},{\mathtt{N}}}\phi_{{\mathtt{x}}}^{-}\\
-\sum\limits_{{\mathtt{x}}\in{\mathbb{D}}}{\mathrm{w}}^{{t}}_{{{\mathtt{x}}},{\mathtt{N}}}\psi_{{\mathtt{x}}}^{-}
\end{bmatrix},
\label{taurigidP2}
\end{split}\end{equation}
\begin{equation}\begin{split}
\boldsymbol{{\mathrm{w}}}^-
&=\recip{\mathbf{J}}{\mathbf{D}}^-(\boldsymbol{c}^{\old}{}^-+\boldsymbol{c}^{\saux}{}^-+\boldsymbol{c}^{\redc}{}^-)\\
&=-\mathbf{L}^-\big({\mathtt{Q}}\recip{{{\mathbf{D}}}}_--{\mathtt{Q}}({z}_{{P}})\recip{{{\mathbf{D}}}}_-({z}_{{P}})+({z}^{-1}-{z}_{{P}}^{-1})\recip{{{\mathbf{D}}}}_-(0)+({z}-{z}_{{P}}){{\mathbf{D}}}_+(\infty)\big)\mathbf{J}\boldsymbol{q}^{{\mathrm{inc}}}{}^+\\&-\mathbf{L}^-(\recip{{{\mathbf{D}}}}_--\recip{{{\mathbf{D}}}}_-(0)){\mathbf{J}}\boldsymbol{p}^{\saux0}-{z}\mathbf{L}^-(\recip{{{\mathbf{D}}}}_--{{\mathbf{D}}}_+(\infty)){\mathbf{J}}\boldsymbol{p}^{\saux1}\\&-\frac{1}{\sqrt{2}}{\su}^{{t}}_{{\mathtt{M}}-1, {\mathtt{N}}}\mathbf{L}^-\begin{bmatrix}\phi_{{\mathtt{M}}}^-\\-\psi_{{\mathtt{M}}}^-\end{bmatrix}+\mathbf{L}^-\boldsymbol{\tau}^-,
\label{wnFTPrigidpre}
\end{split}\end{equation}
which can be simplified to the form
\begin{equation}\begin{split}
\boldsymbol{{\mathrm{w}}}^-&=(-\frac{{z}}{{z}-{z}_{P}}{\mathtt{Q}}\mathbf{I}+\frac{{z}}{{z}-{z}_{P}}{\mathtt{Q}}({z}_{{P}})\recip{\mathbf{J}}{\mathbf{D}}^-\recip{{{\mathbf{D}}}}_-({z}_{{P}})\mathbf{J})\begin{bmatrix}-{\su}^{{\mathrm{inc}}}_{0, 0}\\-{\su}^{{\mathrm{inc}}}_{0, {\mathtt{N}}}\end{bmatrix}\\
&+({z}-{z}_{P}^{-1})\begin{bmatrix}
{\su}^{{\mathrm{inc}}}_{0, 0}\\
{\su}^{{\mathrm{inc}}}_{0, {\mathtt{N}}}
\end{bmatrix}+{\su}^{t}_{-1, 0}({\mathbf{I}}-\recip{\mathbf{J}}{\mathbf{D}}^-\recip{{{\mathbf{D}}}}_-(0){\mathbf{J}})\ensuremath{\hat{\mathbf{e}}}_1\\
&-\frac{1}{\sqrt{2}}{\su}^{{t}}_{{\mathtt{M}}-1, {\mathtt{N}}}
\recip{\mathbf{J}}{\mathbf{D}}^-\begin{bmatrix}\phi_{{\mathtt{M}}}^-\\-\psi_{{\mathtt{M}}}^-\end{bmatrix}+\recip{\mathbf{J}}{\mathbf{D}}^-\boldsymbol{\tau}^-.
\label{wnFTPrigid}
\end{split}\end{equation}
}
{
The expression of ${\su}_{-1, 0}$ \eqref{un10pre} can be simplified further. Following \cite{Bls1}, by deforming the contour ${\mathcal{C}}$ to circular contour of vanishing radius and using the residue calculus, the expression \eqref{un10pre} can be reduced to the condition for ${\su}^{t}_{-1, 0}$ as
\begin{equation}\begin{split}
({\saux}_0+\ensuremath{\hat{\mathbf{e}}}_1\cdot\boldsymbol{{\mathrm{w}}}^-)|_{{z}={z}_q}=0,
\label{condun10}
\end{split}\end{equation}
using the zeros ${z}_q, {z}_q^{-1}$ (with $|{z}_q|<1$) of ${\mathtt{Q}}$ (here, recall \eqref{defW0Q}).
Indeed, in the context of \eqref{un10pre},
\begin{equation}\begin{split}
{\saux}_0+\ensuremath{\hat{\mathbf{e}}}_1\cdot\boldsymbol{{\mathrm{w}}}^-&=-{\su}_{-1, 0}-{{z}} {\su}^{{\mathrm{inc}}}_{0, 0}\\
&+\frac{{z}}{{z}-{z}_{P}}\ensuremath{\hat{\mathbf{e}}}_1\cdot(-{\mathtt{Q}}\mathbf{I}+{\mathtt{Q}}({z}_{{P}})\recip{\mathbf{J}}{\mathbf{D}}^-\recip{{{\mathbf{D}}}}_-({z}_{{P}})\mathbf{J})\begin{bmatrix}
-{\su}^{{\mathrm{inc}}}_{0, 0}\\
-{\su}^{{\mathrm{inc}}}_{0, {\mathtt{N}}}
\end{bmatrix}\\
&+({z}-{z}_{P}^{-1}){\su}^{{\mathrm{inc}}}_{0, 0}+{\su}^{t}_{-1, 0}-{\su}^{t}_{-1, 0}\ensuremath{\hat{\mathbf{e}}}_1\cdot\recip{\mathbf{J}}{\mathbf{D}}^-\recip{{{\mathbf{D}}}}_-(0){\mathbf{J}}\ensuremath{\hat{\mathbf{e}}}_1\\
&-\frac{1}{\sqrt{2}}{\su}^{{t}}_{{\mathtt{M}}-1, {\mathtt{N}}}
\ensuremath{\hat{\mathbf{e}}}_1\cdot\recip{\mathbf{J}}{\mathbf{D}}^-\begin{bmatrix}
\phi_{{\mathtt{M}}}^-\\
-\psi_{{\mathtt{M}}}^-
\end{bmatrix}+\ensuremath{\hat{\mathbf{e}}}_1\cdot\recip{\mathbf{J}}{\mathbf{D}}^-\boldsymbol{\tau}^-.
\label{W0eqw1P}
\end{split}\end{equation}
Hence, \eqref{condun10} gives,
with ${z}={z}_q$ in \eqref{W0eqw1P},
\begin{equation}\begin{split}
{\su}^{t}_{-1, 0}\ensuremath{\hat{\mathbf{e}}}_1\cdot\recip{\mathbf{J}}{\mathbf{D}}^-\recip{{{\mathbf{D}}}}_-(0){\mathbf{J}}\ensuremath{\hat{\mathbf{e}}}_1+\frac{1}{\sqrt{2}}{\su}^{{t}}_{{\mathtt{M}}-1, {\mathtt{N}}}
\ensuremath{\hat{\mathbf{e}}}_1\cdot\recip{\mathbf{J}}{\mathbf{D}}^-\begin{bmatrix}\phi_{{\mathtt{M}}}^-\\-\psi_{{\mathtt{M}}}^-\end{bmatrix}\\
=\frac{{z}}{{z}-{z}_{P}}{\mathtt{Q}}({z}_{{P}})\ensuremath{\hat{\mathbf{e}}}_1\cdot\recip{\mathbf{J}}{\mathbf{D}}^-\recip{{{\mathbf{D}}}}_-({z}_{{P}})\mathbf{J}\begin{bmatrix}
-{\su}^{{\mathrm{inc}}}_{0, 0}\\
-{\su}^{{\mathrm{inc}}}_{0, {\mathtt{N}}}
\end{bmatrix}+\ensuremath{\hat{\mathbf{e}}}_1\cdot\recip{\mathbf{J}}{\mathbf{D}}^-\boldsymbol{\tau}^-,
\end{split}\end{equation}
i.e.,
\begin{equation}\begin{split}
&-\sum\limits_{{\mathtt{x}}\in{\mathbb{D}}}{\mathrm{w}}^{{t}}_{{{\mathtt{x}}},{\mathtt{N}}}
\frac{1}{2}\big(\phi_{{\mathtt{x}}}^{-}({z}_q){{\upalpha}_-}({z}_q)-\psi_{{\mathtt{x}}}^{-}({z}_q){{\upbeta}_-}({z}_q)\big)\\
&+\frac{1}{2}{\su}^{t}_{-1, 0}\big(\recip{{\upalpha}}_-(0){{\upalpha}_-}({z}_q)+\recip{{\upbeta}}_-(0){{\upbeta}_-}({z}_q)\big)\\
&+{\su}^{{t}}_{{\mathtt{M}}-1, {\mathtt{N}}}\frac{1}{2}\big(\phi_{{\mathtt{M}}}^{-}({z}_q){{\upalpha}_-}({z}_q)-\psi_{{\mathtt{M}}}^{-}({z}_q){{\upbeta}_-}({z}_q)\big)\\
&=-\frac{{z}_q}{{z}_q-{z}_{P}}{\frac{1}{2}}{\mathtt{Q}}({z}_{{P}})\big((1-e^{i\upkappa_y{\mathtt{N}}})\recip{{\upalpha}}_-({z}_{{P}}){{\upalpha}_-}({z}_q)\\
&+ (1+e^{i\upkappa_y{\mathtt{N}}})\recip{{\upbeta}}_-({z}_{{P}}){{\upbeta}_-}({z}_q)\big)e^{-i\upkappa_y{\mathtt{N}}}{\su}_{0,{\mathtt{N}}}^{{\mathrm{inc}}}.
\label{un10}
\end{split}\end{equation}
}
{
The expression of ${\su}_{{\mathtt{M}}-1, {\mathtt{N}}}$ \eqref{uMn10Ppre} can be, similarly, simplified further. 
An analogue of \eqref{condun10} also holds in the context of \eqref{uMn10Ppre} where it becomes $({\saux}_{\mathtt{N}}+\ensuremath{\hat{\mathbf{e}}}_2\cdot\boldsymbol{{\mathrm{w}}}^-+\poly{P}^{+})|_{{z}={z}_q}=0$ with
\begin{equation}\begin{split}
{\saux}_{\mathtt{N}}+\ensuremath{\hat{\mathbf{e}}}_2\cdot\boldsymbol{{\mathrm{w}}}^-+\poly{P}^{+}&={z}^{-{\mathtt{M}}}(-{\su}^{t}_{{\mathtt{M}}-1, {\mathtt{N}}})+{\su}^{{\mathrm{inc}}}_{-1, {\mathtt{N}}}-{{z}} {\su}^{{\mathrm{inc}}}_{0, {\mathtt{N}}}+\poly{P}^{+}\\
&+\frac{{z}}{{z}-{z}_{P}}\ensuremath{\hat{\mathbf{e}}}_2\cdot(-{\mathtt{Q}}\mathbf{I}+{\mathtt{Q}}({z}_{{P}})\recip{\mathbf{J}}{\mathbf{D}}^-\recip{{{\mathbf{D}}}}_-({z}_{{P}})\mathbf{J})\begin{bmatrix}
-{\su}^{{\mathrm{inc}}}_{0, 0}\\
-{\su}^{{\mathrm{inc}}}_{0, {\mathtt{N}}}
\end{bmatrix}\\
&+({z}-{z}_{P}^{-1}){\su}^{{\mathrm{inc}}}_{0, {\mathtt{N}}}-{\su}^{t}_{-1, 0}\ensuremath{\hat{\mathbf{e}}}_2\cdot\recip{\mathbf{J}}{\mathbf{D}}^-\recip{{{\mathbf{D}}}}_-(0){\mathbf{J}}\ensuremath{\hat{\mathbf{e}}}_1\\
&-\frac{1}{\sqrt{2}}{\su}^{{t}}_{{\mathtt{M}}-1, {\mathtt{N}}}
\ensuremath{\hat{\mathbf{e}}}_2\cdot\recip{\mathbf{J}}{\mathbf{D}}^-\begin{bmatrix}
\phi_{{\mathtt{M}}}^-\\
-\psi_{{\mathtt{M}}}^-
\end{bmatrix}+\ensuremath{\hat{\mathbf{e}}}_2\cdot\recip{\mathbf{J}}{\mathbf{D}}^-\boldsymbol{\tau}^-.
\end{split}\end{equation}
Therefore, $({\saux}_{\mathtt{N}}+\ensuremath{\hat{\mathbf{e}}}_2\cdot\boldsymbol{{\mathrm{w}}}^-+\poly{P}^{+})|_{{z}={z}_q}=0$ leads to the simplification of the condition \eqref{uMn10Ppre} for ${\su}_{{\mathtt{M}}-1, {\mathtt{N}}}$ as
\begin{equation}\begin{split}
{\su}^{t}_{-1, 0}\ensuremath{\hat{\mathbf{e}}}_2\cdot\recip{\mathbf{J}}{\mathbf{D}}^-\recip{{{\mathbf{D}}}}_-(0){\mathbf{J}}\ensuremath{\hat{\mathbf{e}}}_1+{\su}^{{t}}_{{\mathtt{M}}-1, {\mathtt{N}}}({z}^{-{\mathtt{M}}}+\frac{1}{\sqrt{2}}
\ensuremath{\hat{\mathbf{e}}}_2\cdot\recip{\mathbf{J}}{\mathbf{D}}^-\begin{bmatrix}
\phi_{{\mathtt{M}}}^-\\
-\psi_{{\mathtt{M}}}^-
\end{bmatrix})\\
=\frac{{z}}{{z}-{z}_{P}}{\mathtt{Q}}({z}_{{P}})\ensuremath{\hat{\mathbf{e}}}_2\cdot\recip{\mathbf{J}}{\mathbf{D}}^-\recip{{{\mathbf{D}}}}_-({z}_{{P}})\mathbf{J}\begin{bmatrix}
-{\su}^{{\mathrm{inc}}}_{0, 0}\\
-{\su}^{{\mathrm{inc}}}_{0, {\mathtt{N}}}
\end{bmatrix}+\ensuremath{\hat{\mathbf{e}}}_2\cdot\recip{\mathbf{J}}{\mathbf{D}}^-\boldsymbol{\tau}^-+\poly{P}^{+},
\end{split}\end{equation}
i.e.,
\begin{equation}\begin{split}
&-\sum\limits_{{\mathtt{x}}\in{\mathbb{D}}}{\mathrm{w}}^{{t}}_{{{\mathtt{x}}},{\mathtt{N}}}
({z}_q^{-{\mathtt{x}}}-\frac{1}{2}\big(\phi_{{\mathtt{x}}}^{-}({z}_q){{\upalpha}_-}({z}_q)+\psi_{{\mathtt{x}}}^{-}({z}_q){{\upbeta}_-}({z}_q)\big))\\
&+\frac{1}{2}{\su}^{t}_{-1, 0}\big(-\recip{{\upalpha}}_-(0){{\upalpha}_-}({z}_q)+\recip{{\upbeta}}_-(0){{\upbeta}_-}({z}_q)\big)\\
&+{\su}^{{t}}_{{\mathtt{M}}-1, {\mathtt{N}}}({z}_q^{-{\mathtt{M}}}-\frac{1}{2}\big(\phi_{{\mathtt{M}}}^{-}({z}_q){{\upalpha}_-}({z}_q)+\psi_{{\mathtt{M}}}^{-}({z}_q){{\upbeta}_-}({z}_q)\big))\\
&=-\frac{{z}_q}{{z}_q-{z}_{P}}{\frac{1}{2}}{\mathtt{Q}}({z}_{{P}})\big(- (1-e^{i\upkappa_y{\mathtt{N}}})\recip{{\upalpha}}_-({z}_{{P}}){{\upalpha}_-}({z}_q)\\
&+ (1+e^{i\upkappa_y{\mathtt{N}}})\recip{{\upbeta}}_-({z}_{{P}}){{\upbeta}_-}({z}_q)\big)e^{-i\upkappa_y{\mathtt{N}}}{\su}_{0,{\mathtt{N}}}^{{\mathrm{inc}}}.
\label{uMn10P}
\end{split}\end{equation}
}

\section{Reduction to algebraic equation for two rigid constraints: ${\mathtt{M}}<0$}
\label{ApprigidMN}
{For the scattering due to two rigid constraints with negative offset, the definition of ${\mathbb{D}}$ is same as that stated by \eqref{defnDcrack2}. Following the definition of $\poly{P}^+$ for positive ${\mathtt{M}}$ \eqref{polyPC}, and similar to the case of cracks \eqref{defQcrack}, let}
\begin{equation}\begin{split}
-\poly{Q}^-({z})={{{\mathtt{f}}}}^{-}({z})+{{{\mathtt{f}}}}^{{\mathrm{inc}} {-}}({z})
=-\sum\limits_{{\mathtt{x}}\in{\mathbb{D}}}{\mathrm{w}}^{{t}}_{{{\mathtt{x}}},{\mathtt{N}}}{z}^{-{\mathtt{x}}}.
\label{defQrigid}
\end{split}\end{equation}
{Then, according to the definition of $\boldsymbol{p}^{-}$ in \eqref{qPcAbra}},
$\boldsymbol{p}^{\sgnM}=\boldsymbol{p}^{-}=\begin{bmatrix}
0\\
{-\poly{Q}^-({z})}
\end{bmatrix}.$
{Keeping in mind the expression of $\boldsymbol{c}$ in equation \eqref{cformC} and the need to additively factorize the term involving $\boldsymbol{p}^{-}$ (as well as a part of $\boldsymbol{p}^{\saux}$), it is pertinent to consider the term (also indicated earlier as the second term in \eqref{difficultC}) ${{\mathbf{D}}}_+({z}){\mathbf{J}}\boldsymbol{p}^{-}({z})$, upon expanding which equals}
\begin{equation}\begin{split}
&\frac{1}{\sqrt{2}}\begin{bmatrix}
{\upalpha}_+\poly{Q}^-({z})\\
{-}{\upbeta}_+\poly{Q}^-({z})
\end{bmatrix}
\end{split}\end{equation}
{Using the splitting suggested in \eqref{defFplus} to additively factorize each of the two components in above expression, it is found that}
\begin{equation}\begin{split}
{\upalpha}_+\poly{Q}^-&=F^+\poly{Q}^-
=\sum\limits_{{\mathtt{x}}\in{\mathbb{D}}}{\mathrm{w}}^{{t}}_{{{\mathtt{x}}},{\mathtt{N}}}\Phi_{{\mathtt{x}}}=\sum\limits_{{\mathtt{x}}\in{\mathbb{D}}}{\mathrm{w}}^{{t}}_{{{\mathtt{x}}},{\mathtt{N}}}(\Phi_{{\mathtt{x}}}^{+}+\Phi_{{\mathtt{x}}}^{-}),\\
{\upbeta}_+\poly{Q}^-&=G^+\poly{Q}^-
=\sum\limits_{{\mathtt{x}}\in{\mathbb{D}}}{\mathrm{w}}^{{t}}_{{{\mathtt{x}}},{\mathtt{N}}}\Psi_{{\mathtt{x}}}=\sum\limits_{{\mathtt{x}}\in{\mathbb{D}}}{\mathrm{w}}^{{t}}_{{{\mathtt{x}}},{\mathtt{N}}}(\Psi_{{\mathtt{x}}}^{+}+\Psi_{{\mathtt{x}}}^{-}),
\label{phipsirigidN}
\end{split}\end{equation}
{so that the third term in the expression of $\boldsymbol{c}$ \eqref{cformC}, i.e., $\boldsymbol{c}^{\redc}$, admits the following additive factorization}
\begin{equation}\begin{split}
{\sqrt{2}}\boldsymbol{c}^{\redc}&={\sqrt{2}}\boldsymbol{c}^{\redc}{}^{+}+{\sqrt{2}}\boldsymbol{c}^{\redc}{}^-
=\begin{bmatrix}
\sum\limits_{{\mathtt{x}}\in{\mathbb{D}}}{\mathrm{w}}^{{t}}_{{{\mathtt{x}}},{\mathtt{N}}}\Phi_{{\mathtt{x}}}^{+}\\
{-}\sum\limits_{{\mathtt{x}}\in{\mathbb{D}}}{\mathrm{w}}^{{t}}_{{{\mathtt{x}}},{\mathtt{N}}}\Psi_{{\mathtt{x}}}^{+}\\
\end{bmatrix}
+{{\sqrt{2}}\boldsymbol{\tau}^-}
-{\recip{{{\mathbf{D}}}}_-{\mathbf{J}}}
\boldsymbol{p}^-({z}),
\label{cpfacsNrigid}
\end{split}\end{equation}
{
\begin{equation}\begin{split}
\text{where }\boldsymbol{\tau}^-=\frac{1}{\sqrt{2}}\begin{bmatrix}
\sum\limits_{{\mathtt{x}}\in{\mathbb{D}}}{\mathrm{w}}^{{t}}_{{{\mathtt{x}}},{\mathtt{N}}}\Phi_{{\mathtt{x}}}^{-}\\
-\sum\limits_{{\mathtt{x}}\in{\mathbb{D}}}{\mathrm{w}}^{{t}}_{{{\mathtt{x}}},{\mathtt{N}}}\Psi_{{\mathtt{x}}}^{-}
\end{bmatrix}.
\label{taurigidN1}
\end{split}\end{equation}
}
{In the context of \eqref{cformC} and \eqref{ccincC}, as the only term remaining after additive factorization presented in \eqref{cincfacsC} and \eqref{cpfacsNrigid}, again by the splitting suggested in \eqref{defFplus} and according to \eqref{phipsirigidN},
it is easy to see that (the same can be compared with \eqref{cauxfacs} for analogous factors when ${\mathtt{M}}>0$)
\begin{equation}\begin{split}
\boldsymbol{c}^{\saux}&=\boldsymbol{c}^{\saux-}+\boldsymbol{c}^{\saux+}\\
&=-(\recip{{{\mathbf{D}}}}_--\recip{{{\mathbf{D}}}}_-(0)){\mathbf{J}}\boldsymbol{p}^{\saux0}-{z}(\recip{{{\mathbf{D}}}}_--{{\mathbf{D}}}_+(\infty)){\mathbf{J}}\boldsymbol{p}^{\saux1}\\
&+{z}({{\mathbf{D}}}_+-{{\mathbf{D}}}_+(\infty)){\mathbf{J}}\boldsymbol{p}^{\saux1}+({{\mathbf{D}}}_+-\recip{{{\mathbf{D}}}}_-(0)){\mathbf{J}}\boldsymbol{p}^{\saux0}\\
&-(-{z}^{-{\mathtt{M}}}{\su}^{{t}}_{{\mathtt{M}}-1, {\mathtt{N}}})\recip{{{\mathbf{D}}}}_-{\mathbf{J}}\ensuremath{\hat{\mathbf{e}}}_2
+\frac{1}{\sqrt{2}}{\su}^{{t}}_{{\mathtt{M}}-1, {\mathtt{N}}}\begin{bmatrix}
\Phi_{{\mathtt{M}}}^-+\Phi_{{\mathtt{M}}}^+\\
-\Psi_{{\mathtt{M}}}^--\Psi_{{\mathtt{M}}}^+
\end{bmatrix},
\label{cauxfacsN}
\end{split}\end{equation}
where (similar to \eqref{cauxfacsphiM})
\begin{equation}\begin{split}
\Phi_{{\mathtt{M}}}={{\upalpha}_+}{z}^{-{\mathtt{M}}}&=F^+{z}^{-{\mathtt{M}}}
=\Phi_{{\mathtt{M}}}^{+}+\Phi_{{\mathtt{M}}}^{-},\\
\Psi_{{\mathtt{M}}}={{\upbeta}_+}{z}^{-{\mathtt{M}}}&=G^+{z}^{-{\mathtt{M}}}
=\Psi_{{\mathtt{M}}}^{+}+\Psi_{{\mathtt{M}}}^{-}.
\label{cauxfacsNphiM}
\end{split}\end{equation}
} 
Then, {using the formal solution \eqref{WHsolC} of the Wiener--Hopf equation (and the definition of $\mathfrak{P}_{{}\mathbb{D}}$ \eqref{defdomainNproj})}, the set of first ${\mathtt{M}}$ Fourier coefficients of the second component of $\boldsymbol{{\mathrm{w}}}^-$ yields
\begin{equation}\begin{split}
{\mathfrak{P}_{{}\mathbb{D}}}(\ensuremath{\hat{\mathbf{e}}}_2\cdot\boldsymbol{{\mathrm{w}}}^-)={\mathfrak{P}_{{}\mathbb{D}}}(\ensuremath{\hat{\mathbf{e}}}_2\cdot\mathbf{L}^-\boldsymbol{c}^-)={\mathfrak{P}_{{}\mathbb{D}}}(\ensuremath{\hat{\mathbf{e}}}_2\cdot\mathbf{L}^-(\boldsymbol{c}^{\old}{}^-+\boldsymbol{c}^{\saux}{}^-+\boldsymbol{c}^{\redc}{}^-)),
\label{reducedrigidN}
\end{split}\end{equation}
{which involves} a $|{\mathtt{M}}|\times|{\mathtt{M}}|$ coefficient matrix for $\{{\mathrm{w}}_{{{\mathtt{x}}},{\mathtt{N}}}\}_{{\mathtt{x}}={{\mathtt{M}}}}^{-1}$ {(this equation is the counterpart of \eqref{reducedrigidP}).}

{However, as noted before for the case of the case of positive offset too, there are also} two more unknowns {$\su^{t}_{-1, 0}$ and $\su^{t}_{{\mathtt{M}}-1, {\mathtt{N}}}$, which need to be determined by the expressions of $\su_0^{{\mathrm{F}}}$ and $\su_{{\mathtt{N}}}^{{\mathrm{F}}}$ stated in \eqref{u0expn} and \eqref{uNexpn}, i.e., \eqref{un10pre} and slightly altered form of \eqref{uMn10Ppre}}
\begin{equation}\begin{split}
\su_{{\mathtt{M}}-1, {\mathtt{N}}}&=\frac{1}{2\pi i}\int_{{\mathcal{C}}}(\recip{{\mathtt{Q}}}({\saux}_{{\mathtt{N}}}+\ensuremath{\hat{\mathbf{e}}}_2\cdot\mathbf{L}^-(\boldsymbol{c}^{\old}{}^-+\boldsymbol{c}^{\saux}{}^-+\boldsymbol{c}^{\redc}{}^-))\\
&-\recip{{\mathtt{Q}}}\poly{Q}^{-}-{\su}^{{\mathrm{inc}}}_{0, {\mathtt{N}}}\delta_{D}^{+}({{z}} {z}_{{P}}^{-1})){z}^{{\mathtt{M}}-2}d{z};
\label{uMn10Npre}
\end{split}\end{equation}
{
in combination with this statement, the equation \eqref{reducedrigidN} is {the reduced algebraic problem} for ${\mathtt{M}}<0$ for discrete scattering due to two staggered rigid constraints.}

{
Using the factors of $\boldsymbol{c}^{\old}$ stated in \eqref{cincfacsC}, $\boldsymbol{c}^{\saux}$ in \eqref{cauxfacsN}, and $\boldsymbol{c}^{\redc}$ in \eqref{cpfacsNrigid}, in the context of \eqref{reducedrigidN}, with $\boldsymbol{\tau}^-$ of \eqref{taurigidN1},
it is found that
\begin{equation}\begin{split}
\boldsymbol{{\mathrm{w}}}^-
&=\recip{\mathbf{J}}{\mathbf{D}}^-(\boldsymbol{c}^{\old}{}^-+\boldsymbol{c}^{\saux}{}^-+\boldsymbol{c}^{\redc}{}^-)\\
&=-({\mathtt{Q}}\mathbf{I}-{\mathtt{Q}}({z}_{{P}})\recip{\mathbf{J}}{{{\mathbf{D}}}^-}\recip{{{\mathbf{D}}}}_-({z}_{{P}})\mathbf{J})\boldsymbol{q}^{{\mathrm{inc}}}{}^+\\
&+({z}_{{P}}^{-1}-{z}^{-1})\recip{\mathbf{J}}{{{\mathbf{D}}}^-}\recip{{{\mathbf{D}}}}_-(0)\mathbf{J}\boldsymbol{q}^{{\mathrm{inc}}}{}^++({z}_{{P}}-{z})\recip{\mathbf{J}}{{{\mathbf{D}}}^-}{{\mathbf{D}}}_+(\infty)\mathbf{J}\boldsymbol{q}^{{\mathrm{inc}}}{}^+\\
&-(\mathbf{I}-\recip{\mathbf{J}}{{{\mathbf{D}}}^-}\recip{{{\mathbf{D}}}}_-(0){\mathbf{J}})\boldsymbol{p}^{\saux0}-{z}({\mathbf{I}}-\recip{\mathbf{J}}{{{\mathbf{D}}}^-}{{\mathbf{D}}}_+(\infty){\mathbf{J}})\boldsymbol{p}^{\saux1}\\
&+{\su}^{{t}}_{{\mathtt{M}}-1, {\mathtt{N}}}({z}^{-{\mathtt{M}}}\ensuremath{\hat{\mathbf{e}}}_2+\frac{1}{\sqrt{2}}\recip{\mathbf{J}}{{{\mathbf{D}}}^-}
\begin{bmatrix}
\Phi_{{\mathtt{M}}}^-\\
-\Psi_{{\mathtt{M}}}^-
\end{bmatrix})\\
&+\recip{{\mathbf{J}}}{{{\mathbf{D}}}^-}\boldsymbol{\tau}^--\boldsymbol{p}^-.
\label{wnFTNrigidpre}
\end{split}\end{equation}
Also, as part of above expression in the context of \eqref{reducedrigidN} 
(recall \eqref{relatedtoidentity1inc} and the statement preceding it),
\begin{equation}\begin{split}
-\ensuremath{\hat{\mathbf{e}}}_2\cdot({\mathtt{Q}}\boldsymbol{q}^{{\mathrm{inc}}}{}^++\boldsymbol{p}^{\saux0}+{z}\boldsymbol{p}^{\saux1}+\boldsymbol{p}^-({z}))&={\mathtt{Q}}{\su}^{{\mathrm{inc}}}_{{\mathtt{N}}}{}^{+}-{\su}^{{\mathrm{inc}}}_{-1, {\mathtt{N}}}+{z}{\su}^{{\mathrm{inc}}}_{0, {\mathtt{N}}}-{{\mathtt{f}}}^-({z})-{{\mathtt{f}}}^{{\mathrm{inc}}}{}^-({z})\\
&={{\mathtt{w}}}_{{\mathtt{N}}}^{{\mathrm{inc}}}{}^+({z})-{{\mathtt{f}}}^-({z})-{{\mathtt{f}}}^{{\mathrm{inc}}}{}^-({z}).
\label{relatedtoidentity1incrigid}
\end{split}\end{equation}
Indeed, using above and expanding and re-arranging the terms in \eqref{reducedrigidN} further,
\begin{subequations}
\begin{equation}\begin{split}
&\mathfrak{P}_{{}\mathbb{D}}({{\mathrm{w}}}_{{\mathtt{N}}}^-({z})+{{\mathtt{f}}}^-({z})+{{\mathtt{f}}}^{{\mathrm{inc}}}{}^-({z}))\\
&=\mathfrak{P}_{{}\mathbb{D}}({\frac{1}{2}}{\mathtt{Q}}({z}_{{P}})\mathcal{F}^{{\mathrm{inc}}}({z})+{{\mathtt{f}}}^{{\mathrm{inc}}}{}^-({z})-{\frac{1}{2}}\sum\limits_{{\mathtt{x}}\in{{}\mathbb{D}}}{\mathrm{w}}^{{t}}_{{{\mathtt{x}}},{\mathtt{N}}}\mathcal{A}_{{\mathtt{x}}}({z})\\
&-{\su}^{{t}}_{-1, 0}\mathcal{J}(z)-{\su}^{{t}}_{{\mathtt{M}}-1, {\mathtt{N}}}(\mathcal{K}(z)-{z}^{-{\mathtt{M}}})+\mathcal{G}^{{\mathrm{inc}}}(z)+\{{{\mathtt{w}}}_{{\mathtt{N}}}^{{\mathrm{inc}}}{}^+({z})-{{\mathtt{f}}}^{{\mathrm{inc}}}{}^-({z})\}),
\label{eqnNMbyMrigidpre}
\end{split}\end{equation}
\begin{equation}\begin{split}
\text{where }
\mathcal{A}_{{\mathtt{x}}}({z})&{\,:=}\big(\Phi_{{\mathtt{x}}}^{-}({z}){{\upalpha}_-}({z})+\Psi_{{\mathtt{x}}}^{-}({z}){{\upbeta}_-}({z})\big),
\end{split}\end{equation}
\begin{equation}\begin{split}
\mathcal{F}^{{\mathrm{inc}}}({z})&{\,:=}-\big(- (1-e^{i\upkappa_y{\mathtt{N}}})\recip{{\upalpha}}_-({z}_{{P}}){{\upalpha}_-}({z})\\
&+ (1+e^{i\upkappa_y{\mathtt{N}}})\recip{{\upbeta}}_-({z}_{{P}}){{\upbeta}_-}({z})\big)e^{-i\upkappa_y{\mathtt{N}}}{\su}_{{\mathtt{N}}}^{{\mathrm{inc}}}{}^+({z}),
\end{split}\end{equation}
\begin{equation}\begin{split}
\mathcal{J}(z)&{\,:=}-
\frac{1}{2}\big(-\recip{{\upalpha}}_-({0}){{\upalpha}_-}({z})+\recip{{\upbeta}}_-({0}){{\upbeta}_-}({z})\big),
\label{KrigidN}
\end{split}\end{equation}
\begin{equation}\begin{split}
\mathcal{K}(z)&{\,:=}-
\frac{1}{2}(\Phi_{{\mathtt{M}}}^-{{\upalpha}_-}({z})+\Psi_{{\mathtt{M}}}^-{{\upbeta}_-}({z})),
\label{GrigidN}
\end{split}\end{equation}
\text{and }
\begin{equation}\begin{split}
\mathcal{G}^{{\mathrm{inc}}}(z)
&{\,:=}-\frac{1}{2}\big(-(1-e^{i\upkappa_y{\mathtt{N}}})\recip{{\upalpha}}_-({0}){{\upalpha}_-}({z})\\
&+(1+e^{i\upkappa_y{\mathtt{N}}})\recip{{\upbeta}}_-({0}){{\upbeta}_-}({z})\big)e^{-i\upkappa_y{\mathtt{N}}}\{({z}_{{P}}^{-1}-{z}^{-1}){\su}_{{\mathtt{N}}}^{{\mathrm{inc}}}{}^+({z})-{\su}_{-1,{\mathtt{N}}}^{{\mathrm{inc}}}\}\\
&-\frac{1}{2}\big(-(1-e^{i\upkappa_y{\mathtt{N}}}){{\upalpha}_+}({\infty}){{\upalpha}_-}({z})\\
&+(1+e^{i\upkappa_y{\mathtt{N}}}){{\upbeta}_+}({\infty}){{\upbeta}_-}({z})\big)e^{-i\upkappa_y{\mathtt{N}}}\{({z}_{{P}}-{z}){\su}_{{\mathtt{N}}}^{{\mathrm{inc}}}{}^+({z})+z{\su}_{0,{\mathtt{N}}}^{{\mathrm{inc}}}\}.
\label{GincrigidN}
\end{split}\end{equation}
\label{eqnNMbyMrigidprefull}
\end{subequations}
It can be noticed that the details provided in \eqref{eqnNMbyMrigidprefull} possess the same form as that in \eqref{eqnMbyMrigidprefull}.
It can be also observed that $\mathfrak{P}_{{}\mathbb{D}}({{\mathtt{w}}}_{{\mathtt{N}}}^{{\mathrm{inc}}+}({z})-{{\mathtt{f}}}^{{\mathrm{inc}}}{}^-({z}))=0$ (recall \eqref{defp1incN}) in \eqref{eqnNMbyMrigidpre}, i.e., contribution of the terms in the curly brackets are zero.}
{Moreover, the expression of $\mathcal{G}^{{\mathrm{inc}}}$ is further simplified to be zero in the same way as \eqref{Gincrigid}; see the sentence following \eqref{Gincrigid}.
In fact, alternatively, this can be also obtained as \eqref{wnFTNrigidpre} can be simplified in the same way as \eqref{wnFTPrigidpre}, so that
\begin{equation}\begin{split}
\boldsymbol{{\mathrm{w}}}^-&=(-\frac{{z}}{{z}-{z}_{P}}{\mathtt{Q}}\mathbf{I}+\frac{{z}}{{z}-{z}_{P}}{\mathtt{Q}}({z}_{{P}})\recip{\mathbf{J}}{\mathbf{D}}^-\recip{{{\mathbf{D}}}}_-({z}_{{P}})\mathbf{J})\begin{bmatrix}-{\su}^{{\mathrm{inc}}}_{0, 0}\\-{\su}^{{\mathrm{inc}}}_{0, {\mathtt{N}}}\end{bmatrix}\\
&+({z}-{z}_{P}^{-1})\begin{bmatrix}
{\su}^{{\mathrm{inc}}}_{0, 0}\\
{\su}^{{\mathrm{inc}}}_{0, {\mathtt{N}}}
\end{bmatrix}+{\su}^{t}_{-1, 0}({\mathbf{I}}-\recip{\mathbf{J}}{\mathbf{D}}^-\recip{{{\mathbf{D}}}}_-(0){\mathbf{J}})\ensuremath{\hat{\mathbf{e}}}_1\\
&+{\su}^{{t}}_{{\mathtt{M}}-1, {\mathtt{N}}}({z}^{-{\mathtt{M}}}\ensuremath{\hat{\mathbf{e}}}_2+\frac{1}{\sqrt{2}}\recip{\mathbf{J}}{{{\mathbf{D}}}^-}
\begin{bmatrix}
\Phi_{{\mathtt{M}}}^-\\
-\Psi_{{\mathtt{M}}}^-
\end{bmatrix})+\recip{{\mathbf{J}}}{{{\mathbf{D}}}^-}\boldsymbol{\tau}^--\boldsymbol{p}^-,
\label{wnFTNrigid2}
\end{split}\end{equation}
which has the terms in first two lines exactly same as those in \eqref{wnFTPrigid}.
It is easy to see that above equation \eqref{eqnNMbyMrigidpre} leads to an equation of the same form as \eqref{eqnMbyMrigid} except that ${\mathbb{D}}$ corresponds to \eqref{defnDcrack2}. In particular,
\begin{equation}\begin{split}
\sum\limits_{{\mathtt{x}}\in{{}\mathbb{D}}}{\mathrm{w}}^{{t}}_{{\mathtt{x}},{\mathtt{N}}}\mathfrak{P}_{{}\mathbb{D}}(\mathcal{A}_{{\mathtt{x}}})
&={\mathtt{Q}}({z}_{{P}})\mathfrak{P}_{{}\mathbb{D}}\big(\mathcal{F}^{{\mathrm{inc}}}\big)-2{\su}^{{t}}_{-1, 0}\mathfrak{P}_{{}\mathbb{D}}(\mathcal{J})-2{\su}^{{t}}_{{\mathtt{M}}-1, {\mathtt{N}}}\mathfrak{P}_{{}\mathbb{D}}(\mathcal{K}-{z}^{-{\mathtt{M}}}).
\label{eqnMbyMrigidN}
\end{split}\end{equation}
}

{
{\bf Equation based on evaluation of ${\su}^{{t}}_{-1, 0}$:} In the context of \eqref{un10pre}, the expression of ${\saux}_0+\ensuremath{\hat{\mathbf{e}}}_1\cdot\boldsymbol{{\mathrm{w}}}^-$ is found to be same as \eqref{W0eqw1P} except that $\phi_{{\mathtt{M}}}^-$ and $\psi_{{\mathtt{M}}}^-$ are replaced by $-\Phi_{{\mathtt{M}}}^-$ and $-\Psi_{{\mathtt{M}}}^-$, respectively. Using the zeros ${z}_q, {z}_q^{-1}$ (with $|{z}_q|<1$) of ${\mathtt{Q}}$, hence, $({\saux}_0+\ensuremath{\hat{\mathbf{e}}}_1\cdot\boldsymbol{{\mathrm{w}}}^-)|_{{z}={z}_q}=0$ gives \eqref{un10}.
}

{
{\bf Equation based on evaluation of ${\su}^{{t}}_{{\mathtt{M}}-1, {\mathtt{N}}}$:} In the context of \eqref{uMn10Npre},
\begin{equation}\begin{split}
{\saux}_{\mathtt{N}}+\ensuremath{\hat{\mathbf{e}}}_2\cdot\boldsymbol{{\mathrm{w}}}^--\poly{Q}^{-}&={z}^{-{\mathtt{M}}}(-{\su}^{t}_{{\mathtt{M}}-1, {\mathtt{N}}})+{\su}^{{\mathrm{inc}}}_{-1, {\mathtt{N}}}-{{z}} {\su}^{{\mathrm{inc}}}_{0, {\mathtt{N}}}-\poly{Q}^{-}\\
&+\frac{{z}}{{z}-{z}_{P}}\ensuremath{\hat{\mathbf{e}}}_2\cdot(-{\mathtt{Q}}\mathbf{I}+{\mathtt{Q}}({z}_{{P}})\recip{\mathbf{J}}{\mathbf{D}}^-\recip{{{\mathbf{D}}}}_-({z}_{{P}})\mathbf{J})\begin{bmatrix}
-{\su}^{{\mathrm{inc}}}_{0, 0}\\
-{\su}^{{\mathrm{inc}}}_{0, {\mathtt{N}}}
\end{bmatrix}\\
&+({z}-{z}_{P}^{-1}){\su}^{{\mathrm{inc}}}_{0, {\mathtt{N}}}-{\su}^{t}_{-1, 0}\ensuremath{\hat{\mathbf{e}}}_2\cdot\recip{\mathbf{J}}{\mathbf{D}}^-\recip{{{\mathbf{D}}}}_-(0){\mathbf{J}}\ensuremath{\hat{\mathbf{e}}}_1\\
&+{z}^{-{\mathtt{M}}}{\su}^{t}_{{\mathtt{M}}-1, {\mathtt{N}}}+\frac{1}{\sqrt{2}}{\su}^{{t}}_{{\mathtt{M}}-1, {\mathtt{N}}}
\ensuremath{\hat{\mathbf{e}}}_2\cdot\recip{\mathbf{J}}{\mathbf{D}}^-\begin{bmatrix}
\Phi_{{\mathtt{M}}}^-\\
-\Psi_{{\mathtt{M}}}^-
\end{bmatrix}+\ensuremath{\hat{\mathbf{e}}}_2\cdot\recip{\mathbf{J}}{\mathbf{D}}^-\boldsymbol{\tau}^-+\poly{Q}^{-}\\
&=\frac{{z}}{{z}-{z}_{P}}({\mathtt{Q}}{\su}^{{\mathrm{inc}}}_{0, {\mathtt{N}}}+{\mathtt{Q}}({z}_{{P}})\ensuremath{\hat{\mathbf{e}}}_2\cdot\recip{\mathbf{J}}{\mathbf{D}}^-\recip{{{\mathbf{D}}}}_-({z}_{{P}})\mathbf{J}\begin{bmatrix}
-{\su}^{{\mathrm{inc}}}_{0, 0}\\
-{\su}^{{\mathrm{inc}}}_{0, {\mathtt{N}}}
\end{bmatrix})\\
&-{\su}^{t}_{-1, 0}\ensuremath{\hat{\mathbf{e}}}_2\cdot\recip{\mathbf{J}}{\mathbf{D}}^-\recip{{{\mathbf{D}}}}_-(0){\mathbf{J}}\ensuremath{\hat{\mathbf{e}}}_1\\
&+\frac{1}{\sqrt{2}}{\su}^{{t}}_{{\mathtt{M}}-1, {\mathtt{N}}}
\ensuremath{\hat{\mathbf{e}}}_2\cdot\recip{\mathbf{J}}{\mathbf{D}}^-\begin{bmatrix}
\Phi_{{\mathtt{M}}}^-\\
-\Psi_{{\mathtt{M}}}^-
\end{bmatrix}+\ensuremath{\hat{\mathbf{e}}}_2\cdot\recip{\mathbf{J}}{\mathbf{D}}^-\boldsymbol{\tau}^-,
\end{split}\end{equation}
\text{so that }
\begin{equation}\begin{split}
{\su}_{{{\mathtt{N}}}}^{{\mathrm{F}}}&=\recip{\mathtt{Q}}({{\mathrm{w}}}^-_{{\mathtt{N}}}+{\saux}_{{\mathtt{N}}}-\poly{Q}^{-}-{\mathtt{Q}}{\su}^{{\mathrm{inc}}}_{0, {\mathtt{N}}}\delta_{D}^{+}({{z}} {z}_{{P}}^{-1}))\\
&=\recip{\mathtt{Q}}\bigg(\frac{{z}}{{z}-{z}_{P}}{\mathtt{Q}}({z}_{{P}})\ensuremath{\hat{\mathbf{e}}}_2\cdot\recip{\mathbf{J}}{\mathbf{D}}^-\recip{{{\mathbf{D}}}}_-({z}_{{P}})\mathbf{J}\begin{bmatrix}
-{\su}^{{\mathrm{inc}}}_{0, 0}\\
-{\su}^{{\mathrm{inc}}}_{0, {\mathtt{N}}}
\end{bmatrix}\\
&-{\su}^{t}_{-1, 0}\ensuremath{\hat{\mathbf{e}}}_2\cdot\recip{\mathbf{J}}{\mathbf{D}}^-\recip{{{\mathbf{D}}}}_-(0){\mathbf{J}}\ensuremath{\hat{\mathbf{e}}}_1+\frac{1}{\sqrt{2}}{\su}^{{t}}_{{\mathtt{M}}-1, {\mathtt{N}}}
\ensuremath{\hat{\mathbf{e}}}_2\cdot\recip{\mathbf{J}}{\mathbf{D}}^-\begin{bmatrix}
\Phi_{{\mathtt{M}}}^-\\
-\Psi_{{\mathtt{M}}}^-
\end{bmatrix}+\ensuremath{\hat{\mathbf{e}}}_2\cdot\recip{\mathbf{J}}{\mathbf{D}}^-\boldsymbol{\tau}^-\bigg),
\end{split}\end{equation}
i.e., in terms of the definitions in \eqref{eqnNMbyMrigidprefull},
\begin{equation}\begin{split}
{\su}_{{{\mathtt{N}}}}^{{\mathrm{F}}}&=\recip{\mathtt{Q}}\big({\frac{1}{2}}{\mathtt{Q}}({z}_{{P}})\mathcal{F}^{{\mathrm{inc}}}({z})
-{\frac{1}{2}}\sum\limits_{{\mathtt{x}}\in{\mathbb{D}}}{\mathrm{w}}^{{t}}_{{{\mathtt{x}}},{\mathtt{N}}}\mathcal{A}_{{\mathtt{x}}}({z})\\
&-{\su}^{t}_{-1, 0}\mathcal{J}({z})-{\su}^{{t}}_{{\mathtt{M}}-1, {\mathtt{N}}}(\mathcal{K}({z})-{z}^{-{\mathtt{M}}})\big)-\recip{\mathtt{Q}}{\su}^{{t}}_{{\mathtt{M}}-1, {\mathtt{N}}}{z}^{-{\mathtt{M}}}.
\end{split}\end{equation}
Hence, by virtue of \eqref{eqnMbyMrigidN} (which is part of the set of equations to be solved), it is found that the Fourier expansion of the round bracket term contains none of the powers of ${z}$ in $-\mathbb{D}$. Hence, as ${z}\to0$, ${\su}_{{{\mathtt{N}}}}^{{\mathrm{F}}}{z}^{{\mathtt{M}}-2}\sim{\su}^{{t}}_{{\mathtt{M}}-1, {\mathtt{N}}}{z}^{-1}$ so that the (counter-clockwise) contour integral after deforming it to a circle of vanishing radius evaluates to $({\saux}_{\mathtt{N}}+\ensuremath{\hat{\mathbf{e}}}_2\cdot\boldsymbol{{\mathrm{w}}}^--\poly{Q}^{-}){z}^{{\mathtt{M}}-2}/{\mathtt{Q}}'(z)|_{{z}={z}_q}=0.$
}
{
The final equation is obtained as 
\begin{equation}\begin{split}
&\sum\limits_{{\mathtt{x}}\in{\mathbb{D}}}{\mathrm{w}}^{{t}}_{{{\mathtt{x}}},{\mathtt{N}}}
\frac{1}{2}\mathcal{A}_{{\mathtt{x}}}({z}_q)+{\su}^{t}_{-1, 0}\mathcal{J}({z}_q)+{\su}^{{t}}_{{\mathtt{M}}-1, {\mathtt{N}}}\mathcal{K}({z}_q)={\frac{1}{2}}{\mathtt{Q}}({z}_{{P}})\mathcal{F}^{{\mathrm{inc}}}({z}_q).
\label{uMn10N}
\end{split}\end{equation}
}

{
The remaining symbolism leading to a matrix formulation follows that presented for positive offset, i.e., \eqref{aknueqnPrigidcon}.
}

\end{document}